\documentclass[acmsmall,screen]{acmart}

\setcopyright{none}

\AtBeginDocument{%
  \providecommand\BibTeX{{%
    \normalfont B\kern-0.5em{\scshape i\kern-0.25em b}\kern-0.8em\TeX}}}

\usepackage[normalem]{ulem}
\usepackage{changebar}
\usepackage{hyperref, wrapfig, amsthm, pict2e, booktabs, subcaption,
  subfiles, rotating, pifont, indentfirst, amsmath, amsfonts,
  mathrsfs, adjustbox, mathtools,booktabs, microtype, ebproof,
  enumitem, multirow, multicol, setspace, xcolor, lineno, listings,
  stmaryrd, xspace, stackengine, makecell}

\usepackage{empheq}

\usepackage[tableposition=top]{caption}
\usepackage{tikz-cd}
\usetikzlibrary{decorations.pathmorphing}
\usepackage[ruled,linesnumbered,vlined]{algorithm2e}

\SetCommentSty{mycommfont}
\usepackage{graphicx}
\usepackage[frozencache,cachedir=.]{minted}
\usepackage{tikz}
\newcounter{markeq}
\usepackage{spacingtricks}

\usepackage[T1]{fontenc}

\newcommand{\tyName}{\mbox{\textsf{u}\textsc{Hat}}\xspace}

\usepackage{spacingtricks, commands, refinementtydef}

\def\showappendix{true}
\begin{document}

% \title{Synthesizing Test Controllers from Types: Property-Guided Bug-Finding for Distributed System Models}
\title{Trace-Guided Synthesis of Effectful Test Generators}
\author{Zhe Zhou}
%\authornote{with author1 note}          %% \authornote is optional;
%% can be repeated if necessary
\orcid{0000-0003-3900-7501}             %% \orcid is optional
\affiliation{
  \institution{Purdue University}            %% \institution is required
  \country{USA}                    %% \country is recommended
}
\email{zhou956@purdue.edu}          %% \email is recommended

\author{Ankush Desai}
\orcid{0000-0001-9006-0100}
\affiliation{
  \institution{Snowflake}            %% \institution is required
  \country{USA}                    %% \country is recommended
}
\email{Ankush.desai@snowflake.com}

\author{Benjamin Delaware}
\orcid{0000-0002-1016-6261}
\affiliation{
  \institution{Purdue University}            %% \institution is required
  \country{USA}                    %% \country is recommended
}
\email{bendy@purdue.edu}

\author{Suresh Jagannathan}
\orcid{0000-0001-6871-2424}
\affiliation{
  \institution{Purdue University}            %% \institution is required
  \country{USA}                    %% \country is recommended
}
\email{suresh@cs.purdue.edu}



\begin{abstract}
  Several recently proposed program logics have incorporated notions
  of \emph{underapproximation} into their design, enabling them to
  reason about reachability rather than safety. In this paper, we
  explore how similar ideas can be integrated into an expressive type
  and effect system.  We use the resulting underapproximate type
  specifications to guide the synthesis of test generators that probe
  the behavior of effectful black-box systems. A key novelty of our
  type language is its ability to capture underapproximate behaviors
  of \emph{effectful} operations using symbolic traces that expose
  latent data and control dependencies, constraints that must be
  preserved by the test sequences the generator outputs.  We
  implement this approach in a tool called \name{}, and evaluate it on
  a diverse range of applications by integrating \name{}'s synthesized
  generators into property-based testing frameworks like QCheck and
  model-checking tools like P.  In both settings, the generators
  synthesized by \name{} are significantly more effective than the
  default testing strategy, and are competitive with state-of-the-art,
  handwritten solutions.
\end{abstract}

\maketitle

% \subfile{oopsla25/0-intro}
% \subfile{oopsla25/9-benchmark-MonkeyDB}
% \subfile{oopsla25/10-benchmark-Quickstorm}
% \subfile{oopsla25/11-benchmark-STLCFuzzer}

% \subfile{sections/reviews}
% \subfile{sections/rebuttal}

\section{Introduction}
\label{sec:intro}

Unlike program verifiers, symbolic execution~\cite{CD11,DFLO19},
bounded model checking~\cite{CKL04,DGJ+13,DQ17}, or property-based
testing~\cite{CH00,Quickstrom,BeginnerLuck,HC+24} frameworks
\emph{underapproximate} a program's behavior, with a goal of
generating only true positives (i.e., every reported bug is a real
bug), but admitting false negatives (i.e., not all real bugs may be
reported). To place such tools on a more formal footing, several new
underapproximate program logics have recently been
proposed~\cite{casl,OHP19,LR+22,ZDS23} for reasoning about
\emph{incorrectness} properties of a program, i.e., a characterization
of the conditions under which a particular error state is guaranteed
to be reached.  As one example, O'Hearn defines~\cite{OHP19}
incorrectness triples of the form $[P]\;\texttt{c}\;[Q]$, which assert
that for any post-state of a command \texttt{c} satisfying assertion
$Q$ there \emph{must} exist a pre-state satisfying $P$.  This
interpretation underapproximates \texttt{c}'s behavior because it
requires a witness pre-state to justify every valid post-state.

While these underapproximate axiomatic proof systems have focused on
imperative first-order state-based specifications, other recent
efforts have instead explored notions of incorrectness and
underapproximation using the language of
types~\cite{RW24,WGJ24,CoverageType}. \citet{RW24}, for example, introduces
a two-sided type system that allows for the refutation of typing
formulas, guaranteeing that not only do well-typed programs not go
wrong, but also that ill-typed programs don't
evaluate. \citet{CoverageType}, on the other hand, leverages notions of
underapproximation in the form of \emph{coverage} types to represent
the set of values guaranteed to be produced by a test generator.

These type systems consider compositional underapproximate reasoning
techniques for pure functional programs, and thus cannot be directly
applied to impure functional programs that, e.g., interact with
effectful libraries.  To address this important limitation, this paper
introduces \emph{Underapproximate Hoare Automata Types}, or
\tyName{}s, a type abstraction for \emph{underapproximating effectful
  computations}.  A \tyName{} type judgement has the form:
\vspace{-.4cm} \par\nobreak {
\[\vdash e\, :\, [H]\,t\,[F]\]
}\noindent where $H$ and $F$ qualify type $t$ with pre- ($H$) and post-
($F$) conditions. Both conditions are expressed as symbolic regular
expressions (SREs) that capture the trace of effects performed before
and by $e$, similar to the \textsc{Hat}s proposed by
\citet{ZYDJ24}. Unlike their overapproximate predecessors, however,
\tyName{}s \emph{underapproximate} program behaviors: the judgement
above states that for every trace of effects $\alpha$ accepted by $F$,
there exists a trace accepted by $H$ under which $e$ \emph{must}
produce $\alpha$.  Intuitively, a \tyName{} captures the sufficient
conditions (i.e., $H$) that guarantee a set of effects (i.e., $F$) can
be realized by executing $e$.

To illustrate, consider a program that interacts with a key-value
store that provides $\eff{get}$ and $\eff{put}$ operations. The
following \tyName{} underapproximates the behavior of $\eff{get}$ with
respect to a context (i.e., trace) of previously performed
$\eff{put}$s:
\vspace{-.4cm} \par\nobreak {\small %
\begin{align*}
  \eff{get}: k{:}\Code{Key.t} \sarr
  \effLR{\allA \seqA \msgA{put}{k\;v}\seqA \allA }
  \urt{\Code{Value.t}}{\vnu = v} \effLR{\msgA{get}{k\; v}}
\end{align*}
}\noindent
This \tyName{} asserts that a $\eff{get}$ operation is guaranteed to
produce a trace consisting of a single event, $\msgA{get}{k\; v}$,
provided that some $\eff{put}$ operation binding $k$ to $v$
(i.e., $\allA\seqA\msgA{put}{k\;v}\seqA\allA$) was executed before.  The
underapproximate interpretation given to \tyName{}s encodes a form of
reachability and data dependence between the effects that have
occurred prior to executing an operation, and the effect produced by
the operation --- in this case, the execution of a $\eff{get}$
operation on a key $k$ depends on the execution of a previous
$\eff{put}$ operation on $k$.\footnote{Note that an overapproximate
  interpretation of this judgement is not sensible, as it would
  require that a $\eff{get}$ operation on $k$ be able to
  simultaneously observe every value $v$ previously $\eff{put}$ into
  $k$. This would mean that after performing $\msgA{put}{k_1\;2}$ and
  $\msgA{put}{k_1\;3}$, the result of $\msgA{get}{k_1}$ would be both
  \emph{both} $2$ and $3$, which is obviously impossible.}

In this paper, we show how \tyName{}s can be used to guide a synthesis
procedure for the property-based testing (PBT) domain. PBT is an
increasingly popular framework for automated testing of user-specified
properties. A key component of PBT frameworks are \emph{test
  generators}, nondeterministic functions that supply input values to
the system under test (SUT). To be effective, a generator needs to
produce inputs that are relevant to the property of interest, e.g.,
they should satisfy the preconditions of the SUT.  The need to
manually write effective generators is a well-known pain
point~\cite{HC+24} to wider adoption of PBT, and developing generators
for properties that impose deep structural or semantic constraints on
inputs remains a publication-worthy exercise~\cite{pbt-ifc, PCRH+11,
  FQL24}.

The challenges with writing PBT generators are further exacerbated
when the SUT is effectful and includes opaque components, e.g.,
libraries whose source code may not be available or amenable for
inspection. In this setting, tests are typically comprised of
\emph{sequences} of inputs that induce effects in a particular order
to expose latent data and control dependencies relevant to the
property being checked, indirectly manipulating library-managed state.
Some examples of scenarios where tests are naturally expressed this
way include:
\begin{enumerate}

\item \textbf{Library API invocations:} The SUT is a library
  implementing an effectful abstract data type (ADT) equipped with a
  set of APIs.  A test consists of a sequence of API calls that could
  lead the implementation to an error state (e.g., a violation of a
  representation invariant~\cite{MP+20,ZYDJ24}).

\item \textbf{Serialized inputs:} The SUT is a procedure that
  processes serialized input data whose effects are tokens used to
  reconstruct an AST which is then evaluated. The test generator
  creates (serialized) inputs that produce semantic dependencies among
  nodes in the corresponding AST that could violate a property of
  interest (e.g., evaluation errors due to incorrect treatment of
  variable scoping or binding).

\item \textbf{Database transactions:} The SUT is a database supporting
  transactional read and write operations under various isolation
  levels.  A test consists of a sequence of transaction operations
  whose interleaved order is chosen to expose potential violations of
  isolation or consistency guarantees the database is supposed to
  provide.

\item \textbf{Distributed protocols:} The SUT is a distributed system
  whose components are expected to adhere to protocol invariants.  A
  test is a sequence of message invocations that can lead the system
  to a global state in which these invariants are not maintained.
\end{enumerate}

Despite being drawn from diverse domains, all these examples share the
requirement that an effective test generator needs to be cognizant of
how dependencies between the events in traces are (or are not)
meaningful to manifesting a property violation.  To capture this
information succinctly, we use \tyName{}s as a formal specification
language to precisely describe latent dependencies among effectful
operations, and to ensure that effectful generators adhere to these
specifications by conforming to those dependencies.  We develop a
custom DSL for writing such generators.  For example, using the
specification given above, a well-typed program written in our DSL
would only generate test sequences in which every $\eff{get}$
operation is preceded by at least one $\eff{put}$ on the same key,
while placing no other constraints on the operations that might occur
between the two.  While this use of \tyName{}s enables us to check
whether a manually-written generator explores traces relevant to a
target property, the burden of writing such a generator still remains.
To alleviate this overhead, we connect these two elements by defining
a type-guided program synthesis technique that automatically
constructs DSL generator programs using \tyName{} specifications. The
resulting generators are guaranteed to produce traces that violate a
target property when executed angelically~\cite{BC+10}, and can thus
evince a concrete bug assuming the SUT provides responses conforming to
provided specifications when performing effects.

% Test generators that use \tyName{}s to guide the search for
% interesting inputs thus implement a novel form of angelic
% non-determinism~\cite{BC+10} where underapproximation is used to
% intelligently inform SUT input generation.
\begin{wrapfigure}{l}{.6\textwidth}
  \centering
  \includegraphics[width=0.98\linewidth]{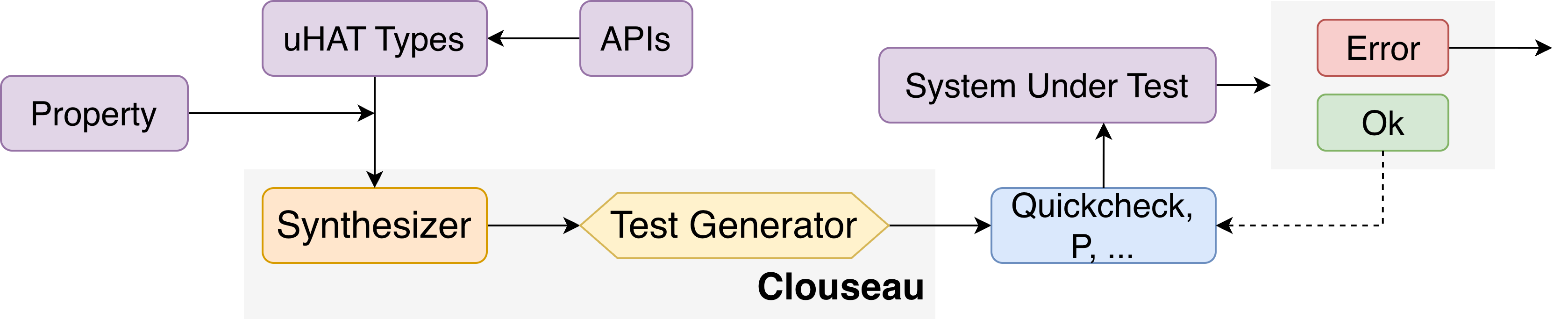}
  \vspace*{-.1in}
  \caption{\name{} workflow}
  \label{fig:workflow}
  \vspace*{-.2in}
\end{wrapfigure}

\autoref{fig:workflow} depicts the workflow of our system, named
\name{}.  The relevant APIs of the SUT are specified using \tyName{}s.
Along with the property of interest, these specifications drive the
synthesis of a test generator in \name{}'s DSL, which is then
transpiled into different testing frameworks; thus far, we have
implemented QCheck~\cite{QCheck} (OCaml) and P~\cite{DGJ+13}
targets.\footnote{For P, the generator replaces the default (random)
  scheduler.}  Synthesized test generator programs are guaranteed to
produce test sequences that respect the underapproximate
interpretation of the \tyName{} API specifications.

This paper makes the following contributions:

\begin{enumerate}

\item We present a new type abstraction, \tyName{}, that ascribes an
  underapproximate interpretation to an effect history trace used to
  justify the effects produced by an expression.

\item We develop a DSL for writing test generators typed using
  \tyName{}s.  The DSL, an extension of a typed $\lambda$-calculus
  with support for (asynchronous) effects, nondeterministic choice, and
  recursion, is sufficiently expressive to be applicable across a diverse
  range of applications,

\item We describe a trace-guided synthesis procedure that uses
  \tyName{} specifications associated with SUT-provided effectful
  operators, along with an SRE specification of the property to be
  tested, to generate well-typed test generators.  \cbnewadding{%
    We evaluate this procedure by comparing the generators synthesized
    by \name{} to state-of-the-art baselines on benchmarks whose bugs
    require deep (effectful) semantic and structural properties of the
    target SUT to surface.}

%% \item \cbremoved{We provide a detailed evaluation study, comparing the
%%   effectiveness of our synthesized generators against state-of-the-art
%%   baselines.  We focus on the synthesis of test generators that
%%   exercise deep (effectful) semantic and structural properties of the
%%   target SUT.}

\end{enumerate}

The rest of the paper is organized as follows.  \autoref{sec:overview}
provides additional motivation and examples.  ~\autoref{sec:formal}
presents the operational semantics and type system for the \name{} DSL.
~\autoref{sec:synthesis} describes the synthesis algorithm. A detailed
evaluation study over a broad range of diverse benchmarks drawn from
the domains given above is given in ~\autoref{sec:impl}; these include
generating tests for semantically complex data structure properties
such as well-typed simply-typed lambda calculus terms using de Bruijn
indices, as well as real-world distributed protocols deployed in
industry.  \autoref{sec:related} presents related work and conclusions.

% \BD{Do we want an explicit set of contributions?}

\section{Overview}\label{sec:overview}

We propose a trace-based type-guided synthesis framework for
generating effectful input test sequences in a PBT system interacting
with a SUT comprised of opaque, potentially concurrent and/or
distributed, components.  Before formalizing our underapproximate type
system and detailing our synthesis algorithm, we first present the key
features of our approach, followed by illustrative examples from
multiple domains.

\subsection{Trace-based Specifications}

We capture the inputs and observed outputs of an effectful operation
(e.g., $\eff{put}$ or $\eff{get}$ in our earlier example) as an
\emph{event} in a \emph{trace}, which is simply a sequence of events.
To reason compositionally about dependencies between the effects
performed by a generator program, we rely on local specifications of
operations in the form of \tyName{}s.  Our type system uses these
specifications to capture the constraints of individual operations in
a modular fashion, delegating the burden of managing the event and
data dependencies of effectful operations from top-level trace
specifications to the underlying type system.

Concretely, we model these local specifications as a collection of
\tyName{}-based signatures like the one for $\eff{get}$ from the
introduction: parameter types constrain the payloads of events, while
return types capture relationships between events. We employ symbolic
regular expressions (SREs) to qualify trace-based properties within a
type, following prior work on trace-based type
systems~\cite{KT+14,ZYDJ24,NUKT+18}. However, our formulation is
distinguished from these other efforts by how it generalizes beyond
history-sensitive specifications of effects or state, additionally
allowing types to express latent dependencies among arbitrary events.

%% The incorrectness using a triple $[P]\;e\;[Q]$ to represent an
%% underapproximation of program $e$, i.e., for all post states in $Q$,
%% there exists a pre state in $P$ such that the program $e$ can reach
%% the post state from the pre state. We adapt this idea into our
%% trace-based setting as follows:
%% \begin{align*}
%%   [\underbrace{\I{presumption}}_{\texttt{state before execution}}]\;\I{code}\;[\underbrace{\I{result}}_{\texttt{state after execution}}]  \;\; \effL\underbrace{\I{history}}_{\texttt{trace before execution}}\effR \;\I{code}\;\effL\underbrace{\I{future}}_{\texttt{new produced trace}} \effR
%% \end{align*}\noindent
%% The triple $\effLR{H}\;e\;\effLR{F}$ consists of a history regex $H$, a code $e$, and a future regex $F$. It means ``For all trace $\alpha$ that is matched by future regex $F$, there exists a history trace $\alpha_h$ matched by history regex $H$ such that the trace $\alpha$ can be produced by program $e$ under the history trace $\alpha_h$.''

\subsection{The \name{} DSL}

To support test generators across diverse application domains, we
introduce a domain-specific language (DSL) that extends a typed lambda
calculus core, with recursion, nondeterministic choice ($\oplus$), and
effectful operations ($\eff{op}$).  This language serves as the basis
for writing (and, in our case, synthesizing) effectful test
generators.
%% Beyond
%% its ability to express generators for effectful sequential libraries
%% (e.g., libraries that implement ADTs, transformers, serializers,
%% etc.), its support for asynchrony also enables it to generate events
%% that separate the initiation of an effect (modeled as a \emph{request}
%% event) from its completion (modeled as a \emph{response} event).  For
%% example, \reqOpNoArgs{read} asynchronously generates a read request to
%% the SUT, returning a tag $\iota$ that can be subsequenty used to
%% observed its result via an \respOp{read}{$\iota$} operation.  This
%% style of trace generation facilitates concise expression of
%% interactive test generators for inherently concurrent and distributed
%% applications (e.g., databases).
The DSL further allows (angelic) random value generation and property checking
via $\mykw{assume}$ and
$\mykw{assert}$ statements. Concretely, an $\mykw{assume}(\phi)$
expression constrains all the program variables at the point at which
it appears to satisfy the formula $\phi$, while $\mykw{assert}(\phi)$
halts program execution if the formula $\phi$ fails to hold.

\subsection{Examples}

\cbnewadding{%
  In order to bias test exploration toward desirable execution paths
  within a huge search space, test developers typically embed domain
  knowledge about the SUT into specialized test generators.  This
  knowledge includes not only \emph{basic correctness properties} SUT
  operations must respect, but also incorporates features of the SUT
  that can help a test generator focus on interesting configurations.
  While similar observations guide the design of prior PBT
  frameworks~\cite{pbt-ifc},\tyName{}s allow this knowledge to be
  stated formally and compositionally. These specifications inform our
  automated synthesis procedure and ensure conformance by the
  generators it synthesizes, without the need to \emph{manually craft
    and tune generation strategies} directly into test generator
  implementations.  The following three examples, drawn from different
  domains, highlight this point. We show both weak \tyName{}s that
  express only basic correctness properties as well as more precise
  specifications that are sufficient to guide the synthesizer to
  produce more effective generators. }

\paragraph{Read Atomicity}
Database transactions allow concurrent transactions to operate on
database state.  To mask latency, client-issued read operations are
typically implemented as asynchronous request and response event
pairs.  It is desirable, however, that read operations be logically
atomic, i.e., the response to a client-issued read request on a key
$k$ always returns the last committed value of $k$ \emph{at the time
  of the request}, regardless of writes that may have been committed
between the request and response.
The following trace captures an
execution that violates read atomicity\footnote{For simplicity, we
  assume the database manages a single key that is elided from event
  arguments, and that write operations are synchronous and commit
  their value to the database immediately.}: %
\vspace{-.4cm} \par\nobreak {\small %
    \begin{align*}
      & \text{\textcolor{orange}{Client 1:}}\ \msgA{write}{3};
      & \msgA{write}{4} & \\
      & \text{\textcolor{orange}{Client 2:}} \quad \quad \quad \quad \quad \quad \msgA{readReq}{\iota_1};\msgA{readRsp}{\iota_1\;3};\msgA{readReq}{\iota_2};& &\textcolor{red}{\msgA{readRsp}{\iota_2\;4}}
    \end{align*}}\noindent
  Here, \textcolor{orange}{Client 2} issues two consecutive asynchronous
  read operations.  A $\eff{readReq}$ event indicates that an asynchronous
  read request effect has been issued by the client on the single key managed
  by the database; the asynchronous response
  from the database is recorded by the $\eff{readRsp}$ event in the trace.  Requests
  are matched to responses using the tag $\iota$.  In this trace, however,
  the second read response witnesses the value of a write performed by another
  transaction that was issued after its matching request.   We are
  interested in surfacing executions like these from \tyName{}
  specifications.

  The following \tyName{} type on $\eff{readRsp}$
  operations admits executions that satisfy read atomicity:
  \vspace{-.4cm} \par\nobreak {\small %
  \begin{align*}
    &\eff{readRsp}:
      \effLR{\Code{LAST}(\msgA{write}{v})\;\seqA\;\msgA{readReq}{\iota}\;\seqA\;
      \allA}\urt{\Code{int}}{\vnu = v}\effLR{\msgA{readRsp}{\iota\;v}}
  \end{align*}}\vspace{-.5cm} \par\nobreak\noindent
This type specifies that the response to a previously issued
asynchronous $\eff{readReq}$ operation can return a value $v$, if the
last committed write prior to the request wrote the value $v$. (Here,
$\iota$ and $v$ are treated as implicitly quantified ghost variables;
we make this notion precise in~\autoref{sec:formal}.)
\cbnewadding{%
  Note that this specification also encodes a simple
  correctness property requiring read responses to correspond to
  a prior read request with the same tag $\iota$ and to reflect a write operation that
  occurred before the request.  A store that enforces read atomicity
  further constrains the write operations visible to a response.}
Because this specification is
interpreted as an underapproximation of the possible behaviors of the
SUT, it establishes a clear dependence between $\eff{write}$
and asynchronous $\eff{readReq}$/$\eff{readRsp}$ operations that should be respected by any
trace which contains these events.  This \tyName{} uses
\Code{LAST}($\msgA{op}{v}$) as an alias for the following regex:
\vspace{-.4cm} \par\nobreak {\small %
  \begin{align}
    \Code{LAST}(\msgA{op}{v}) \equiv \allA\;\seqA\; \msgA{op}{v}\;\seqA\; (\anyA \setminus \msgC{op})^*
  \end{align}
}\vspace{-.5cm} \par\nobreak\noindent %
An implementation of $\eff{readRsp}$ that does not satisfy this type
could yield a trace in which a response to a request yields a
different value than previously seen by the transaction, thus exposing
a read atomicity violation.  This safety property is captured by the
following SRE:\footnote{SREs are implicitly universally quantified
  over any free variables occurring in them.}%
\vspace{-.4cm} \par\nobreak {\small %
  \begin{align}
    &\Code{LAST}(\msgA{write}{v})\; \seqA\; \msgA{readReq}{\iota}\; \seqA\; \allA\seqA\msgA{readRsp}{\iota\;v}\seqA\allA
      \tag{$A_\text{atomic}$}\label{eqn:ReadAtomicity}
  \end{align}
}
\vspace{-.5cm} \par\nobreak\noindent
A violation of this property can be manifested by a test
generator that performs sequences of synchronous $\eff{write}$ and
asynchronous $\eff{read}$ effects in search of an offending trace
similar to the one given above.

Asynchronous effect pairs like $\eff{readReq}$/$\eff{readRsp}$ allow
our DSL to support a form of \emph{interactive} test generation,
instructing the SUT to perform effects via $\eff{opReq}$, and later
observing the results via their $\eff{opResp}$ counterparts.  For
instance, the test generator shown below is automatically
\begin{wrapfigure}{r}{.4\textwidth}
  \vspace{-10pt}
  \begin{minted}[fontsize = \footnotesize, escapeinside=??, xleftmargin=18pt, linenos]{ocaml}
let rec transaction_gen () =
  () ?$\oplus$? let x = int_gen () in
    let _ = ?$\eff{write}$? x in
    let i = ?$\eff{readReq}$? in
    transaction_gen ();
    let y = ?$\eff{readResp}$? i in
    ?$\mykw{assert}$? (x == y)
  \end{minted}
  \vspace{-10pt}
\end{wrapfigure}
synthesized from the above \tyName{} specification and
\ref{eqn:ReadAtomicity}.  The presence of the $\eff{readResp}~i$
operation at line 6 corresponds to the last effect present in
\ref{eqn:ReadAtomicity}.  The \tyName{} associated with this operation
admits arbitrary operations between it and the matching request,
captured by ``$\allA$'' in the specification, following the
$\msgA{readReq}{\iota}$ event in the history trace of $\eff{readRsp}$.
This behavior is subsumed by the recursive call to
\Code{transaction\_gen} at line 5.  The $\eff{readReq}$ effect at line
4 matches the $\msgA{readRsp}{\iota}$ event and the $\eff{write}$
operations on line 3 are synthesized from the $\msgA{write}{v}$ event
that precedes the $\EvtNoArgs{readReq}$ in the history trace; since $v$
is unconstrained in the specification, it is instantiated by a call
to the \Code{int\_gen} primitive generator at line 2.  The
choice operator at line 2 determines if more recursive calls are
required, i.e., if more $\eff{readReq}$ and $\eff{writeResp}$ events
are generated.  The assertion comparing \Code{y} with \Code{x}
determines if a violation (i.e., if the read result is different from
the first written value) occurs.

\cbnewadding{%
  Alternatively, we can provide a weaker specification that encodes
  the more basic correctness property that read responses can only
  return values written before its matching request, making no
  assumptions about \emph{which} write the response is paired with:}
\vspace{-.4cm} \par\nobreak \cbnewadding{{\small %
  \begin{align*}
    &\eff{readRsp}:
      \effLR{\allA\;\seqA\;\msgA{write}{v}\seqA\;\;\allA\;\seqA\;\msgA{readReq}{\iota}\;\seqA\;
      \allA}\urt{\Code{int}}{\vnu = v}\effLR{\msgA{readRsp}{\iota\;v}}
  \end{align*}}}
\vspace{-.5cm} \par\nobreak\noindent
\cbnewadding{This specification only requires that the value $v$ returned
  by the read response is \emph{some} value written before the
  read was requested: it need not be the \emph{last} value written before the corresponding request.
  The following generator is also type-safe with respect to this weaker specification:}
  \begin{minted}[fontsize = \footnotesize, escapeinside=??, xleftmargin=18pt, linenos]{ocaml}
  let x = int_gen () in let _ = ?$\eff{write}$? x in
  let y = int_gen () in let _ = ?$\eff{write}$? y in
  let i = ?$\eff{readReq}$? in let z = ?$\eff{readResp}$? i in
  ?$\mykw{assert}$? (z == x || z == y)
  \end{minted}
  \cbnewadding{This generator issues two writes followed by a read,
    and asserts that the read value must be one of the two previously
    written ones. Notably, although this program is type-safe
    according to the weaker uHAT, it is not allowed by the original
    specification, which requires the read value to be the last
    one. Because this weaker specification does not capture the core
    property of read atomicity, generators synthesized from this
    specification may produce benign test cases that would not surface
    read atomicity violations.  Nonetheless, the generator
    $\Code{transaction\_gen}$ can also be derived from this weaker
    specification.  Thus, an imprecise specification simply means that
    our synthesis algorithm has less guidance on the structure of the
    generators it produces. In the limit, a specification that places
    no constraints on the operations of the system under test results
    in synthesized generators that simply perform random testing. }
\paragraph{Information-Flow Control}
  \citet{pbt-ifc} applies the PBT methodology to a secure
  information-flow control (IFC) abstract machine meant to enforce
  end-to-end noninterference (EENI).  This property captures the
  notion that secret inputs should not influence public outputs: any
  two executions starting initial states that are indistinguishable
  (i.e., states in which all ``low'' security values are the same)
  should terminate in final states that are also indistinguishable.
  \citet{pbt-ifc} encode the IFC machine input as a sequence of stack
  machine commands.  For instance, the following serialized input
  corresponds to two IFC programs that both store value $1$ from stack
  into memory, but at different addresses; the only difference is the
  pointer address which is labeled \textsf{H} and thus not visible to
  an observer: \vspace{-.4cm} \par\nobreak {\small %
  \begin{flalign*}
   & \quad \quad \quad \quad  \zevent{push}{\Code{L}, \langle 1, 1\rangle};\zevent{push}{\Code{H}, \langle 0, 1\rangle};\eff{store} \equiv
  \\
  \text{\textcolor{gray}{Program 1:}} \quad &\zevent{push_1}{\Code{L}, 1};\zevent{push_1}{\Code{H}, 0};\eff{store} \quad \text{\textcolor{gray}{and}} \quad\text{\textcolor{gray}{Program 2:}} \quad \zevent{push_2}{\Code{L}, 1};\zevent{push_2}{\Code{H}, 1};\eff{store}
  \\\text{\textcolor{gray}{Memory 1:}} \quad & [(\Code{L}, 0); (\Code{L}, 0)] \Downarrow \textcolor{red}{[(\Code{L}, 1), (\Code{L}, 0)]} \quad\quad\quad\quad\text{\textcolor{gray}{Memory 2:}} \quad [(\Code{L}, 0); (\Code{L}, 0)] \Downarrow \textcolor{red}{[(\Code{L}, 0), (\Code{L}, 1)]}
\end{flalign*}}\vspace{-.5cm} \par\nobreak\noindent
Here, $\Code{L}$ and $\Code{H}$ are the low (public) and high (secret)
security levels, with $\Code{L} \sqsubseteq \Code{H}$.  To ensure
indistinguishability, public commands (level $\Code{L}$) must use
identical argument values in both programs, whereas secret commands
(level $\Code{H}$) may choose different address values ($0$ and $1$).
Since IFC is a property that holds over pairs of executions, the trace
specification $\zevent{push}{\Code{L}, \langle 1,1\rangle}$ identifies
two programs performing a $\eff{push}$ operation with the same
security label ($\Code{L}$) and value (1); in contrast,
$\zevent{push}{\Code{H}, \langle 0, 1\rangle}$ captures a $\eff{push}$
operation with security label $\Code{H}$ that pushes value 0 in one
execution and 1 in the other.  Although both executions start from the
same initial memory state ([{(\Code{L}, 0), (\Code{L}, 0)}]), the IFC
machine may yield different final states ([{(\Code{L}, 0), (\Code{L},
    1)}] and [{(\Code{L}, 1), (\Code{L}, 0)}]) after performing the
$\eff{store}$ which uses the $\Code{H}$ secret pointer (either 0 or 1)
to store the $\Code{L}$ value (1).  The goal of a PBT generator is to
generate input sequences like these to determine if an IFC machine
implementation will fault prior to the $\eff{store}$ being executed.

\cbnewadding{ To avoid synthesizing generators that produce
  uninteresting test programs, we would like the executions they
  induce to: (1) exercise a noninterference correctness property - two
  input programs must be indistinguishable at the low security level
  (e.g., $\textbf{push}(\texttt{L}, 1)$ and
  $\textbf{push}(\texttt{L}, 2)$ are distinguishable); (2) bias the
  search to avoid programs that expose uninteresting secrets, for
  example by storing a high security value at the same address in both
  executions (e.g., $\zevent{push}{\Code{H}, \langle 1, 1\rangle}$),
  or pushing values to invalid addresses (e.g.,
  $\zevent{push}{\Code{L},\langle -1, -1\rangle}$).  The following
  \tyName{} specification addresses both concerns:}
%% \cbremoved{To avoid generators that could produce uninteresting test programs
%% that cannot violate non-interference (e.g.,
%% $\zevent{push}{\Code{L}, \langle 0, 1 \rangle}$), operate over invalid
%% addresses (e.g., $\zevent{push}{\Code{L},\langle -1, -1\rangle}$), or
%% construct executions with uninteresting secrets (e.g.,
%% $\zevent{push}{\Code{H}, \langle 1, 1\rangle}$), we provide our
%% synthesizer with the following \tyName{} specifications:}%
\vspace{-.4cm} \par\nobreak {\small %
  \begin{align*}
    \eff{push} ~:~
    & \effLR{\allA} \;\Code{unit}\; \effLR{\msgB{push}{\I{lvl}\;x\;y}{\I{lvl} = \Code{L} {\iff} x = y}}
    \\
    \eff{store} ~:~
    & \effLR{\allA\; \seqA\; \msgA{push}{\I{lvl}\;x_1\;y_1\rangle} \seqA\;
      (\anyA \setminus \msgC{store})^* \seqA\; \msgB{push}{\I{lvl}\;x_2\;y_2}{\I{isAddr}(x_2) \land \I{isAddr}(y_2)}}\; \Code{unit} \; \effLR{\eff{store}}
  \end{align*}}\vspace{-.5cm} \par\nobreak\noindent
The post-condition SRE for $\eff{push}$ enforces that its arguments
should only be equal if they are intended to be public, while the
pre-condition SRE of a $\eff{store}$ event describes stacks
containing at least two elements (i.e., there must be at least two
$\eff{push}$ operations with no intervening $\eff{store}$), with the
top element of the stack referencing the memory location to be
updated.
The desired global property that a generator needs to test can be expressed as $\allA\;\seqA\;\eff{store}\seqA\;\allA$ that ensures $\eff{store}$ can be safely performed.
The following test generator generates effect sequences consistent with
this property and the \tyName{} specifications for $\eff{push}$ and $\eff{store}$ given above:
\begin{minted}[fontsize = \footnotesize, escapeinside=??, xleftmargin=18pt, linenos]{ocaml}
  let push_gen () =
    let (x: int) = ?$\mykw{assume}(\mbox{\textit{isAddr}}$? x) in let (y: int) = ?$\mykw{assume}(\mbox{\textit{isAddr}}$? y) in
    (?$\mykw{assume}$? (x = y); ?$\eff{push}$? L ?$\langle$?x, y?$\rangle$?) ?$\oplus$? (?$\mykw{assume}$?(x != y); ?$\eff{push}$? H ?$\langle$?x, y?$\rangle$?)
  in push_gen (); push_gen (); ?$\eff{store}$?;
\end{minted}

\cbnewadding{%
  On the other hand, the following specifications only guarantee that
  tests exercise noninterference, and do not impose any additional
  constraints on the diversity of secret inputs or address validity:%
  {\small %
  \begin{align*}
    \eff{push} ~:~
    & \effLR{\allA} \;\Code{unit}\; \effLR{\msgB{push}{\I{lvl}\;x\;y}{\I{lvl} = \Code{L} \impl x = y}}
    \\
    \eff{store} ~:~
    & \effLR{\allA\; \seqA\; \msgA{push}{\I{lvl}\;x_1\;y_1} \seqA\;
      (\anyA \setminus \msgC{store})^* \seqA\;
      \msgB{push}{\I{lvl}\;x_2\;y_2}{\true}}\; \Code{unit}\;
      \effLR{\eff{store}}
  \end{align*}}\noindent
  This weaker specification can synthesize generators that produce uninteresting secret values (\textbf{push}'s post-condition weakens
  the bi-implication found in the stronger specification given earlier) or invalid addresses (\textbf{store}s no longer constrain
  the arguments supplied to a previously executed \textbf{push} to be addresses); in either case, overall test efficiency is reduced.}

%\end{example}
%\begin{example}[De Bruijn Index]  \label{ex:de-bruijn-index}
\paragraph{STLC Terms}
Our final example applies our approach to the
well-studied~\cite{BeginnerLuck,CoverageType,FQL24} problem of using
PBT to test the correctness of an interpreter for simply-typed lambda
calculus (STLC) terms.  \cbnewadding{%
  The SUT in this example is an STLC interpreter, and the synthesized
  test generator produces a family of \emph{serialized} STLC terms,
  i.e., sequences of string tokens, using a single effectful
  operation, $\eff{token} : \Code{String} \to \Code{unit}$.}  Unlike
these earlier efforts, we consider terms in which variables are
represented by their de Bruijn index~\cite{DeBruijn}. Consider a SUT
that fails to correctly increment the index during substitution under
a lambda abstraction, leading to incorrect variable capture. %
\vspace{-.4cm} \par\nobreak {\small %
  \begin{align*}
    \text{\textsc{Reduction rule}:} \quad & (\lambda t. e)\; v \hookrightarrow e [0 \mapsto v]\\
    \text{\textsc{Incorrect substitution}:} \quad & (\lambda t. e)[n \mapsto v] = \lambda t. (e[\textcolor{red}{n} \mapsto v]) \quad \textcolor{red}{\xmark} \text{ \textcolor{gray}{should be $n + 1$}}
  \end{align*}}\vspace{-.5cm} \par\nobreak\noindent
The following term will trigger a bug in our SUT, as this well-typed
term will, in fact, get stuck: %:
\vspace{-.4cm} \par\nobreak {\small %
    \begin{flalign*}
      & (\lambda \Int. (\lambda (\Int\sarr\Int). [0]\;[1])\; (\lambda \Int. [0]))\; 3 && \tag{$e_\Code{STLC}$}\label{ex:e4} \\
      &\hookrightarrow (\lambda (\Int\sarr\Int). \textcolor{red}{3}\;[1])\; (\lambda \Int. \textcolor{red}{3})  \quad &&\text{\textcolor{red}{(wrong substitution ``$[0] \mapsto 3$'')}}
    \\&\hookrightarrow  3\;[1] \quad &&\text{\textcolor{red}{(stuck!)}}
  \end{flalign*}}\vspace{-.5cm} \par\nobreak\noindent %
Here, our buggy implementation mistakenly substitutes the bound
variable ``$[0]$'' with the argument value $3$ in the two inner
abstractions, instead of ``$[1]$''.

Guiding test generation for interpreters over \emph{serialized} STLC
terms is challenging, as the generated token sequences must contain
\emph{nested function abstractions} with appropriate de Bruijn
indices; otherwise, errors involving incorrect variable substitution
with non-zero bound variable indices cannot be detected. To enforce
the correct use of de Bruijn indices we need to express the relation
between indices and abstraction depth.  To do so, we introduce a
\emph{ghost} event $\effGray{depth}$ that tracks the current
abstraction depth;  as we describe
in the following sections, ghost events carry a payload that enables
establishing additional dependencies with other events beyond those involving
the effects parameters and result that can still be efficiently checked using
our type system.  Here, $\eff{token}$ is an effect that appends its
argument to a serialized string that is eventually passed to the
interpreter being tested.    A \tyName{} specification to track depth can
  now be given as follows:
  \vspace{-.4cm} \par\nobreak {\small %
    \begin{align*}
      \effLR{\Code{LAST}(\msgAGray{depth}{d}) \seqA\;
      \msgA{token}{\boxed{(\lambda t.}\;}\seqA\; \msgAGray{depth}{d +
      1}\seqA\; \allA }\; unit \;\effLR{\msgA{token}{\boxed{)}\;} \seqA\; \msgAGray{depth}{d}}
    \end{align*}}\vspace{-.5cm} \par\nobreak\noindent
  As shown above, upon entering a function abstraction, we retrieve
  the last tagged depth (assumed to be initialized to $0$), increment it at the
  start of the body, and restore it upon exiting. Leveraging
  $\effGray{depth}$, we can specify a correct de Bruijn index
  safety property using the following SRE:
  \vspace{-.4cm} \par\nobreak {\small %
    \begin{align*}
      \Code{LAST}(\msgBGray{depth}{d}{n < d}) \seqA\; \eff{token}\;\boxed{[\widehat{n}]} \seqA \allA
      \tag{$A_\lambda$}\label{eqn:deBruijn+prop}
    \end{align*}}\vspace{-.5cm} \par\nobreak\noindent
  The bound variable index $n$ is required to be less than the
  current depth $d$; the notation $\widehat{n}$ signifies that $n$'s value must
  be substituted in the string output produced by $\eff{token}$.  The following code snippet shows a test generator
  for nested abstractions synthesized by \name{} based on the above specification
  and safety property:
  \begin{minted}[fontsize = \footnotesize, escapeinside=??, xleftmargin=18pt, linenos]{ocaml}
    let rec abs_gen (d: depth) =
      (let (i: int) = ?$\mykw{assume} $? (0 <= i && i < d) in ?$\eff{token}\; \boxed{[\widehat{i}]}$?) ?$\oplus$?
      (let t = random_stlc_ty () in ?$\eff{token}\;\boxed{(\lambda t.};$? abs_gen (d + 1); ?$\eff{token}\;\boxed{)}$?
  \end{minted}
  \noindent Observe that the ghost event used in
  \ref{eqn:deBruijn+prop} and the specification for
  $\eff{token}\;\boxed{)}$ is not present in the code; instead it acts
  as a hint to the synthesis procedure to construct a recursive
  function $\Code{abs\_gen}$ parameterized by the current depth $d$,
  ensuring correct tracking of abstraction nesting.

  \cbnewadding{In order to facilitate test exploration, token
    sequences should correspond to STLC terms that are both: (1)
    \emph{well-formed}: all opened lambda abstractions need to
    eventually be closed and all variables in a term should be bound,
    so that terms are not immediately rejected by the parser; and (2)
    \emph{well-typed}: to avoid terms that are discarded by the type
    checker.  The specification given above encodes only the first
    property, and thus permit a synthesized test generator to produce
    well-formed but ill-typed STLC terms. A more precise specification
    that encodes both conditions is provided in the supplementary
    material. Although this weaker specification does not prohibit the
    generation of ill-typed terms, it nonetheless correctly encodes
    the structure of Bruijn indices, using the ghost variable $d$ to
    track abstraction depth.}

%SJ: removing as it seems redundant with what we've already said.
%% \cbnewadding{\paragraph{Summary} As shown in the examples above, \tyName{}s often combine
%% (1)~\emph{basic correctness}, e.g., IFC example requires the two input programs to be indistinguishable at the low security level; otherwise, it violates the precondition of EENI and test exploration stops even before the IFC machine is executed.
%% (2)~\emph{search biases}, e.g., biasing integers
%% toward valid addresses to avoid memory faults unrelated to the target
%% property. Without the basic correctness constraints, test generators that always raise exceptions (e.g., performing \texttt{store} as the first operation) can even be allowed.
%% The latter constraints are relatively ``soft'' because not every value pushed
%% onto the IFC stack is later used as a memory address.
%% The same predicate may play either role depending on the testing setup.
%% In the STLC example, if the interpreter type-checks first and
%% rejects ill-typed fragments immediately, well-typedness behaves as basic correctness; otherwise, it is treated as search bias since well-typed programs are more likely to trigger potential substitution errors.
%% }

%\end{example}

  In the following, we present a procedure to synthesize generators that only
  produce traces consistent with \tyName{}-based specifications and
  SRE-given safety properties.  We emphasize that because these
  specifications are used primarily to constrain the search space of
  terms used to synthesize a test generator, they give developers
  significant leeway on the level of precision captured: less precise
  underapproximations simply lead to (potentially) less effective
  generators.

\section{Language}
\label{sec:formal}

\begin{figure}[th!]
  \vspace{-.3cm}
{\footnotesize
  \begin{alignat*}{2}
    \text{\textbf{Variables }}& \quad &\quad& x, y, z, \vnu, ...
    \\[-0.2em]\text{\textbf{Base Types}}& \quad & b  ::= \quad &  \Code{unit} ~|~ \Code{bool} ~|~ \Code{nat} ~|~ \Code{int} ~|~ \Code{id} ...
    \\[-0.2em]\text{\textbf{Basic Types}}& \quad & s  ::= \quad &  b ~|~ s\sarr s
    \\\text{\textbf{Pure Operations}}& \quad & \primop ::= \quad & {+} ~|~ {-} ~|~ {==} ~|~ {<} ~|~ {\leq} ~|~ \Code{rand\_int}~...
    \\[-0.2em]\text{\textbf{Constants }}& \quad &c ::= \quad & () ~|~ \mathbb{B} ~|~ \mathbb{Z} ~|~ \mathbb{I}...
    \\[-0.2em]\text{\textbf{Qualifiers}}& \quad & \phi ::= \quad & c ~|~ x ~|~ \primop\;\overline{v} ~|~ \bot ~|~ \top  ~|~ \neg \phi ~|~ \phi \land \phi ~|~ \phi \lor \phi ~|~ \phi \impl \phi ~|~ \forall x{:}b.\, \phi
    \\[-0.2em]\text{\textbf{Effectful Operations}}& \quad & \eff{op} ::= \quad & \eff{put} ~|~ \eff{get} ~|~ \eff{readReq} ~|~ \eff{readRsp} ~|~ \eff{write} ~|~ \eff{push} ~|~ \eff{pop} ~|~ ...
    % \\\text{\textbf{Event Kinds}}& \quad & k ::= \quad & \eff{gen}
    % ~|~ \eff{obs}
      % ~|~ \reqOp{op}{\ensuremath{\overline{v}}} ~|~ \respOp{op}{\ensuremath{\iota}}
    \\[-0.2em]\text{\textbf{Values }}& \quad & v ::= \quad & c ~|~ x
    ~|~ \zlam{x}{s}{e} ~|~ \zfix{f}{s}{x}{s}{e}
    \\[-0.2em]\text{\textbf{Expressions}}& \quad & e ::=\quad & v ~|~
    e \oplus e ~|~ v\; v ~|~ \primop\; \overline{v} ~|~ \eff{op}\;
    \overline{v} ~|~\zlet{x{:}b}{e}{e} ~|~ \sassume{\phi} ~|~
    \sassert{\phi}
    \\[-0.2em]\text{\textbf{Events}} & \quad & m ::= \quad &
    \langle\eff{op}{\; \overline{v}}\rangle \\
    \text{\textbf{Traces}} & \quad & \alpha ::= \quad & \emptr ~|~ m
    \;{::}\; \alpha ~|~ \alpha \listconcat \alpha
    \\[5pt]
    \text{\textbf{Notations}} & \quad & e_1;e_2 \doteq \quad &
    \zlet{\_}{e_1}{e_2}  \\
    & & \hspace{-2.5cm}\zlet{x{:}t}{\sassume{\phi(x)}}{e} \doteq \quad &
    \zlet{x{:}t}{\Code{rand\_int ()}}{\sassume{\phi(x)}; e}
    %  \mykw{let}\;\eff{op}\;\overline{v}\;\mykw{in}\;e \doteq \zgen{op}{\overline{v}}{e} \qquad\qquad
  \end{alignat*}
}
\vspace*{-.3in}
\caption{Syntax of \DSL{} expressions}
\label{fig:term-syntax}
\vspace*{-.2in}
\end{figure}

We formalize our approach using a core language, \DSL{}, for
expressing \textit{trace generators}-- nondeterministic programs that
probe the behaviors of a black-box system under test via a set of
(effectful) operators that the system provides. This calculus is
equipped with a type system that characterizes the set of traces that
a program \emph{must} be able to produce, \emph{if} the system under
test produces appropriate responses. The syntax of \DSL{} is shown in
\autoref{fig:term-syntax}; the calculus includes both pure and
effectful operations ($\primop$ and $\eff{op}$), nondeterministic
choices ($\oplus$), recursion ($\mykw{fix}$), and $\mykw{assume}$ and
$\mykw{assert}$ statements in the style of solver-aided
languages~\cite{BC+10}. As we saw in \autoref{sec:overview},
asynchronous effects can be modeled as a pair of $\eff{opReq}$ /
$\eff{opResp}$ operations, where the former generates an id (given as
$\iota$ in the syntax) that can then be subsequently referenced when
performing a corresponding response operation.
%\autoref{fig:term-syntax} defines some shorthand notations used by our
%examples, e.g., we use $e_1;e_2$ to simplify $\mykw{let}$-expressions
%when the remaining program does not use the bound variable, i.e.,
%$\zlet{\_}{e_1}{e_2}$.

\begin{wrapfigure}{r}{.5\textwidth}
  \footnotesize
  \begin{minipage}{.45\textwidth}
    \begin{flalign*}
    &\text{\textbf{Operator Semantics }}\
    \fbox{$\vDash \primop(\overline{c}) \Downarrow c \quad \alpha \vDash \eff{op}(\overline{c}) \Downarrow c$} \\
    & \text{\textbf{Operational Semantics }}\
    \fbox{$
    % \phi \Downarrow c \quad
    \steptr{\alpha}{e}{\alpha}{e}$}
\end{flalign*}
\end{minipage}
\\[5pt]
\begin{center}
  \begin{prooftree}
    \hypo{\vDash \primop(\overline{c}) \Downarrow c'}
    \infer1[\textsc{\small StPOp}]{
      \steptr{\alpha}{\primop \; \overline{c}\;}{[\;]}{c'}
    }
  \end{prooftree}
  \begin{prooftree}
    \hypo{\alpha \vDash \eff{op}\overline{c} \Downarrow c'}
    \infer1[\textsc{\small StEfOp}]{
      \steptr{\alpha}{\eff{op}\; \overline{c}\;}{[\eff{op}\;\overline{c}\;c']}{c'}
    }
  \end{prooftree}
  \end{center}
  \vspace*{-.05in}
    \caption{Selected reduction rules}
    \label{fig:selected-semantics}
  \vspace*{-.2in}
\end{wrapfigure}
\DSL{} is equipped with a small-step operational semantics similar to
other calculi for trace-based type systems~\cite{ZYDJ24}; selected
reduction rules are shown in
\autoref{fig:selected-semantics}.\footnote{\label{fn:sm-formal-details}%
  The full reduction rules; complete definitions of the auxiliary typing
  relations (including well-formedness); the full typing rules beyond
  \autoref{fig:selected-typing-rules}; further denotational details; and
  proofs of the theorems in \autoref{sec:metatheory} are all given in the
  \techreport{}.}
The reduction relation
$\steptr{\alpha}{e}{~\alpha'}{e'}$ constrains the output trace
$\alpha'$ produced by executing $e$ based on a \textit{history trace}
$\alpha$.  The evaluation of pure operations (rule \textsc{\small
  StPOp}) makes no use of trace, either history or output. The
result of an effectful operation, in contrast, can depend on the
events that preceded it, and it produces an output trace containing an
event whose payload contains both the arguments and return value of
the operation.

\subsection{Type System}

The syntax of types in \DSL{} is shown in \autoref{fig:type-syntax}.
Types include \emph{pure} refinement types, which describe pure
computations, and underapproximate Hoare Automata Types (\tyName{}s),
which describe effectful computations.
% which describe incorrectness triple $\effLR{H}\;e\;\effLR{F}$ as a
% type $\uhat{H}{\Unit}{F}$ of term $e$.
Pure refinement types are similar to those found in prior
underapproximate refinement type systems~\cite{CoverageType}, and
allow base types (e.g., {\Code{int}}) to be further constrained by a
logical formula or qualifier. In contrast to traditional type systems
where a pure type qualifier $\ort{b}{\phi}$ constrains every value a
pure expression \emph{may} evaluate to; the type qualifier
$\urt{b}{\phi}$ constrains a subset of the values it \emph{must}
evaluate to.  For example, the constant value $1$ can be assigned the
types $\ort{\Int}{\vnu = 1}$ and $\ort{\Int}{0 < \vnu}$ -- the value 1
satisfies the constraint in both \emph{overapproximate} qualifiers --
but, 1 cannot be assigned the type,
$\urt{\Int}{\vnu = 1 \lor \vnu = 2}$, since it cannot evaluate to both
1 and 2 as required by this \emph{underapproximate} qualifier. The
term $1 \oplus 2 \oplus 3$, in contrast, does have this
underapproximate type.

\paragraph{Underapproximate Hoare Automata Types}

Similar to other recent trace-based type systems~\cite{ZYDJ24}, \DSL{}
uses \emph{Symbolic Finite Automata} (SFAs)~\cite{Vea13,
  SFA-minimization, SFA-Transducers} to qualify the traces generated
by effectful computations; in contrast to those prior works, the
\tyName{}s in $\DSL{}$ use SFAs to \emph{underapproximate} the set of effects
that a program must produce. A \tyName{} $\uhat{H}{x{:}t}{F}$ consists
of a pure refinement type $t$ and two SFAs: a \emph{history} SFA $H$
that constrains the history trace under which a term is executed, and a
\emph{future} SFA $F$ that describes a subset of the events that it
must generate as a result. \tyName{}s represent SFAs using symbolic
regular expressions (SREs), but in principle could support any
encoding of SFAs, e.g., Symbolic LTL$_f$~\cite{LTLf}. SREs support
\emph{symbolic events}, $\msgB{op}{\overline{x}}{\phi}$, atomic
predicates that describe an effectful operation {$\eff{op}$} whose
arguments $\overline{x}$ satisfy the qualifier $\phi$. SREs provide a
convenient language for writing SFAs, supporting versions of the
standard regex operators shown in \autoref{fig:type-syntax}, including
union $A \lorA A$, sequence $A \seqA A$, and Kleene star $A^*$
operators.

\begin{figure}[tb!]
{\small
\begin{alignat*}{2}
    % \\\text{\textbf{Symbolic Event}}& \quad &se  ::= \quad & \msgB{op}{\overline{x}}{\phi}
    \text{\textbf{Symbolic Regex}}& \quad &H,F,A  ::= \quad & \emptyset ~|~ \epsilon ~|~ \msgB{op}{\overline{x}}{\phi} ~|~ A \lorA A ~|~ A \seqA A ~|~ A^* ~|~ A \setminus A ~|~ A \landA A ~|~ \anyA
    % \\& \quad &  \quad & \text{\textcolor{gray}{in Symbolic LTL$_f$ }}  se ~|~ A \lorA A ~|~ \nextA A ~|~ A \untilA A ~|~ \negA A ~|~ A \landA A
    \\\text{\textbf{Pure Refinement Types}}& \quad & t  ::= \quad & \urt{b}{\phi} ~|~  x{:}t\sarr t ~|~ x{:}t\sarr \tau ~|~ x{:}b \garr t ~|~ \effGray{op}{:}s \sarr t
    \\ \text{\textbf{Underapproximate HATs}}& \quad & \tau  ::= \quad & \uhat{H}{x{:}t}{F} ~|~ \tau \interty \tau
    \\\text{\textbf{Type Contexts}}& \quad &\Gamma ::= \quad  &\emptyset ~|~ x{:}t, \Gamma
  \end{alignat*}
  \vspace*{-.15in}
  {\footnotesize\begin{flalign*}
    {\small\text{\textbf{Type Aliases}}} & \;\; \msg{op}{\phi} \doteq \msgB{op}{\overline{x}}{\phi} \;\; \msgA{op}{\overline{v}} \doteq \msgB{op}{\overline{x}}{\overline{x = v}} \;\; \msgC{op} \doteq \msgB{op}{\overline{x}}{\top} \;\; A_1 A_2 \doteq A_1 \seqA A_2 \;\; b \doteq \urt{b}{\top}
    % \;\; \tau \doteq x{:}b\garr\tau
    % \\
    % & \widehat{\ort{b}{\phi}} \doteq \urt{b}{\phi} \;\; \widehat{\urt{b}{\phi}} \doteq \ort{b}{\phi} \;\; \widehat{x{:}t\sarr t'} \doteq x{:}t\sarr t' \;\; \widehat{x{:}t\sarr \tau} \doteq x{:}t\sarr \tau
  \end{flalign*}}
}
\vspace*{-.24in}
\caption{\DSL{} types}
\label{fig:type-syntax}
\vspace*{-.2in}
\end{figure}

The figure also gives the definition of the type aliases that
\autoref{sec:overview} used to limit user annotation burden.  When the
fields of a symbolic event are clear from context, we omit its
arguments $\overline{x}$, e.g., $\langle\eff{push}\;|\;v > 0\rangle$
means $\langle\eff{push}\;v\;|\;v > 0\rangle$. We omit quantifiers
with equality constraints, e.g., writing
$\langle\eff{push}\;v\;|\;v = 3\rangle$, as $\msgA{push}{3}$, and the
top ($\top$) qualifier in symbolic events and types.

%% \BD{We need to describe ghost variables at some point. If we are not
%%   going to use them in the intro+overview, we should explain that we
%%   left them are implicit there.}

\begin{example}[Ghost Events]
  \label{ex:stack-example}
  Consider an effectful stack library implemented as a fixed-size
  array whose elements grow monotonically.  An incorrect implementation
  that fails to grow the array when the stack is full may cause
  $\eff{push}$ operations to be dropped.  To detect this violation, we
  need to explore traces of the following form, where the last pop
  operation raises an ``empty stack'' exception:
  \vspace{-.6cm} \par\nobreak {\small%
    \begin{align*}
      \zevent{push}{1};\zevent{push}{2}; \zevent{push}{3}; \zevent{pop}{2}; \zevent{pop}{1};\textcolor{red}{\zevent{pop}{?}}
    \end{align*}}\noindent
  Because we would like the \tyName{} specifications for $\eff{push}$
  and $\eff{pop}$ to guide the synthesis of a test generator, they
  should be interpreted as underapproximations, consisting of nested
  pairs of $\eff{push}$ and $\eff{pop}$ events. This can be expressed
  by the following \tyName{} specifications, using ghost events
  $\effGray{pushI}$ and $\effGray{popI}$ as shown below:
  \vspace{-.4cm} \par\nobreak {\small %
    \begin{align*} %
      &\eff{push} :~ n{:}\Int \garr y{:}\Int \garr
        x{:}\urt{\Int}{y < \vnu} \sarr
        \uhat{\Code{Last}(\msgAGray{pushI}{n\;y})}{\Unit}{\msgA{push}{x}\seqA\msgAGray{pushI}{n{+}1\;x}}
      \\
      &\eff{pop} :~ n{:}\Int \garr m{:}\Int \garr \\[-0.4em]
      & \quad
        \uhat{\Code{Last}(\msgAGray{popI}{m\;\_}) \land
        \Code{Last}(\msgAGray{pushI}{n\;\_}) \land
        (\allA\seqA\msgAGray{pushI}{(n{-}m)\;x}\seqA\allA)}{x{:}\Int}{\msgA{pop}{x}\seqA\msgAGray{popI}{m{+}1\;
        x}}
    \end{align*}}\vspace{-.5cm} \par\nobreak\noindent
  Intuitively, the $\effGray{pushI}$ and $\effGray{popI}$ ghost events are
  simply $\eff{push}$ and $\eff{pop}$ events equipped with indices that
  record the number of push and pop operations.  As seen here, ghost events allow \tyName{}
    specifications to capture latent data dependencies between effectful
    operations.   The signature for $\eff{push}$ requires
  that the new pushed value ($x$) is greater than the last pushed value
  ($y$), and adds a new ghost event $\effGray{pushI}$ with an
  incremented index (i.e., $\msgAGray{pushI}{n{+}1\;x}$).
  The
  signature for $\eff{pop}$ also records the number of pop operations
  in the same way (i.e., $\msgAGray{popI}{m{+}1\;x}$); additionally, it
  requires the last popped value ($x$) to be equal to the
  $(n{-}m)^{th}$ pushed value, reflecting the ``last-in-first-out’'
  property of the stack.
\end{example}

\subsection{Typing Rules}

$\DSL{}$ is equipped with a pair of typing judgments on pure and
effectful expressions, $\Gamma \vdash v : t$ and $\Gamma\vdash e:
\uhat{H}{x{:}t}{F}$, respectively. Both judgments rely on a type
context, $\Gamma$, that maps from variables to pure refinement types
(i.e., $t$). As in prior work~\cite{CoverageType}, typing contexts are not
allowed to contain \tyName{}s--- doing so breaks several structural
properties (e.g., weakening) that are used to prove type safety. Both
rules are parameterized over an additional context $\Delta$, which
provides types for pure and effectful operators.

\paragraph{Auxiliary typing relations}
Our type system also relies on three additional auxiliary
relations. The first is a well-formedness relation $\Gamma
\wellfoundedvdash \tau$ that ensures, among other things, that all
qualifiers appearing in a type $\tau$ are closed under the current
typing context $\Gamma$ and that the verification conditions generated by
the synthesis algorithm presented in the next section can be encoded
as effectively propositional (EPR) sentences~\cite{Ramsey1987}, which
can be efficiently handled by an off-the-shelf theorem prover such as
Z3~\cite{de2008z3}.  We also require the
history regex $H$ to not be empty, which guarantees that it represents
a meaningful underapproximate precondition.

\begin{figure}[t!]
  {\footnotesize
  {\footnotesize
  \begin{flalign*}
   &\text{\textbf{Auxiliary Typing}} & & \fbox{$\Gamma \wfvdash \tau \quad \Gamma \vdash A \subseteq A \quad \Gamma \vdash \tau <: \tau$} &\text{\textbf{Typing}} &  & \fbox{$\Gamma\vdash \eff{op}: t \quad \Gamma\vdash v : t \quad \Gamma \vdash e: \tau$}
  \end{flalign*}
  }
  \\ \
  \mprooftr{30mm}
  {$\Gamma\wfvdash \tau[x\mapsto v]$}
  {SubG}
  {$\Gamma \vdash x{:}b\garr \tau <: \tau[x\mapsto v]$}
  \;
  \mprooftr{41mm}
  {$\Gamma \vdash t_1 <: t_2$ \quad
  $\Gamma,x{:}t_1 \vdash F_2 \subseteq F_1$ \quad
  $\Gamma,x{:}t_1 \vdash H_1 \subseteq H_2$
   }
  {SubHF}
  {$\Gamma \vdash \uhat{H_1}{x{:}t_1}{F_1} <: \uhat{H_2}{x{:}t_2}{F_2}$}
  \;
  \mprooftr{30mm}
  {$\Gamma \vdash \eff{op} : \overline{x{:}t_x}\sarr\uhat{H}{t}{F}$ \quad
  $\Gamma,\overline{x{:}t_x} \vdash H' \subseteq H$ \quad
  $\Gamma,\overline{x{:}t_x} \not\vdash H' \subseteq \emptyset$}
  {TOpHis}
  {$\Gamma\vdash \eff{op} : \overline{x{:}t_x}\sarr\uhat{H'}{t}{F}$}
  % \quad
  % \mprooftr{30mm}
  % {$\Gamma\vdash e: \tau$ \quad
  % $\Gamma \vdash \tau <: \tau'$\quad $\Gamma \vdash \tau' <: \tau$}
  % {TEq}
  % {$\Gamma\vdash e: \tau'$}
  \\[.8em]
  \mprooftr{20mm}
  {$\Gamma\vdash v : t$}
  {TRet}
  {$\Gamma\vdash v : \uhat{H}{t}{\epsilon}$}
  \quad
  \mprooftr{15mm}
  {$\Delta(\eff{op}) = t$}
  {TOpCtx}
  {$\Gamma\vdash \eff{op} : t$}
  \quad
  \mprooftr{58mm}
  {
  $\Gamma \vdash \eff{op} : \overline{x_i{:}t_i}\sarr \uhat{H}{t}{\msgA{op}{\overline{x_i}} \seqA F}$ \quad
  $\forall i. \;\Gamma \vdash v_i : t_i$}
  {TEffOp}
  {$\Gamma\vdash \eff{op}{\;\overline{v_i}} : \uhat{H}{t}{\msgA{op}{\overline{v_i}} \seqA F}$}
  \\[.8em]
  %%\mprooftr{56mm}
  %%{
  %%$\Gamma \vdash \eff{op} : \overline{x_i{:}t_i}\sarr \uhat{H}{t}{\msgA{op}{\overline{x_i}} \seqA F}$ \\
  %%  $\forall i. \;\Gamma \vdash v_i : t_i$}
  %% {TReqOp}
  %% {$\Gamma\vdash \reqOp{op}{\;\ensuremath{\iota\; \overline{v_i}}} : \uhat{H}{\Code{id}}{\reqEvt{op}{\ensuremath{\iota\; \overline{v_i}}} \seqA F}$}
  %% % \mprooftr{62mm}
  %% % {
  %% % $\Gamma\vdash \eff{op} : \overline{x_i{:}t_i}\sarr \uhat{H}{\Unit}{\msgA{op}{\overline{x_i}} \seqA F}$ \quad
  %% % $\forall i. \;\Gamma \vdash v_i : \widehat{t_i}$ \quad
  %% % $\Gamma\vdash e : \uhat{H \seqA (\msgA{op}{\overline{v_i}} \land \msg{op}{\phi}) \seqA F}{t}{F'}$}
  %% % {TGen}
  %% % {$\Gamma\vdash \zgen{op}{\overline{v_i}}{e} : \uhat{H}{t}{\msg{op}{\phi} \seqA F \seqA F'}$}
  %% \;
  %% \mprooftr{51mm}
  %% {
  %%   $\Gamma \vdash \eff{op} : \overline{x_i{:}t_i}\sarr \uhat{H}{t}{\msgA{op}{\overline{x_i}} \seqA F}$\\
  %%   $\forall i. \;\Gamma \vdash v_i : t_i$
  %% }
  %% {TRspOp}
  %% {$\Gamma\vdash \respOp{op}{\ensuremath{\iota}} : \uhat{H}{t}{\respEvt{op}{\ensuremath{\iota\; \overline{v_i}}} \seqA F}$}
  %% \\[.8em]
  \mprooftr{35mm}
  {
  $\Gamma\vdash e_1 : \tau_1$ \quad $\Gamma\vdash e_2 : \tau_2$}
  {TChoice}
  {$\Gamma\vdash e_1 \oplus e_2 : \tau_1 \interty \tau_2$}
  \quad
  \mprooftr{66mm}{
    $\Gamma\vdash e_1 : \uhat{H}{x{:}t_x}{F_1}$ \quad
    $\Gamma, x{:}t_x\vdash e_2 : \uhat{H \seqA F_1}{t}{F_2}$}
  {TLet}
  {
    $\Gamma\vdash \zlet{x}{e_1}{e_2} : \uhat{H}{t}{F_1 \seqA F_2}$
  }
  \\[.8em]
  \mprooftr{22mm}{
   $\Gamma, x{:}t_x\vdash e : \tau$
  }{TFun}{
   $\Gamma\vdash \zzlam{x}{e} : x{:}t_x\sarr \tau$}
  \
  \mprooftr{47mm}{
  }{TFixBase}{
   $\Gamma\vdash \zzfix{f}{x}{e} : \overline{x{:}s}\sarr\uhat{H}{\urt{b}{\bot}}{\emptyset}$}
  \
  \mprooftr{27mm}{
   $\Gamma\vdash \zzfix{f}{x}{e} : t$ \quad
   $\Gamma, f{:}t\vdash \zzlam{x}{e} : t'$
  }{TFixInd}{
   $\Gamma\vdash \zzfix{f}{x}{e} : t'$}
  % \quad
  % \mprooftr{60mm}{
  %  $\Gamma\vdash v_1 : y{:}t_y\sarr\uhat{H}{A}$ \quad
  %  $\Gamma\vdash v_2 : t_y$ \quad
  %  $\Gamma; \Delta' \vdash e' : \uhat{H \seqA A}{A'}$
  % }{TApp}{
  %  $\Gamma\vdash \zlet{x}{v_1\;v_2}{e} : \uhat{H}{A \seqA A'}$}
  }
  \vspace*{-.1in}
  \caption{Selected typing rules.}
  \label{fig:selected-typing-rules}
  \vspace*{-.15in}
  \end{figure}

  The type system also uses a semantic subtyping relation tailored for
  underapproximation; \autoref{fig:selected-typing-rules} presents the
  key subtyping rule for \tyName{}s.  The rule \textsc{SubG} allows
  the instantiation of a ghost variable in a \tyName{} into
  well-formed types. The rule \textsc{SubHF} uses an auxiliary
  inclusion relation on SREs $\Gamma \vdash A_1 \subseteq A_2$ to relate
  the history and future regexes of two \tyName{}s under the current
  type context $\Gamma$, as well as the return type using the standard
  subtyping relation on pure refinement types. Notably, because our
  typing judgements are underapproximate, this subtyping relation is
  covariant in the history regex and contravariant in the future
  regex, yielding behavior similar to the reverse rule of consequence
  found in Incorrectness Logic (IL)~\cite{OHP19}.  Like that rule,
  \textsc{SubHF} allows us to strengthen a postcondition (i.e., shrink
  the future trace), and weaken the precondition (i.e., enlarge the
  history trace).  Since the future automata in a \tyName{} can refer
  to the pure value it returns, the inclusion check between the future
  automata extends the typing context with a binding for the return
  value.

A subset of our typing rules is shown in
\autoref{fig:selected-typing-rules}. All our
rules assume any types they use are well-formed, so we elide the
corresponding well-formedness judgments from their premises.  The
\textsc{TRet} rule requires the future regex of a pure term
($\epsilon$) to only accept the empty trace.  The operator type is
retrieved from the global operator context $\Delta$ using the
\textsc{TOpCtx} rule.  For a given event, there may be multiple
reachable history traces (underapproximate preconditions) that can be
allowed by an operator that is determined by the context in which the
operator executes, and by the operator type alone.  We, therefore,
expect that the operator type found in the operator context is the
most precise one, i.e., it doesn't allow any unreachable history
traces, and admits underapproximation specialization.  The rule
\textsc{TOpHis} can then freely choose a \emph{non-empty} subset of the reachable
history traces to align with the actual calling context.  The
application rule \textsc{TEffOp} requires the arguments of an
effectful operator to be well-typed against its parameter types; it
records the effect as an event in the future trace along with any
additional ghost events ($F$).  The nondeterministic choice operator
is typed using the \textsc{TChoice} rule which merges the \tyName{}s
of the two branches.  The \textsc{TLet} rule requires the binding
expression $e_1$ to generate the first part of its future trace $F_1$,
with the rest of trace produced by the body expression $e_2$ under a
new history that includes $F_1$.  The typing rule for functions
(\textsc{TFun}) is standard, while the rule for recursive functions is
modeled after the corresponding IL rule, unrolling the recursion as
needed.  The \textsc{FixBase} rule types any recursive function with a
\tyName{} lacking reachability constraints (i.e.,
$\uhat{H}{\urt{b}{\bot}}{\emptyset}$). In contrast, the
\textsc{FixInd} rule permits typing a recursive function with type
$t'$ by requiring it to have another type $t$, and by checking that
its body can be assigned type $t'$ under the assumption that the
function itself has type $t$.

%% The \textsc{TReqOp} and the
%% \textsc{TRspOp} rule handle asynchronous effects, using
%% well-formedness conditions to ensure that the values carried by a
%% response are the same as those of a request in the history trace with
%% the same id.

\begin{example}[Operator Application Typing]
  \label{ex:operator-application-typing}
  We present a typing derivation for a simple generator that generates
  $\eff{push}$ and $\eff{pop}$ operations using the signatures from
  Example~\ref{ex:stack-example}. Consider how we might type check the
  following expression: %
  \vspace{-.4cm} \par\nobreak {\small %
    \begin{align*}
      &\emptyset \vdash \zlet{x{:}\Code{int}}{\sassume{0 < x}}{\eff{push}\;x; \eff{pop}} :
      \\&\qquad \uhat{\msgAGray{pushI}{0\;0}\msgAGray{popI}{0\;0}}{x{:}\urt{\Int}{0 < \vnu}}{\msgA{push}{x}\msgAGray{pushI}{1\;x}\msgA{pop}{x}\msgAGray{popI}{1\;x}}
    \end{align*}}%
  \noindent
  The signature of the generator is a \tyName{} that assumes a history
  trace in which the stack is empty (i.e., the ghost events
  $\effGray{pushI}$ and $\effGray{popI}$ both hold a $0$ index and a
  dummy value).  The body of the generator generates a positive number
  $x$ using $\mykw{assume}$; the underapproximate refinement type of
  $x$ signifies that it must evaluate to every positive integer. The
  generator then is expected to pop the same value.  We can thus type
  check its two effectful operation applications using the following
  judgement: %
  \vspace{-.4cm} \par\nobreak {\small %
    \begin{flalign*}
      &x{:}\urt{\Int}{0 < \vnu} \vdash \eff{push}\;x; \eff{pop} :
      \\&\qquad\qquad\ \uhat{\msgAGray{pushI}{0\;0}\msgAGray{popI}{0\;0}}{x{:}\urt{\Int}{0 < \vnu}}{\msgA{push}{x}\msgAGray{pushI}{1\;x}\msgA{pop}{x}\msgAGray{popI}{1\;x}}
    \end{flalign*}}%
  \noindent
  Using the \textsc{TEffOp} rule to type the first $\eff{push}$
  effect, we retrieve the operator's types from the operator context
  $\Delta$:%
  \vspace{-.4cm} \par\nobreak {\small %
    \begin{flalign*}
      &... \vdash \eff{push} : n{:}\Int \garr y{:}\Int \garr &&
      \\[-0.2em]& \qquad\qquad\ \  x{:}\urt{\Int}{y < \vnu} \sarr
      \uhat{\Code{Last}(\msgAGray{pushI}{n\;y})}{\Unit}{\msgA{push}{x}\msgAGray{pushI}{n{+}1\;x}}  &&\text{(by \textsc{TOpCtx})}
    \end{flalign*}}\noindent We can then apply the subtyping rule
  (\textsc{SubG}) to instantiate the signature's ghost variables
  (e.g., $n \mapsto 0$ and $y \mapsto 0$): \vspace{-.4cm} \par\nobreak
  {\small %
    \begin{flalign*}
      &... \vdash \eff{push} : x{:}\urt{\Int}{0 < \vnu} \sarr
      \uhat{\Code{Last}(\msgAGray{pushI}{0\;0})}{\Unit}{\msgA{push}{x}\msgAGray{pushI}{1\;x}} &&\text{(by \textsc{SubG})}
    \end{flalign*}}%
  \noindent
  The rule \textsc{TOpHis} can then be used to specialize its history
  trace to align with that of the actual calling context
  ($\msgAGray{pushI}{0\;0}\msgAGray{popI}{0\;0}$). Note that we do not
  need to also specialize the future regex, since it is already a
  prefix of the future trace of the target type: %
  \vspace{-.4cm} \par\nobreak {\small %
    \begin{flalign*}
      &... \vdash \eff{push} : x{:}\urt{\Int}{0 < \vnu} \sarr
      \uhat{\msgAGray{pushI}{0\;0}\msgAGray{popI}{0\;0}}{\Unit}{\msgA{push}{x}\msgAGray{pushI}{1\;x}}
      &&\text{(by \textsc{TOpHis})}
    \end{flalign*}}%
  \noindent
  The rest of program is now typed via the following judgement:
  \vspace{-.4cm} \par\nobreak {\small %
    \begin{align*}
      &x{:}\urt{\Int}{0 < \vnu}\vdash \eff{pop} : \uhat{\msgAGray{pushI}{0\;0}\msgAGray{popI}{0\;0}\msgA{push}{x}\msgAGray{pushI}{1\;x}}{x{:}\urt{\Int}{0 < \vnu}}{\msgA{pop}{x}\msgAGray{popI}{1\;x}} &&
    \end{align*}}%
  \noindent
  As before, after looking up the type of $\eff{pop}$ using
  \textsc{TOpCtx} (\ref{eqn:typing+l1}), we can instantiate the ghost variable $n$ as $1$
  and $m$ as $0$ via \textsc{SubG} (\ref{eqn:typing+l2}), align the return type via
  \textsc{SubHF} ((\ref{eqn:typing+l3}) and refine the history regex via the rule
  \textsc{TOpHis} (\ref{eqn:typing+l4}): %
  \vspace{-.4cm} \par\nobreak {\footnotesize %
    \begin{flalign}
      &... \vdash \eff{pop} : n{:}\Int \garr m{:}\Int \garr \notag
      \\[-0.2em]& \qquad\qquad\  \uhat{\Code{Last}(\msgAGray{popI}{m\;\_}) \land
      \Code{Last}(\msgAGray{pushI}{n\;\_}) \land
      (\allA\msgAGray{pushI}{(n{-}m)\;x}\allA)}{x{:}\Int}{\msgA{pop}{x}\msgAGray{popI}{m{+}1\;x}}
    && \tag*{\ding{172}}\label{eqn:typing+l1}
      \\&... \vdash \eff{pop} : \uhat{\Code{Last}(\msgAGray{popI}{0\;\_}) \land
      \Code{Last}(\msgAGray{pushI}{1\;\_}) \land
      (\allA\msgAGray{pushI}{1\;x}\allA)}{x{:}\Int}{\msgA{pop}{x}\msgAGray{popI}{1\;x}}
    && \tag*{\ding{173}}\label{eqn:typing+l2}
      \\&... \vdash \eff{pop} : \uhat{\Code{Last}(\msgAGray{popI}{0\;\_}) \land
      \Code{Last}(\msgAGray{pushI}{1\;\_}) \land
      (\allA\msgAGray{pushI}{1\;x}\allA)}{x{:}\urt{\Int}{0 <
        \vnu}}{\msgA{pop}{x}\msgAGray{popI}{1\;x}} &&
    \tag*{\ding{174}}\label{eqn:typing+l3}
      \\&... \vdash \eff{pop} :
      \uhat{\msgAGray{pushI}{0\;0}\msgAGray{popI}{0\;0}\msgA{push}{x}\msgAGray{pushI}{1\;x}}{x{:}\urt{\Int}{0
          < \vnu}}{\msgA{pop}{x}\msgAGray{popI}{1\;x}} &&
      \tag*{\ding{175}}\label{eqn:typing+l4}
    \end{flalign}}%
  \noindent
\end{example}

\subsection{Metatheory}
\label{sec:metatheory}

Although ghost events constrain \tyName{}s, they do not
appear in actual traces produced by the operational semantics. We
define an erasure function $\eraseGhost{\alpha}$ to remove ghost events from a trace $\alpha$.

\paragraph{Type denotations} Similar to other refinement type
systems~\cite{JV21}, types in $\DSL$ are denoted as their inhabitants
(i.e., $\denot{t}$ and $\denot{\tau}$).  The type context $\Gamma$ is
denoted as a \emph{substitution} $\sigma$ (i.e.,
$[x_1\mapsto v_1, x_2\mapsto v_2]$) that provides the assignments for
binding variables in $\Gamma$ and the denotation (the denoted
language) of SREs is a set of traces. Then, trace inclusion under a
type context is defined as
$\Gamma \vdash A \subseteq A' \doteq \forall \sigma \in
\denot{\Gamma}. \denot{\sigma(A)} \subseteq \denot{\sigma(A')}$.  The
type denotation of $\uhat{H}{x{:}\urt{b}{\phi}}{F}$ is:
\vspace{-.4cm} \par\nobreak {\small\begin{align*}
  \{e ~|~ \emptyset \basicvdash e : b \land \forall v{:}b. \phi[x\mapsto v] \impl \forall \alpha_f \in \denot{F[x\mapsto v]}. \exists \alpha_h \in\denot{H[x\mapsto v]}. \msteptr{\eraseGhost{\alpha_h}}{e}{\eraseGhost{\alpha_f}}{v} \}
\end{align*}}\noindent
A term $e$ inhabits $\uhat{H}{x{:}\urt{b}{\phi}}{F}$ iff, for trace $\alpha_f$ accepted by $F$, there exists history trace $\alpha_h$ accepted by $H$, such that the execution of $e$ under
$\alpha_h$ produces $\alpha_f$.
Our formulation constitutes a trace-based analogue of IL~\cite{OHP19}, in which effects are modeled as program state transitions:
\vspace{-.4cm} \par\nobreak {\small %
  \begin{align*}
    [P]\; e\; [Q] \doteq \forall s_\Code{post} \in \denot{Q}. \exists s_\Code{pre} \in \denot{P}. (s_\Code{pre}, e) \Downarrow s_\Code{post}
  \end{align*}}\noindent
A notable distinction between these two presentations is that we
characterize $F$ as the set of \emph{newly produced} traces rather
than the full post-execution trace, reflecting our trace-based
type rules. This formulation additionally avoids
redundant prefixes inherited from the execution history.

We expect the \tyName{}s provided by the operator typing context
$\Delta$ to be consistent with the corresponding auxiliary operator
semantics:
\begin{definition}[Well-formed operator typing context]
  \label{lemma:built-in-typing}
  The operator typing context $\Delta$ is well-formed
  iff the semantics of every operator $\eff{op}$ is
  consistent with the type $\overline{x{:}b_{\I{x}}}\garr
  \overline{y{:}t_{\I{y}}}\sarr \uhat{H}{z{:}t}{F}$ provided by $\Delta$, and $H$
  does not specify any unreachable traces.
  \vspace{-.4cm} \par\nobreak {\small\begin{align*}
\overline{\forall
      x{:}b_{\I{x}}}.\; \overline{\forall y \in \denotation{t_{\I{y}}}}.\; \forall H'. H'\not=\emptyset \land  H' \subseteq H \impl (\eff{op}\, \overline{y}) \in \denotation{\uhat{H'}{z{:}t}{F}}
  \end{align*}} \vspace{-.4cm} \par\nobreak\noindent
\end{definition}

\begin{theorem}[Fundamental Theorem]
  \label{theorem:fundamental} A well-typed term, i.e.,
  \(\Gamma\vdash e : \uhat{H}{t}{F}\), generates traces consistent
  with its \tyName:
  $\forall \sigma \in \denotation{\Gamma}. \sigma(e) \in
  \denotation{\sigma(\uhat{H}{t}{F})}$.
\end{theorem}

\begin{corollary}[Type Soundness]
  \label{theorem:sound} A generator \(e\) that satisfies
  \(\emptyset \vdash e : \uhat{\epsilon}{\Unit}{F}\) must produce all
  traces accepted by $F$, i.e.,
  $\forall \alpha \in
  \denot{F}. \msteptr{[]}{e}{\eraseGhost{\alpha}}{()}$.
\end{corollary}

\section{Synthesis}
\label{sec:synthesis}

Given an SRE $A$ describing some property of a black-box system under
test, Corollary~\ref{theorem:sound} ensures that a well-typed test generator
program $\emptyset \vdash e : \uhat{\epsilon}{\Unit}{F}$ \emph{must}
be able to produce \emph{every} trace $\alpha$ that is consistent with
both $F$ and $A$, i.e. $\alpha \in F \land A$. Thus, the high-level
goal of our synthesis procedure is to find a generator that a) uses
effectful operators in a way that is consistent with the temporal and
data-dependency constraints captured by their specifications in
$\Delta$, and b) has a \tyName{} whose future automaton shares a
meaningful overlap with $A$. Our synthesis algorithm simultaneously
solves both problems by using a loop, depicted in
\autoref{fig:syn-gen}, to iteratively \emph{refine} the target SRE
into a set of more concrete SREs, each of which can be directly mapped
to a well-typed generator program. Each iteration of this loop targets
a single event in an element of this set and adds events before and
after that event so that its dependencies are satisfied, ensuring that
the corresponding operator in a synthesized program is well-typed. Our
loop also employs a regex non-emptiness check to ensure that each
element of this set represents a generator that produces at least one
feasible trace. After the refinement loop has finished, a well-typed
generator program can be mechanically extracted using the refined
property as a template.

\begin{figure}[h!]
  \vspace*{-.15in}
  \centering
  \includegraphics[width=0.75\linewidth]{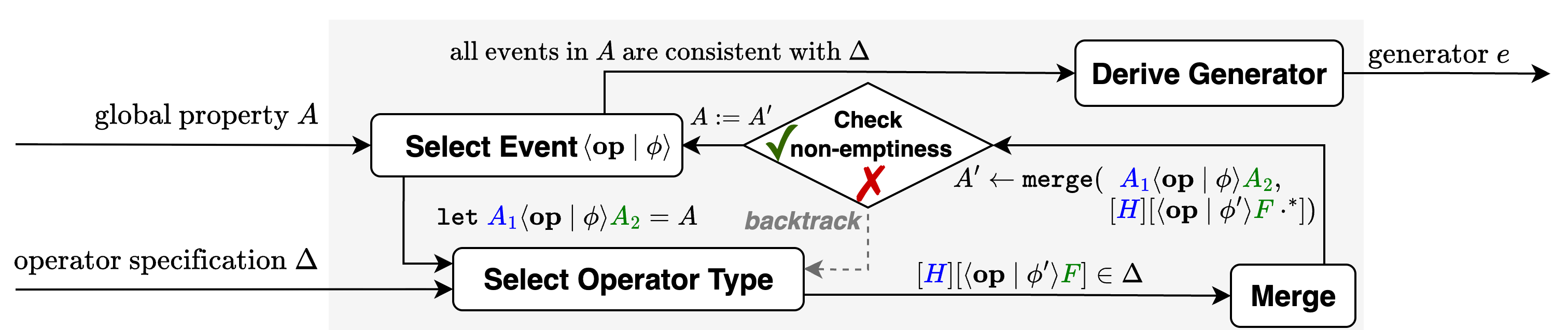}
  \caption{\name{}'s high-level test generator synthesis algorithm}
  \vspace*{-.15in}
  \label{fig:syn-gen}
  \vspace*{-.15in}
\end{figure}

\subsection{Synthesis Algorithm}

The main refinement loop in our synthesis algorithm operates over
\emph{abstract traces} --- sequences of symbolic events
$\msg{op}{\phi}$ plus alternatives defined by the following grammar: %
\vspace{-.4cm} \par\nobreak {\small
  \begin{flalign*}
    &\text{\textbf{Abstract Trace}} \quad \pi ::= \epsilon ~|~
    \msg{op}{\phi}^b ~|~ \pi \seqA \pi ~|~
    [\overline{\msg{op}{\phi}^b}]^*
  \end{flalign*}}\noindent Each symbolic event $\msg{op}{\phi}^b$ in
an abstract trace is tagged with a boolean $b$ indicating whether the
surrounding context satisfies its constraints. Operating over abstract
traces helps our synthesizer easily identify events with unresolved
dependencies. Every SRE can be normalized into a set of abstract
traces via an algorithm that applies a standard minimization algorithm
and then decomposes unions into abstract traces:
$\msg{op_1}{\phi_1} \cup \msg{op_2}{\phi_2}$ becomes
$\{\msg{op_1}{\phi_1}^\bot, \msg{op_2}{\phi_2}^\bot\}$, for
example.\footnote{The full definition of our SRE normalization
  procedure is provided in the \techreport{}.}

\begin{example}[SRE Normalization]\label{ex:sre-norm}
  The set of abstract traces corresponding to the history regex of
  $\eff{pop}$ from Example~\ref{ex:stack-example} includes the
  following abstract trace: %
  \vspace{-.4cm} \par\nobreak {\footnotesize %
    \begin{align*}
      \text{Original SRE:} \quad&  \Code{Last}(\msgAGray{popI}{m\;\_}) \land \Code{Last}(\msgAGray{pushI}{n\;\_}) \land (\allA\msgAGray{pushI}{(n{-}m)\;x}\allA)
      \\
      \text{One normalized trace:} \quad& \Pi_\Code{any}^*\seqA \msgAGray{popI}{m\;\_}^\bot \seqA \Pi_\Code{noPopI}^* \msgAGray{pushI}{(n-m)\;x}^\bot \seqA\Pi_\Code{noPopI}^* \msgAGray{pushI}{n\;\_}^\bot \seqA \Pi_\Code{noI}^* \tag{$\pi_\Code{pop1}$}\label{prop:popHis1}
      \\[-0.2em]
      \text{where} \quad& \Pi_\Code{any} \equiv [\msgCGray{pushI}^\bot \mid \msgCGray{popI}^\bot \mid \msgC{push}^\bot \mid \msgC{pop}^\bot]
      \\[-0.2em]
                                & \Pi_\Code{noPopI} \equiv [\msgCGray{pushI}^\bot \mid \msgC{push}^\bot \mid \msgC{pop}^\bot] \qquad\quad \Pi_\Code{noI} \equiv [\msgC{push}^\bot \mid \msgC{pop}^\bot]
    \end{align*}}\noindent
  This abstract trace imposes an ordering on the three intersected
  symbolic events in the original SRE; other traces correspond to
  different orderings. The alternatives after
  $\msgAGray{popI}{m\;\_}$ and $\msgAGray{pushI}{(n-m)\;x}$ in the
  trace disallow $\msgCGray{popI}$ event and
  $\msgCGray{pushI}$ events, respectively. The minimization step
  also produces traces that merge these intersected events: the
  procedure automatically considers the case where $m = 0$, which
  unifies the two $\effGray{pushI}$ events, resulting in the inclusion
  of the following abstract trace in the normalized set: %
  \vspace{-.4cm} \par\nobreak {\small %
\begin{align*}
  \Pi_\Code{any}^* \seqA \msgAGray{popI}{0\;\_}^\bot \seqA
  \Pi_\Code{noPopI}^* \seqA \msgAGray{pushI}{n\;x}^\bot \seqA \Pi_\Code{noI}^* \tag{$\pi_\Code{pop2}$}\label{prop:popHis2}
\end{align*}}\noindent
\end{example}

\setlength{\textfloatsep}{4pt}
\begin{algorithm}[t!]
  \Params{\quad Operator context $\Delta$, variables $\overline{x{:}b}$, and
    target trace property $A$}%
  \Output{\quad Generator $e$, such that
    $\overline{x{:}\urt{b}{\top}} \vdash e :
    \uhat{\epsilon}{\Unit}{A'} \land%
    \overline{x{:}\urt{b}{\top}} \vdash A' \subseteq A$} %
  $\Gamma \leftarrow \overline{x{:}\urt{b}{\top}}$ \tcp*[l]{initial
    type context} %
  $C \leftarrow \{(\Gamma, \pi) ~\mid~ \pi \in \normPlan(A)\}$
  \tcp*[l]{initialize set of abstract traces} %
  \While(\tcp*[h]{events still need to be refined}) %
  { %
    $\exists (\Gamma', \pi) \in C.\; %
    \pi = \pi_h \seqA \msg{op}{\phi}^\bot \seqA \pi_f$ } {
    $\Pi \leftarrow \refine(\Delta, \Gamma', \pi_h,
    \msg{op}{\phi}^\bot, \pi_f) \cup C - \{(\Gamma', \pi)\}$ }
\Return{$\derive(C)$} \tcp*[l]{synthesize recursions and derive generator program}
\caption{Test Generator Synthesis}
\label{algo:top-syn}
\end{algorithm}

\setlength{\textfloatsep}{4pt}
\begin{algorithm}[t!]
    \Params{\quad Operator context $\Delta$, typing context $\Gamma$,
      abstract trace $\pi_h \seqA \msg{op}{\phi}^\bot \seqA \pi_f$}%
    \Output{\quad Set of refined abstract traces $\Pi$ where the
      dependencies of $\msg{op}{\phi}$ are satisfied}
    $\overline{z{:}b}\garr\overline{y{:}t_y}\sarr\uhat{H}{x{:}t_x}{\msg{op}{\phi'}\seqA
      F} \leftarrow \Delta(\eff{op})$\tcp*[l]{retrieve \tyName of
      $\eff{op}$} %
    $\Gamma' \leftarrow \Gamma, \overline{z{:}\urt{b}{\top}},
    \overline{y{:}t_y}, x{:}t_x$ \tcp*[l]{add ghost variables and
      parameters types to type context} %
    $\msg{op}{\phi''} \leftarrow \msg{op}{\phi \land \phi'}$
    \tcp*[l]{refine target operation} %
    $H' \leftarrow \pi_h \land H$ \tcp*[l]{merge history regex}
    $F' \leftarrow \pi_f \land (F \seqA \allA)$ \tcp*[l]{merge future
      regex} %
    $ P \leftarrow \normPlan(H'\seqA \msg{op}{\phi''}^\top\seqA F')$
    \tcp*[l]{normalize refined trace} %
    \Return{$\{(\Gamma', \pi') \mid \pi' \in P \land \Gamma' \vdash
      \pi' \not \subseteq \emptyset \}$}
    \tcp*[l]{filter empty abstract traces} %
    \caption{Abstract Trace Refinement (\refine{})}
\label{algo:refine}
\end{algorithm}

Our top-level synthesis algorithm is shown in \autoref{algo:top-syn}.
Given a target trace property $A$ whose qualifiers only contain the
variables $\overline{x{:}b}$, this algorithm synthesizes a well-typed
$\DSL{}$ generator. The algorithm maintains a set of candidate
abstract traces, $C$, each of which is a refinement of $A$; this set
is initialized directly from $A$ (line 2). Each trace in $C$ is paired
with a typing context that the synthesizer uses to filter out
refinements that do not correspond to well-typed programs. The
algorithm first uses a loop (lines $3-4$) to refine $C$ until it only
contains abstract traces that are consistent with the operator context
$\Delta$, and then uses the resulting set of traces to derive the
final (recursive) generator (line $5$). Each iteration of this loop
nondeterministically chooses a target event in one of the current
abstract traces (line $2$) whose tag indicates its dependencies are
not satisfied; it then uses the $\refine$ subroutine to insert the
required operations before and after the target event.

\paragraph{Abstract trace refinement} The abstract trace refinement
subroutine, \refine{}, is shown in \autoref{algo:refine}. It first
retrieves the \tyName{} of the target operation $\eff{op}$ from
$\Delta$ (line $1$), and adds its parameters to the type context (line
$3$). \refine{} then merges the selected \tyName{} with the current
abstract trace (lines $4-6$) and normalizes the result to produce a
set of refinements of the input trace that are consistent with
$\eff{op}$ (line $7$).  Line $6$ effectively applies the
\textsc{SubHF} rule to drop any states that are unreachable after the
target operation; while line 4 effectively applies \textsc{TOpHis} to
ensure the trace that precedes the target operation satisfies its
specification. % \BD{The careful reader will notice we're using a very
  % precise specification of an operator here... Do we need to argue why
  % that's okay?}
Before returning the refined trace, \refine{} checks that it
corresponds to a meaningful test generator by filtering out any empty
traces (line $7$).

\begin{example}[Abstract Trace Refinement]\label{ex:refinement-process}
  We illustrate the refinement process by showing how it derives a
  trace corresponding to the generator from in
  Example~\ref{ex:operator-application-typing}. Assume the refinement
  loop begins with the following (singleton) set of abstract traces:
  \vspace{-.4cm} \par\nobreak {\small %
    \begin{align*}
      \Pi_\Code{any}^* \seqA \msg{pop}{\top} \seqA \Pi_\Code{any}^*
    \end{align*}}\noindent
  where the user requires that the test generator performs at least
  one $\eff{pop}$ operation.  In the first iteration, the symbolic
  event $\msg{pop}{\top}$ is selected for refinement using
  \refine{}. The algorithm first retrieves the \tyName{} of
  $\eff{pop}$ from $\Delta$: %
  \vspace{-.4cm} \par\nobreak{\small %
    \begin{align*}
      &\Delta(\eff{pop}) \equiv
        n{:}\Int \garr m{:}\Int \garr \uhat{\pi_\Code{pop1} \cup \pi_\Code{pop2} \cup ...}{x{:}\Int}{\msgA{pop}{x}\seqA\msgAGray{popI}{m{+}1\;x}}\seqA\Pi_\Code{any}^*
    \end{align*}}%
  \noindent We have normalized the history regex to illustrate that
  merging this \tyName{} with the current abstract trace and
  normalizing can result in multiple new abstract traces.  For
  example, property~\ref{prop:popHis2} in the history regex can yield
  the following trace: %
  \vspace{-.4cm} \par\nobreak{\small %
  \begin{align*}
    &(m{:}\urt{\Int}{\vnu = 0}, n{:}\urt{\Int}{\top},y{:}\urt{\Int}{\top},x{:}\urt{\Int}{\top}, \\
    &\qquad\qquad \Pi_\Code{any}^* \msgAGray{popI}{m\;y}^\bot \Pi_\Code{noPopI}^* \msgAGray{pushI}{n\;x}^\bot \Pi_\Code{noI}^* \msgA{pop}{x}^\top\msgAGray{popI}{m\;x}^\bot)
  \end{align*}} \vspace{-.4cm} \par\nobreak\noindent
The new variables in the type context require $m = 0$, as
mentioned in Example~\ref{ex:sre-norm}.  Although the \tyName{} of
$\eff{pop}$ does not explicitly add any new $\eff{push}$ or
$\eff{pop}$ events to the abstract trace, it does require multiple
ghost events, which will recover corresponding concrete events during
further refinement.  For example, the \tyName{} of $\effGray{pushI}$
is given below; it requires that a $\eff{pushI}$ event always occurs
after a $\eff{push}$ event, unless it is an initialization event
(i.e., $i = 0$): %
\vspace{-.4cm} \par\nobreak{\footnotesize %
  \begin{align*}
    \Delta(\effGray{pushI}) \equiv &\; i{:}\Int\sarr x{:}\Int \sarr \uhat{\allA\msgC{push}}{\Unit}{\msgBGray{pushI}{i\;x}{i > 0}} \interty \uhat{(\anyA\setminus\msgCGray{pushI})^*}{\Unit}{\msgBGray{pushI}{i\;x}{i = 0}} \\
    \Delta(\effGray{popI}) \equiv &\; i{:}\Int\sarr x{:}\Int \sarr \uhat{\allA\msgC{pop}}{\Unit}{\msgBGray{popI}{i\;x}{i > 0}} \interty \uhat{(\anyA\setminus\msgCGray{popI})^*}{\Unit}{\msgBGray{popI}{i\;x}{i = 0}}
  \end{align*}}\vspace{-.4cm} \par\nobreak\noindent
On the other hand, the $\msgAGray{popI}{m\;y}$ event in the abstract trace has index $m = 0$, and the algorithm can smartly determine that there is no previous $\eff{pop}$ event in this case. It is because the merge result of the first intersected type of $\effGray{popI}$ fails the non-emptiness check, where $m > 0 \land m = 0$ is unsatisfiable. Finally, this trace will be refined to one consistent with the program shown in Example~\ref{ex:operator-application-typing}:
\vspace{-.4cm} \par\nobreak{\small %
  \begin{align*}
    &(m{:}\urt{\Int}{\vnu = 0}, n{:}\urt{\Int}{\vnu = m{+}1}, y{:}\urt{\Int}{\top}, x{:}\urt{\Int}{\vnu > y},
    \\&\;\Pi_\Code{noI}^* \seqA \msgAGray{pushI}{m\;y}^\bot
    \seqA \msgAGray{popI}{m\;y}^\bot \seqA \Pi_\Code{noI}^* \seqA \msgA{push}{x}^\bot \seqA \msgAGray{pushI}{n\;x}^\bot \seqA \Pi_\Code{noI}^* \seqA \msgA{pop}{x}^\top \seqA \msgAGray{popI}{n\;x}^\bot \seqA \Pi_\Code{any}^*) \tag{$\pi_\Code{stack}$}\label{prop:stackTrace}
  \end{align*}}\vspace{-.6cm} \par\nobreak\noindent
\end{example}

\paragraph{Deriving a Generator}
The second component of our synthesis algorithm, \derive{}, derives
generators for each abstract trace in the set produced the refinement
loop and then combines them using $\oplus$ to construct the final
generator.  This process relies on a subroutine, \deriveTrace{}, that
generates straightline programs without recursion by dropping by all
ghost events and variables and then mapping each symbolic event to a
corresponding effectful operation with appropriate arguments. So,
given the following context and abstract trace:
\vspace{-.4cm} \par\nobreak{\small%
  \begin{align*} (m{:}\urt{\Int}{\vnu = 0}, n{:}\urt{\Int}{\vnu =
      m{+}1}, y{:}\urt{\Int}{\top}, x{:}\urt{\Int}{\vnu > y},
    \msgA{push}{x} \seqA \msgA{pop}{x})
  \end{align*}
}\noindent \deriveTrace{} will synthesize the following program:
\begin{minted}[fontsize = \footnotesize, escapeinside=??, xleftmargin=10pt, linenos]{ocaml}
  let (x: int) = ?$\mykw{assume}\; (\exists m\;n\;y. m = 0 \land n = m{+}1 \land x > y)$? in push x; (* equivalent to assume true *)
  let (z: int) = pop () in assert (x == z)
\end{minted}
\derive{} applies a sketch-style methodology to synthesize a recursive
generator, using \deriveTrace{} to fill in the holes of a template
recursive function definition. It does so by associating each of the
holes in the template with a subsequence of the abstract trace,
treating each occurence of a Kleene star in an abstract trace as a
candidate placeholder for a recursive call. In order to ensure a
particular partitioning is sensible, the algorithm simulates a bounded
number of loop iterations to ensure a completed sketch is well-typed.

% in abstract traces are interpreted as recursion
% points. The synthesis algorithm $\derive$ non-deterministically either
% (i) generates a sequential program by eliminating all ghost events and
% the Kleene stars \footnote{The definition of $\derive$ is provided in
%   the \techreport{}.}, or (ii) synthesizes a recursive program by
% expanding the Kleene star using recursion. We employ a template-based
% strategy, in which the algorithm matches the given abstract trace with a
% template $T$, together with abstract traces for the recursion body
% ($\Pi_{\Code{rec}}$), as well as its associated history ($\Pi_h$) and
% future ($\Pi_f$) traces. As dictated by \textsc{TFixInd}, the algorithm
% unrolls the recursive body up to a bounded depth and invokes
% refinement loop (lines 2-3 in \autoref{algo:top-syn}) to ensure that
% the unrolled abstract traces are well-typed. The synthesized programs
% for these three traces are then composed to instantiate the
% placeholders in the recursive template.

\begin{example}[Recursion Template]
  Consider the following template which defines and uses a recursive
  function $f$ that calls itself once:
  $T \doteq \zlet{f}{\mykw{fix}\ f.\lambda ().\ () \oplus \boxed{e_1};
    f\;(); \boxed{e_2}}{(\boxed{e_3};f\;();\boxed{e_4})}$. This
  template includes a pair of holes before and after the recursive
  call in the definition of $f$, and another pair around the top-level
  call to $f$. We can build the following two program sketches by
  mapping portions of $\pi_\Code{stack}$ from
  Example~\ref{ex:refinement-process} to the four holes in $T$: %
  \vspace{-.4cm} \par\nobreak {\footnotesize %
    \begin{align*}
      &
        \underbrace{\Pi_\Code{noI}^*\msgAGray{pushI}{m\;y}\msgAGray{popI}{m\;y}\Pi_\Code{noI}^*}_{e_3}\;\;
        \underbrace{\msgA{push}{x}\msgAGray{pushI}{n\;x}}_{e_1}\;\;
        \underbrace{\Pi_\Code{noI}^*}_\text{$f$\;()}\;\;
        \underbrace{\msgA{pop}{x}\msgAGray{popI}{n\;x}}_{e_2}\;\;
        \underbrace{\Pi_\Code{noI}^*}_{e_4}
      \\[1pt]
      & \underbrace{\Pi_\Code{noI}^*\msgAGray{pushI}{m\;y}\msgAGray{popI}{m\;y}\Pi_\Code{noI}^*\msgA{push}{x}\msgAGray{pushI}{n\;x}}_{e_3}\;\;
        \underbrace{\Pi_\Code{noI}^*}_{e_1}\;\;
        \underbrace{\Pi_\Code{noI}^*}_\text{$f$\;()}\;\;
        \underbrace{\msgA{pop}{x}\msgAGray{popI}{n\;x}}_{e_2}\;\;
        \underbrace{\Pi_\Code{noI}^*}_{e_4}
    \end{align*}}\vspace{-.3cm} \par\nobreak \noindent %
  We can then simulate unrolling the recursive call in a solution to
  each sketch by duplicating the subtraces labeled $e_1\; f\;()\;
  e_2$ a bounded number of times and checking whether they are consistent with
  $\pi_\Code{stack}$: doing so reveals that the second sketch would
  not yield a valid solution, as the recursive call would repeatedly
  perform $\eff{pop}$ events without any $\eff{push}$ events. The
  first sketch, in contrast, is consistent with $\pi_\Code{stack}$;
  using \deriveTrace{} to synthesize generators for each $e$ yields
  the following:
  \vspace{-.4cm}
  \begin{minted}[fontsize = \footnotesize, escapeinside=??, xleftmargin=10pt, linenos]{ocaml}
let rec f () = () ?$\oplus$?
    let (x: int) = int_gen () in push x; (* ?$e_1$? *)
    f ();
    let (z: int) = pop () in assert (x == z) (* ?$e_2$? *)
in (* ?$e_3$? *) f () (* ?$e_4$? *)
\end{minted}
\vspace{-.2cm}
\end{example}

\begin{theorem}[\autoref{algo:top-syn} is Sound]
  \label{theorem:algo-sound}
  Every generator synthesized by \autoref{algo:top-syn} is type-safe
  with respect to our declarative typing system.
\end{theorem}

\section{Implementation And Evaluation}\label{sec:impl}
We have implemented a tool based on the above approach, called
\name{}, that synthesizes effectful test generators from user-provided
\tyName{} specifications and property.  The output program in our DSL
can be transpiled to OCaml (QCheck) or P.  Our experiments use the
OCaml transpilation to test effectful abstract datatype libraries,
serializers, transformers, and interpreters, concurrent data
structures and databases; we use the P target to test reactive
distributed system models (i.e., message-passing systems defined as a
collection of actors) by using the synthesized generator as a
property-directed scheduler that replaces P's default random one.
Notably, the same specification and synthesis strategy is used in both
cases.  \name{} consists of approximately 14K lines of OCaml code and
uses Z3~\cite{de2008z3} as its backend SMT solver.

\paragraph{Evaluation setting} \name{} takes two inputs: a target
property, expressed as an SRE, and an operator context that contains
\tyName{} specifications for the operators used in the SUT.  During
synthesis, \name{} decomposes the target property to build a generator
consistent with the traces that the property accepts, but constrained
to respect the history and future traces of any effects it
references.  Because operator types are interpreted
underapproximately, there are potentially many valid candidate
programs that are consistent with the \tyName{} specifications of the
operators they use.  Our implementation considers a union of
well-typed synthesis candidates, the number of which is bounded by the
amount of time allotted for synthesis (3 minutes in our experiments).
% Non-determinism also manifests in how recursive functions are unrolled
% - recall that in an underapproximate setting, the body of a recursive
% function need only be unrolled until it can be shown that the effects
% performed by the loop's continuation are justified by the effects
% performed within the loop body and its history trace.
Synthesized generators use $\mykw{assume}$ statements to ensure that
data dependencies among effectful operations are respected and use
$\mykw{assert}$ statements to flag property violations.
\cbnewadding{%
  Like other PBT setups that rely on precondition-style filtering, our
  evaluation adopts an \emph{angelic} reading of $\mykw{assume}$: when
  the SUT or runtime rejects a candidate trace (e.g., when the
  interpreter aborts on an ill-typed expression), we treat that run as
  uninformative and \emph{retry} with a fresh sample rather than
  counting it toward a found violation.  This matches how generators
  are typically used in practice and avoids conflating specification
  imprecision with genuine safety failures.}

Our experimental evaluation addresses four research questions:
\begin{itemize}
\item[\textbf{Q1}:] Is \name{} \emph{expressive}? Can it synthesize
  generators for a diverse set of benchmarks with interesting
  non-trivial structural and semantic safety properties?
\item[\textbf{Q2}:] Is \name{} \emph{effective}? Do synthesized
  generators enable more targeted exploration of the state
  space to witness violations of provided safety properties than
  existing techniques?
\item[\textbf{Q3}:] Is \name{} \emph{efficient}? Is it able to
  synthesize meaningful generators reasonably quickly?
  \cbnewadding{\item[\textbf{Q4}:] Is \name{} \emph{robust}? Is it
    able to synthesize meaningful generators for benchmarks even when
    provided with weaker specifications than necessary?}
\end{itemize} %

\newcolumntype{P}[1]{>{\centering\arraybackslash}p{#1}}
\newcommand{\timeout}{{\tiny Timeout}}

\begin{table}[t!]
  \renewcommand{\arraystretch}{0.8}
  % \vspace*{-.1in}
  \caption{\small Using \name{} to synthesize effectful test
    generators in OCaml. Benchmarks from prior work are annotated with
    their source: HAT~\cite{ZYDJ24}, the OCaml multicore PBT
    framework~\cite{OcamlMulticorePBT}, QuickCheck on IFC
    machines~\cite{pbt-ifc} and well-typed De Bruijn STLC
    programs~\cite{CoverageType}; as well as
    OLTP-Bench~\cite{OLTPBench}, an extensible testbed for
    benchmarking relational databases, and MonkeyDB~\cite{MonkeyDB}, a
    mock storage system for testing weak isolation levels.  The
    baselines used for comparison are taken from the original sources
    when they exist (these are labeled with $^{\dagger}$), otherwise
    we use QCheck's default random generator.  \name{} synthesized
    correct-by-construction test generators for all benchmarks within
    3 minutes; the number of generator variants synthesized is bounded
    at 3.  We set a 30-minute bound for the baselines to find a
    property violation. }
  \vspace*{-.1in} \footnotesize
%\multicolumn{2}{c||}{} &
\setlength{\tabcolsep}{0.55pt}
\begin{tabular}{c|P{5.1cm}||ccc||cc||c|ccc}
  \toprule
  \multicolumn{2}{c||}{Benchmark} & \#$\eff{op}$ & \multicolumn{2}{c||}{\#qualifier} & \multicolumn{2}{c||}{\scriptsize \# Num. Executions} & t$_\text{total}$(s) & \#evt & \#refine & \#SMT \\
  \midrule
Name & Property &   &uHAT &goal & {\scriptsize Clouseau} & {\scriptsize Baseline} &  &  \\
\midrule
\textsf{Stack}\cite{ZYDJ24} & {\scriptsize  Pushes and pops are correctly paired.} & $6$ & $12$ & $3$ & $1.0$ & $31.9$ & $0.60$ & $15$ & $4$ & $98$\\
\midrule
\textsf{Set}\cite{ZYDJ24} & {\scriptsize Membership holds for every element inserted into the set.} & $7$ & $22$ & $3$ & $1.0$ & $22.1$ & $1.28$ & $12$ & $12$ & $255$\\
\midrule
\textsf{Filesystem}\cite{ZYDJ24} & {\scriptsize A valid file path only contains non-deleted entries.} & $7$ & $67$ & $4$ & $1.3$ & $2812.1$ & $6.02$ & $13$ & $14$ & $853$\\
\midrule
\midrule
\textsf{Graph}\cite{ZYDJ24} & {\scriptsize A serialized stream of nodes and edges is restored to
  form a fully-connected graph.} & $6$ & $24$ & $4$ & $1.0$ & \timeout{} & $16.61$ & $16$ & $16$ & $2114$\\
\midrule
\textsf{NFA}\cite{ZYDJ24} & {\scriptsize An NFA reaches a final state for every string in the language it accepts.} & $8$ & $50$ & $4$ & $1.0$ & \timeout{} & $21.85$ & $26$ & $6$ & $5080$\\
\midrule
\textsf{IFCStore}\cite{pbt-ifc} & {\scriptsize A well-behaved IFC program containing a $\Code{Store}$ command never leaks a secret.} & $8$ & $37$ & $2$ & $2.8$ & $4811.7^{\dagger}$ & $1.02$ & $8$ & $8$ & $160$\\
\midrule
\textsf{IFCAdd}\cite{pbt-ifc} & {\scriptsize A well-behaved IFC program containing an $\Code{Add}$ command never leaks a secret.} & $8$ & $37$ & $2$ & $1.6$ & $3383.4^{\dagger}$ & $0.96$ & $8$ & $8$ & $154$\\
\midrule
\textsf{IFCLoad}\cite{pbt-ifc} & {\scriptsize A well-behaved IFC program containing a $\Code{Load}$ command never leaks a secret.} & $8$ & $37$ & $2$ & $15.4$ & $11293.7^{\dagger}$ & $5.76$ & $16$ & $18$ & $930$\\
\midrule
\cbnewadding{\textsf{IFCLoad}$_{\mathsf{W}}$\cite{pbt-ifc}} & \cbnewadding{{\scriptsize \textsf{IFCLoad} without search bias (cf. \autoref{sec:overview}).}} & \cbnewadding{$8$} & \cbnewadding{$33$} & \cbnewadding{$2$} & \cbnewadding{$115.4$} & \cbnewadding{$11293.7^{\dagger}$} & \cbnewadding{$4.83$} & \cbnewadding{$16$} & \cbnewadding{$18$} & \cbnewadding{$754$}\\
\midrule
\textsf{DeBruijnFO}\cite{CoverageType} & {\scriptsize An STLC interpreter correctly evaluates a
\emph{well-typed first-order} STLC program that uses a de Bruijn representation.} & $10$ & $117$ & $4$ & $1.5$ & $634.0^{\dagger}$ & $40.38$ & $18$ & $26$ & $4765$\\
\midrule
\textsf{DeBruijnHO}\cite{CoverageType} & {\scriptsize An STLC interpreter correctly evaluates a
\emph{well-typed higher-order} STLC program that uses a de Bruijn representation.} & $10$ & $118$ & $4$ & $1.4$ & \timeout{}$^{\dagger}$ & $91.73$ & $18$ & $31$ & $9034$\\
\midrule
\cbnewadding{\textsf{DeBruijnHO}$_{\mathsf{W}}$\cite{CoverageType}} & \cbnewadding{{\scriptsize \textsf{DeBruijnHO} with \emph{untyped higher-order} programs (cf. \autoref{sec:overview}).}} & \cbnewadding{$10$} & \cbnewadding{$107$} & \cbnewadding{$4$} & \cbnewadding{\timeout{}} & \cbnewadding{\timeout{}$^{\dagger}$} & \cbnewadding{$3.23$} & \cbnewadding{$7$} & \cbnewadding{$22$} & \cbnewadding{$957$}\\
\midrule
\midrule
\textsf{Shopping}\cite{MonkeyDB} & {\scriptsize All items added to a cart can be checked-out.} & $10$ & $60$ & $3$ & $1.0$ & $20.0^{\dagger}$ & $22.10$ & $17$ & $30$ & $1269$\\
\midrule
\textsf{HashTable}\cite{OcamlMulticorePBT} & {\scriptsize No updates to a concurrent hashtable are ever lost.} & $13$ & $38$ & $4$ & $1.0$ & $2.5^{\dagger}$ & $0.76$ & $6$ & $6$ & $57$\\
\midrule
\textsf{Transaction} & {\scriptsize Asynchronous read operations are logically atomic.} & $8$ & $45$ & $3$ & $1.2$ & \timeout{} & $2.89$ & $15$ & $16$ & $423$\\
\midrule
\cbnewadding{\textsf{Transaction}$_{\mathsf{W}}$} & \cbnewadding{{\scriptsize \textsf{Transaction} without read atomicity enforcement (cf. \autoref{sec:overview}).}} & \cbnewadding{$8$} & \cbnewadding{$40$} & \cbnewadding{$3$} & \cbnewadding{\timeout{}} & \cbnewadding{\timeout{}} & \cbnewadding{$2.24$} & \cbnewadding{$15$} & \cbnewadding{$15$} & \cbnewadding{$372$}\\
\midrule
\textsf{Courseware}\cite{MonkeyDB} & {\scriptsize Every student enrolled in a course exists in the
  enrollment database for that course.} & $16$ & $106$ & $3$ & $1.0$ & $57.5^{\dagger}$ & $27.71$ & $18$ & $33$ & $1479$\\
\midrule
\textsf{Twitter}\cite{MonkeyDB} & {\scriptsize Posted tweets are visible to all followers.} & $16$ & $99$ & $3$ & $1.0$ & $6.3^{\dagger}$ & $61.96$ & $24$ & $24$ & $2339$\\
\midrule
\textsf{Smallbank}\cite{OLTPBench} & {\scriptsize Account updates are strongly consistent.} & $22$ & $162$ & $3$ & $1.9$ & \timeout{} & $163.55$ & $28$ & $29$ & $10263$\\
\bottomrule
\end{tabular}
\label{tab:evaluation}
\vspace*{-.15in}
\end{table}

\begin{table}[t!]
  \renewcommand{\arraystretch}{0.8}
  % \vspace*{-.1in}
  \caption{\small Experimental results of using \name{} to synthesize
    property-directed schedulers for reactive distributed
    systems. Benchmarks from prior work are annotated with their
    source: P~\cite{DGJ+13}($^{\dagger}$),
    ModP~\cite{ModP}($^{\diamond}$) an extension of P with support for
    modules, and MessageChain~\cite{MessageChain}($^{\star}$), an
    automated verification tool for P.  We also include a real-world
    model of a two-phase commit protocol
    (\textsf{AnonReadAtomicity}$^{\square}$) used by a major cloud vendor.
    The components under test are written in P, and handler
    specifications for the SUTs components (message-passing actors)
    are given as \tyName{}s.  \name{} synthesizes a set
    of schedulers, each of which specifies a distinct scheduling
    order for messages, all consistent with provided specifications.
    We set a $3$ minute time bound for the synthesis procedure
    (t$_\text{total}$ is the average time to find a single
    controller.)  We set a bound of 30 minute time bound for the P
    baselines to generate a faulty execution.}
\vspace*{-.1in}
\footnotesize
\setlength{\tabcolsep}{1.3pt}
\begin{tabular}{c|P{4.0cm}||ccc||cc||c|ccc}
  \toprule
  \multicolumn{2}{c||}{Benchmark} & \#$\eff{op}$ & \multicolumn{2}{c||}{\#qualifier} & \multicolumn{2}{c||}{\scriptsize \# Num. Executions} & t$_\text{total}$(s) & \#evt & \#refine & \#SMT \\
  \midrule
Name & Property &   &uHAT &goal & {\scriptsize Clouseau} & {\scriptsize Baseline} &  &  \\
  \midrule
  {\scriptsize\textsf{Database}} & {\scriptsize The database maintains a Read-Your-Writes policy.} & $4$ & $15$ & $3$ & $1.0$ & $5.6$ & $0.32$ & $6$ & $6$ & $49$\\
  \midrule
  {\scriptsize\textsf{Firewall}}\cite{MessageChain} & {\scriptsize  Requests generated inside the firewall are eventually propagated to the outside.} & $5$ & $26$ & $4$ & $1.0$ & $10.0$ & $0.59$ & $10$ & $10$ & $222$\\
  \midrule
  {\scriptsize\textsf{RingLeaderElection}}\cite{MessageChain} & {\scriptsize There is always a single unique  leader.} & $3$ & $14$ & $2$ & $1.0$ & $17.7$ & $0.83$ & $8$ & $8$ & $100$\\
  \midrule
  {\scriptsize\textsf{BankServer}}\cite{DGJ+13} & {\scriptsize Withdrawals in excess of the available balance are never allowed.} & $6$ & $42$ & $1$ & $1.0$ & $2.2^{\dagger}$ & $0.17$ & $5$ & $5$ & $27$\\
  \midrule
  {\scriptsize\textsf{Simplified2PC}}\cite{DGJ+13} & {\scriptsize Transactions are atomic.} & $9$ & $32$ & $5$ & $2.1$ & $5.8^{\dagger}$ & $1.03$ & $8$ & $8$ & $88$\\
  \midrule
  {\scriptsize\textsf{HeartBeat}}\cite{DGJ+13} & {\scriptsize All available nodes are identified by a  detector.} & $7$ & $31$ & $3$ & $1.0$ & $6.0^{\dagger}$ & $1.10$ & $14$ & $14$ & $145$\\
  \midrule
  {\scriptsize\textsf{ChainReplication}}\cite{ModP} & {\scriptsize Concurrent updates are never lost.} & $7$ & $72$ & $4$ & $1.0$ & $500.0^{\dagger}$ & $19.24$ & $12$ & $169$ & $2054$\\
  \midrule
  {\scriptsize\textsf{Paxos}}\cite{ModP} & {\scriptsize Logs are correctly replicated.} & $10$ & $110$ & $2$ & $1.0$ & $667.7^{\dagger}$ & $23.70$ & $14$ & $77$ & $1763$\\
  \midrule
  {\scriptsize\textsf{Raft}} & {\scriptsize Leader election is robust to faults.} & $9$ & $39$ & $6$ & $1.0$ & \timeout{} & $26.84$ & $12$ & $78$ & $1262$\\
  \midrule
  {\scriptsize\textsf{AnonReadAtomicity}} & {\scriptsize Read Atomicity is preserved.} & $17$ & $113$ & $3$ & $1.0$ & $53.3^{\dagger}$ & $18.35$ & $16$ & $16$ & $1909$\\
\bottomrule
  \end{tabular}
\label{tab:evaluation2}
\vspace*{-.15in}
\end{table}

We have evaluated \name{} on a diverse corpus of programs\footnote{The
  supplemental material includes additional details about our
  benchmarks.} % as well as a docker image that contains the \name{}
  % source code and benchmarks.
drawn from a variety of sources (described in the captions of
\autoref{tab:evaluation} and~\autoref{tab:evaluation2}) % ; all of the
% benchmarks except for \textsf{Database} and \textsf{Raft} were
% written by others BD: This statement is incomplete (what about
% AnonReadAtomicity and Transaction?) and redundant (the previous
% sentence stated that they came from several sources), so I'm
% dropping it.
(\textbf{Q1}). Our benchmarks are grouped into four categories.
The first are data structure benchmarks (\textsf{Stack}, \textsf{Set},
\textsf{Filesystem}) where the test generator probes for violations of
structural invariants.  The second are evaluators - programs that
input serialized data, transform them into an internal AST
representation, and then interpret them; we consider evaluators for
data structures like graphs and automata (\textsf{Graph},
\textsf{NFA}), information flow control machines, and simply-typed
lambda-calculus, the latter two both described
in~\autoref{sec:overview}.  The third are concurrent and database
benchmarks (\textsf{Transaction}, \textsf{Shopping},
\textsf{Courseware}, \textsf{Twitter}, \textsf{Smallbank}) that test
properties related to consistency and synchronization.  Property
violations in any of these categories only manifest in specific
configurations (e.g., a violating test case for \textsf{NFA} involves
constructing an automata containing multiple nodes with more than two
outgoing edges).  We evaluate these three classes of benchmarks by
transpiling synthesized generators into OCaml, leveraging its QCheck
library as needed. The SUT for each of these benchmarks is known to be
buggy; in all cases we use implementations either available from prior
sources directly (if the OCaml source was available), or by directly
translating them into OCaml.  The fourth group of benchmarks shown
in~\autoref{tab:evaluation2} considers applications involving various
kinds of distributed protocols; these are executed using the P runtime
environment, by transpiling synthesized generators written in our DSL
into P's state machine modeling language.  Here as well, the protocol
being tested is known to be faulty from the sources they were adopted
from. \cbnewadding{The \textsf{IFCLoad}$_{\mathsf{W}}$,
  \textsf{DeBruijnHO}$_{\mathsf{W}}$, and
  \textsf{Transaction}$_{\mathsf{W}}$ benchmarks are variants of other
  benchmarks that use the weaker specifications described in
  \autoref{sec:overview}; the end of this section discusses these
  results in more detail (\textbf{Q4}).}

%% For example, \textsc{RingLeaderElection}
%% models a leader election algorithm.  Each message transmitted by the
%% model's actors contain a source ($\I{src}$) and destination
%% ($\I{dest}$), to allow recipient's to distinguish
%% messages from different nodes. Then, an $\eff{elect}$ handler can
%% be given the following \PAT{}: {\small\begin{align*}
%%   &\Code{n}\hasoty{tNode}{\top} \sarr \Code{m}\hasoty{tNode}{\top}
%%   \sarr \Code{ld}\hasoty{tNode}{v = \Code{m}} \sarr
%%   \\&\rg{\globalA\evparenth{\top}}{\lastA\msgB{elect}{\I{src}\;\I{dest}\;\I{leader}}{
%%       \I{src} = \Code{n} \land \I{dest} = \Code{m} \land \I{leader} =
%%       \Code{m} }}{\finalA\msgB{beLeader}{\I{leader}}{\I{leader} =
%%       \Code{m}}}
%% \end{align*}}\noindent
%% This specification captures a useful property of a leader
%% election algorithm: when a node $\Code{m}$ receive a message that elects
%% ($\eff{elect}$) itself as leader ($\I{leader} = \Code{m}$) from
%% another node $\Code{n}$, then node $\Code{m}$ can announce itself to
%% be the leader ($\eff{beLeader}$).

\autoref{tab:evaluation} divides the results of our experiments into
four categories, separated by double bars. The first measures the
complexity of benchmarks with respect to the number of distinct
effectful operators (\#$\eff{op}$) and the number of qualifiers (i.e.,
symbolic events) used in \tyName{} specifications and the property
being tested.  For the benchmarks in~\autoref{tab:evaluation}, our
results show that writing underapproximate specifications require
anywhere from $6$ - $22$ different operators and include over 100
symbolic events in our most complex examples.  The second group of
columns characterizes \name{}'s ability to discover a violation
compared to the baseline (\textbf{Q2}).  To make the comparison
objective, we execute generators in the two systems 10K times, and
count the number of runs in which a fault was uncovered.  For example,
for \textsf{Stack}, using QCheck to generate random sequences of
$\eff{push}$ and $\eff{pop}$ effects, resulted in an erroneous
execution every 31.9 times on average; in contrast, every execution of
the generator(s) synthesized by \name\ manifested the bug.  Notably,
even for benchmarks in which there is a sophisticated customized
generator available\footnote{When these custom generators are
  available, we use them as the baseline in our experiments.} (e.g.,
\textsf{IFC} or \textsf{de Bruijn}), \tyName{}-guided synthesis leads
to significantly better outcomes.  We use MonkeyDB~\cite{MonkeyDB}, a
mock storage system that can be used to find violations of database
isolation properties, as the baseline for a number of the concurrency
benchmarks.  Although it is specialized for this particular
application class, and \name{} is not, evaluation results are
comparable, with \name{} yielding generators that identify weak
isolation violations modestly more effectively than the baseline.
Similar conclusions apply to the P benchmarks shown
in~\autoref{tab:evaluation2}.  For simple benchmarks, P's default
scheduler, which performs enumerative state exploration to construct
schedules, independently of the behaviors of the actors under test or
the target property, or handwritten ones which inject additional
actors to control input message generation and prevent uninteresting
message orderings (e.g. \textsf{Simplified2PC} or \textsf{Heartbeat}),
are effective and \name{}'s performance is comparable to these.  As in
the previous set of experiments, however, for more complex benchmarks
(e.g., \textsf{ChainReplication}, \textsf{Paxos}), even manual
specialization to generate interesting test sequences is significantly
less effective than \name{}'s trace-guided synthesis approach.
Moreover, as indicated by the \textsf{Raft} numbers, a random
exploration procedure, without manually engineered guidance, is too
na\"{i}ve to find a violating execution.  It is notable that here too,
\name{}'s performance on these P benchmarks is competitive, and in
several cases, superior to handcrafted P schedulers, even though
\name{} has no built-in knowledge of P's underlying programming model.

The last group of columns in both tables provide details on the cost
of our synthesis procedure. The first column presents total synthesis
time (t$_\text{total}$), which takes anywhere from $.6$ to $165$
seconds (\textbf{Q3}); synthesis time is roughly proportional to
benchmark complexity, as reflected in the \#$\eff{op}$ and \#qualifier
columns.  The last three columns additionally analyze the
characteristics of \name{}'s synthesis procedure with respect to the
number of events in the abstract refined trace (\#evt), the number of
SMT queries (\#SMT), as well as the number of refinement steps
(\#refine); recall that the synthesis procedure uses a refinement loop
to construct a more specialized underapproximation from \tyName{}
specifications that are consistent with the provided safety property.
In general, these numbers correlate with each other, and serve as
rough proxies for benchmark complexity.  They indicate that even for
the most challenging benchmarks, \name{} is able to synthesize an
effective generator in a reasonable amount of time.
\paragraph{Case study}
To demonstrate that \name{} can be effective in real-world scenarios,
we have applied it to \textsf{AnonReadAtomicity}, a distributed
transaction system in use at a major cloud provider, where
asynchronous reads are expected to execute atomically, as discussed
in~\autoref{sec:overview}. The property to be checked requires that if
there exists a key $\Code{k}$ updated within an active transaction,
any successful read response asking its value should return the value
last written to $\Code{k}$ made by that transaction.
%% We can express a
%% violation of this property as the following SRE: %
%% {\small
%%   \begin{align*}
%%     \allA &\msgA{write}{\Code{tid}\;\Code{k}\;\Code{x}\;\Code{OK}}(\anyA\setminus\msgA{write}{\Code{tid}\;\Code{k}\;\Code{x}\;\Code{OK}})^*
%%                 \msgA{readRsp}{\iota\;\Code{tid}\;\Code{k}\;\Code{x}\;\Code{OK}}\;\allA
%% \end{align*}}%
%% \SJ{Fix....}
%% \noindent where the field $\mathit{tid}$ represents a transaction id, while
%% other fields have the same meanings as in the example from
%%\autoref{sec:intro}.
Generating a fault-inducing scenario requires (a) initiating a new
transaction with transaction id $\Code{tid}$, (b) successfully
performing a write within that transaction, and then (c) subsequently
performing a read within $\Code{tid}$ that yields a different value
than the one last written.  As the example in~\ref{sec:overview}
implies, using \tyName{} to express the behavior of $\eff{readRsp}$
enables the synthesis of a test generator that can strategically
request a new transaction to initiate triggering the intended
violation.  A version of the benchmark in which this sequence
structure is enforced manually discovers the violation in 53
executions, while \name{}'s version can uncover the violation on all
of its executions.
\cbnewadding{ \paragraph{Robustness to Specification Quality} As
  discussed in \autoref{sec:overview}, \name{}'s effectiveness depends
  on the quality of user-supplied \tyName{} specifications.  To
  quantify \name{}'s robustness to specification quality, we conducted
  an additional study using variants of the benchmarks in
  Tables~\ref{tab:evaluation} and \ref{tab:evaluation2} that keep the
  SUT and the global property fixed while weakening the \tyName{}
  specifications of individual operators provided by the SUT.
  Specifications were weakened in a similar manner to the examples
  described in \autoref{sec:overview}: the variant of the IFC
  benchmark (\textsf{IFCLoad}$_{\mathsf{W}}$), for example, eliminates
  the ``valid address'' constraint (i.e., $\I{isAddr}$) used in the
  original specification. Broadly speaking, the weakened
  specifications can be categorized as either relaxing basic
  correctness properties of the SUT or removing search biases. Both
  categories of specifications are \emph{lower-quality} because they
  provide less guidance to \name{} and thus increase the risk that it
  will synthesize generators misaligned with the SUT's intended
  behavior.}

\cbnewadding{ \autoref{tab:evaluation} includes the results of this
  experiment for the three examples from \autoref{sec:overview}
  (\textsf{Transaction}$_{\mathsf{W}}$,
  \textsf{IFCLoad}$_{\mathsf{W}}$, and
  \textsf{DeBruijnHO}$_{\mathsf{W}}$).\footnote{The experimental
    results for the full set of benchmarks are provided in the
    supplemental material.}  Under these weaker specifications,
  average synthesis time (t$_\text{total}$) improves by $18.18\%$, and
  by $78.07\%$ in the extreme case (\textsc{ChainReplication}). This
  comes at the cost of decreased testing effectiveness, however:
  $26.9\%$ (7 out of 26) of the benchmarks now time out; the remaining
  benchmarks show an elevated retry rate (\#Num. Executions),
  $1.9\times$ on average, and $7.5\times$ in the worst case
  (\textsc{IFCLoad}$_{\mathsf{W}}$). In general, weakening basic
  correctness properties significantly impacts \name{}'s
  effectiveness, as the resulting generators regularly yield
  executions that do not align with the intended behavior of the SUT:
  \textsf{DeBruijnHO}$_{\mathsf{W}}$, for example, generates a large
  number of ``useless'' untyped STLC programs that are immediately
  discarded by the type-checker. Obtaining an informative witness from
  such a generator may thus require generating many more inputs, and
  can quickly exhaust the testing time budget. In contrast, weakening
  the search bias was comparatively benign in our experiments (e.g.,
  the synthesized generator for \textsf{IFCLoad}$_{\mathsf{W}}$ still
  finds EENI violations relatively quickly), and primarily manifested
  in an increased number of retries before a bug was found.  Of
  course, since weaker specifications can lead to shorter synthesis
  times, the loss of testing efficiency may sometimes be an acceptable
  trade-off.}
% \section{Discussion}\label{sec:discussion}

\cbnewadding{ %
  \paragraph{Discussion}
  We have anecdotal evidence that the cost to write \tyName{}s is
  relatively low; the \tyName{}s for the $\textsc{HashTable}$ and
  $\textsc{Courseware}$ benchmarks in \autoref{tab:evaluation} were
  written by an undergraduate student, and those for
  $\textsc{Smallbank}$ and $\textsc{Twitter}$ by two first-year PhD
  students. None of the students had prior exposure to the \tyName{}
  specification language or \name{}'s synthesis algorithm.  Nonetheless,
  each was able to independently encode their knowledge about the SUT
  as \tyName{}s, requiring only a few person-hours to complete all the
  benchmarks.}

\cbnewadding{ %
  One potential area for future work is to provide better
  developer-facing feedback when synthesis fails. Because our
  algorithm consists of a refinement loop that maintains a set of
  candidate generators (see \autoref{fig:syn-gen}), integrating a
  refinement and debugging phase as part of the synthesizer and test
  execution pipeline would be a natural extension.
}

\section{Related Work and Conclusions}
\label{sec:related}
\paragraph{\it Underapproximate Reasoning} A number of recently
proposed logical frameworks provide program reasoning principles based
on under-, rather than over-, approximate abstractions.  Incorrectness
Logic~\cite{OHP19} and related variants~\cite{LR+22,RBD+20,ZDS23}
present compositional proof rules that enable verification of
reachability (aka incorrectness) properties of first-order imperative
programs; their primary motivation is to enable formal reasoning of
(unsound) program analysis tools such as Infer~\cite{DFLO19}.  Ideas
related to these notions have also been incorporated into type
systems~\cite{RW24,CoverageType,JY15,LPR26} for functional languages,
but none of these efforts have considered their application to the
synthesis of effectful test generators of the kind proposed in this
paper.  For example~\cite{CoverageType} uses underapproximation to
typecheck that a PBT generator is \emph{complete} with respect to the
functional inputs it generates to the SUT, while ~\cite{LZ+25}
considers an enumerative repair strategy to patch incomplete
generators so that they become complete.  In contrast, in this paper,
we show how interpreting type specifications from the lens of
underapproximation can guide a synthesis procedure to construct
bespoke effectful test generators tailored to the property of
interest, informed by the behavior of the operators provided by the
SUT.  As a mechanism to characterize underapproximate behavior, types
identify \emph{sufficient} effectful constraints that any test input
sequence produced by the generator \emph{must} satisfy based on these
specifications.
\paragraph{{\it Specializing PBT Generators}}
There has been significant prior work to make PBT implementations more
effective, by providing tools that enable reasoning and customization
of effect traces produced by a test generator.  For example,
~\citet{CP+09} describes a customizable user-level scheduler that uses
Quviq QuickCheck's support for state machine models to implement
deterministic replayable test inputs to an Erlang SUT.
Concuerror~\cite{CGS13} systematically explores process interleavings
in a concurrent Erlang program via program instrumentation and
preemption bounding.  Quickstrom~\cite{Quickstrom} uses specifications
of possible actions expressed in its Specstrom specification language,
a dialect of LTL, to dynamically choose the next input that the
generator should provide for interactive applications.  In contrast,
\name{}'s type-guided synthesis strategy yields
generators influenced by both the global property of interest and the
local \tyName{} specifications of the operations used by the SUT.
Importantly, the synthesis procedure is guided by the history and
future SREs of provided specifications to determine a semantically
meaningful ordering of effectful actions.
%% TLA+~\cite{Lam02} is a specification language based on LTL for
%% modeling finite-state distributed systems; the correctness of these
%% specifications are verified using the TLC explicit-state model
%% checker. TLA+ and its associated tooling has had notable real-world
%% impact~\cite{NR+15}.  While \name{}'s use of LTL$_f$ specifications in
%% \PAT{}s is a point of commonality with TLA+, the integration of these
%% specifications within a refinement type system, their role in driving
%% a component-based synthesis procedure, and the top-down (TLA) vs.
%% bottom-up (\name{}) exploration mechanism, differentiates \name{}'s
%% motivation and design from TLA+ and TLC in obvious ways.
\paragraph{Type and Effect Systems}
Type and effect systems that target \emph{temporal} properties on the
sequences of effects that a program may \emph{produce} is a
well-studied subject. For example, \citet{SS+04} presents a type and
effect system for reasoning about the shape of histories (i.e., finite
traces) of events embedded in a program.  \citet{KT+14} present a type
and effect system that additionally supports verification properties
of infinite traces, specified as B\"{u}chi automata.  More recently,
\citet{Temporal-Verification} have considered how to support richer
control flow structures, e.g., delimited continuations, in such an
effect system.  Closest to our work are \emph{Hoare Automata Types}
(\textsc{HAT}s)~\cite{ZYDJ24}, which integrate symbolic finite
automata into a refinement type system. \textsc{HAT}s enable reasoning
about stateful sequential programs structured as a functional core
interacting with opaque effectful libraries.  \tyName{}s revise
\textsc{HAT}s by viewing them as underapproximate specifications -
every trace accepted by the future regex in a \tyName{} is
justified by a trace accepted by its history.

\paragraph{\textbf {Conclusions}}
This paper proposes \name, a trace-guided synthesizer of effectful PBT
generators.  Its key innovations are (i) the adaptation of SREs,
expressed as types, that define traces of effectful actions that occur
before and after the execution of an expression, as underapproximate
specifications, and (ii) the synthesis of bespoke test generators
derived from these specifications.  Experimental results on a wide
range of applications show that \name{} is significantly more
effective in identifying test input sequences that falsify a property
than the existing state-of-the-art.

% \bibliographystyle{ACM-Reference-Format}
% \bibliography{bibliography}

% \section*{Acknowledgements}


%\bibliographystyle{ACM-Reference-Format}
\bibliographystyle{plainnat}
% \bibliographystyle{plain}
% \bibliographystyle{ACM-Reference-Format}
% \citestyle{acmauthoryear}
\bibliography{bibliography}

\ifdefined\showappendix
\newpage
\appendix
\section{Outlines of Supplemental Materials}

The supplemental material is organized as follows. 
The detailed description of STLC example is provided in \autoref{sec:tech:details-overview}.
The complete set of rules for our operational semantics, basic typing, and declarative typing judgments are provided in \autoref{sec:tech:semantics}, \autoref{sec:tech:basic-typing}, and \autoref{sec:tech:typing}. The type denotation is presented in \autoref{sec:tech:denot}. Details of the auxiliary functions in our typing algorithm are given in \autoref{sec:tech:algo}. Proofs of the theorems in our paper are provided in \autoref{sec:tech:proof}. Finally, \autoref{sec:tech:bench-explain} presents the benchmark specification explanations, and \autoref{sec:tech:evaluation} offers detailed results of our experiments.
\section{Details in Overview}\label{sec:tech:details-overview}

We introduce how to use \tyName to specify the well-typedness of the De Bruijn STLC example introduced in the overview section.

\begin{figure}[th!]
  \begin{minted}[xleftmargin=5pt, numbersep=4pt, linenos = true,  escapeinside=??]{OCaml}
    type ty = Int | Arr of ty * ty
    type serialized_term = Const of int | Var of int | Abs of ty | EndAbs
                         | App | AppL | AppR
  \end{minted}
  \vspace*{-.1in}
  \caption{STLC terms in a serialized form.}
  \label{fig:stlc-syntax}
  \vspace*{-.1in}
\end{figure}

We explicitly define the syntax of the serialized STLC terms in \autoref{fig:stlc-syntax}. The definitions of constants and variables are standard, and variables
carry their de Bruijn index. However, the usual function abstraction
constructor ``$\Code{Abs\;\mykw{of}\;ty * term}$'' is now represented
by two operations, ``$\Code{Abs\;\mykw{of}\;ty}$'' and
``$\Code{EndAbs}$'', where the function body is the term that is
generated between them. Similarly, the typical application constructor
``$\Code{App\;\mykw{of}\;term * term}$'' is now captured by three
operations: ``$\Code{App}$'', ``$\Code{AppL}$'', and ``$\Code{AppR}$''.
Here, ``$\Code{App}$'' begins the application, ``$\Code{AppL}$'' marks the beginning of
the argument, and ``$\Code{AppR}$'' signals the end of the application.

An example of a serialized STLC term is shown below, where we color
the corresponding parts of the term and the serialized term for better
readability.
\begin{align*}
  \text{\textcolor{gray}{STLC term:}} \quad & \textcolor{blue}{(\lambda \Int}. \textcolor{orange}{[0]} \textcolor{blue}{)}\; \textcolor{DeepGreen}{3} \\
  \text{\textcolor{gray}{serialized STLC term:}} \quad & \eff{app}; \textcolor{blue}{\zevent{abs}{\Int}}; \textcolor{orange}{\zevent{var}{0}}; \textcolor{blue}{\eff{endAbs}}; \eff{appL}; \textcolor{DeepGreen}{\zevent{const}{3}}; \eff{appR}
\end{align*}
\noindent
We denote de Bruijn indices with ``$[...]$'' to distinguish them from
constants. The interpreter accepts a
serialized lambda term and evaluates it to a value. In the example
above, the input sequence should be reduced as $(\lambda \Int. [0])\;3
\hookrightarrow 3$.

\begin{figure}[t!]
  {\footnotesize
    \begin{align*}
      \eff{const} :~ &  x{:}\Int\sarr \uhat{\allA}{\Unit}{\msgA{const}{x}\msgAGray{ty}{\CInt}} \\
      \eff{appR} :~ & \uhat{\allA\msgC{app}\allA  \msgAGray{ty}{\Code{Arr}(t_1, t_2)}  \msgC{appL}  \allA  \msgAGray{ty}{t_1}  }{\Unit}{\msgC{appR}  \msgAGray{ty}{t_2}} \\
      \eff{endAbs} :~ & \uhat{\allA\msgAGray{depth}{d}\msgA{abs}{t_2}\msgAGray{depth}{d+1}\allA  \msgAGray{ty}{t_1}}{\Unit}{ \msgC{endAbs}\msgAGray{depth}{d}
        \msgAGray{putTy}{\Code{Arr}(t_1,t_2)}} \\
      \eff{var} :~ &  i{:}\urt{\Int}{\vnu \geq 0} \sarr \uhat{\underbrace{\allA\msgA{abs}{t}\msgAGray{depth}{d}(\anyA\setminus\msgC{endAbs})^*}_\text{type $t$ at the $d$ level is open}\underbrace{\msgAGray{depth}{d + i}}_\text{current level is $d + i$} }{\Unit}{\msgA{var}{i}\msgAGray{putTy}{t}  \allA}
    \end{align*}
  }
  \caption{\tyName{} specifications of STLC operations. } %
  \label{fig:stlc-spec}
\end{figure}

The types for STLC operations are shown in \autoref{fig:stlc-spec}. The type for $\eff{const}$ has no constraint on its prefix trace
($\allA$), but records a $\textcolor{gray}{\eff{ty}}$ event in
the trace that records the type of the output of the $\eff{const}$
expression as $\Code{Int}$. The arguments to a ghost event are either
serialized terms (like $\Code{Int}$) or ghost variables.
Type-checking a trace using ghost events involves solving constraints
over their arguments based on algebraic axiomatizations, examples of
which are given below. The history and future regex of the $\eff{appR}$ type in
\autoref{fig:stlc-spec} characterizes both the syntax and well-typing
constraints of an STLC function application. It requires that (a)
prior to evaluating a function argument ($\eff{appL}$), the type of
the abstraction is \Code{Arr(t$_1$,t$_2$)}, i.e., the abstraction
being applied is a function type ($\msgAGray{ty}{\Code{Arr}(t_1,t_2)}$); (b) after the argument is evaluated, its type is
  determined to be \Code{t$_1$} ($\msgAGray{ty}{t_1}$); and (c)
  the type of the application's result is \Code{t$_2$}
  ($\msgAGray{ty}{t_2}$).

The de Bruijn index constraint is captured via the $\eff{\textcolor{gray}{depth}}$ operation in the \tyName{}s of the $\eff{endAbs}$ and $\eff{var}$ operations. In addition to the well-typing constraint, the \tyName{} of the $\eff{endAbs}$ operation maintains the depth of nested function abstractions using $\eff{\textcolor{gray}{depth}}$. Specifically, its history regex obtains the current depth ($\msgAGray{depth}{d}$), increases it by one at the beginning of the function body, and restores it after the $\eff{endAbs}$ event.
On the other hand, a valid index $i$ in $\eff{var}$ must be derived from an open abstraction at some depth in the prefix trace. The history regex of the $\eff{var}$ operation precisely characterizes this constraint: the abstraction at depth $d$ must be open ($\allA\msgA{abs}{t}\msgAGray{depth}{d}(\anyA\setminus\msgC{endAbs})^*$), and the current depth must be $d + i$.

\section{Operational Semantics}\label{sec:tech:semantics}

The small-step operational semantics of our core language are shown in
\autoref{fig:semantics}.

\begin{figure}[t!]
    {\footnotesize
    \vspace*{-.15in}
    {\small\begin{flalign*}
      &\text{\textbf{Handler Semantics }}\ 
      \fbox{$\alpha \vDash \zevent{op}{\overline{c}} \Downarrow c$ \quad $\primop(\overline{c}) \Downarrow c$} &
      \text{\textbf{Operational Semantics }}\ 
      \fbox{$\steptr{\alpha}{e}{\alpha}{e}$}
    \end{flalign*}}
        \begin{prooftree}
        \mhypo{30mm}
        {
          $\alpha \vDash \zevent{op}{\overline{c}} \Downarrow c'$
        }
          \infer1[\textsc{\small StEfOp}]{
            \steptr{\alpha}{\eff{op}\;\overline{c}}{[\zevent{op}{\overline{c}\;c'}]}{c'}
          }
        \end{prooftree}
          \quad
          \mprooftr{35mm}{
            $\primop\ \overline{c} \Downarrow c'$
            }{StOp}
            {$\steptr{\alpha}{\primop\ \overline{c}}{[]}{c'}$}
            \quad
    \\[0.8em]
    \mprooftr{25mm}{
    $\phi\Downarrow \top$
    }{StAssume}{
    $\steptr{\alpha}{\sassume{\phi}}{[]}{()}$
    }
    \quad
    \mprooftr{25mm}{$\phi \Downarrow \top$}{StAssert}{
    $\steptr{\alpha}{\sassert{\phi}}{[]}{()}$
    }
    \quad
    \mprooftr{35mm}{$\phi \Downarrow \bot$}{StAssert}{
    $\steptr{\alpha}{\sassert{\phi}}{[]}{\Code{raise\;Error}}$
    }
    \\[.8em]
    \mprooftr{65mm}{
      $\steptr{\alpha}{e_1}{\alpha'}{e_1'}$
    }{StLetE1}{
    $\steptr{\alpha}{\zlet{y}{e_1}{e_2}}{\alpha'}{\zlet{y}{e_1'}{e_2}}$
    }\quad
    \mprooftr{35mm}{}{StChoice1}{
    $\steptr{\alpha}{(e_1 \oplus e_2)}{[]}{(e_1)}$
    }
    \\[0.8em]
    \mprooftr{55mm}{}{StLetE2}{
    $\steptr{\alpha}{\zlet{y}{v}{e_2}}{[]}{e_2[y\mapsto v]}$
    }
    \quad
    \mprooftr{35mm}{}{StChoice2}{
    $\steptr{\alpha}{(e_1 \oplus e_2)}{[]}{(e_2)}$
    }
    \\[0.8em]
    \mprooftr{40mm}{}{StAppLam}{
    $\steptr{\alpha}{\zlam{x}{s}{e_1}\ v_x}{[]}{e_1[x\mapsto v_x]}{e_2}$
    }
    \\[0.8em]
    \mprooftr{80mm}{}{StAppFix}{
    $\steptr{\alpha}{(\zfix{f}{s}{x}{s_x}{e_1})\ v_x}{[]}{(\zlam{f}{s}{e_1[x\mapsto v_x]}) \ (\zfix{f}{s}{x}{s_x}{e_1})}$
    }
}
\caption{Full Operational Semantics}
\label{fig:semantics}
\end{figure}

\newpage
\section{Basic Typing Rules}\label{sec:tech:basic-typing}

The basic typing rules of our core language and qualifiers are shown
in \autoref{fig:basic-type-rules}. We use an auxiliary function
$\Code{Ty}$ to provide a basic type for the primitives of our language,
e.g., constants, built-in operators, and data constructors.

\begin{figure}[h!]
 {\footnotesize
 {\small
\begin{flalign*}
 &\text{\textbf{Basic Typing }} & \fbox{$\Gamma \basicvdash e : s$ \quad $\Gamma \basicvdash \eff{op} : s$ \quad $\Gamma \basicvdash \primop : s$}
\end{flalign*}}
\\ \
\begin{prooftree}
\hypo{}
\infer1[\textsc{BtConst}]{
\Gamma \basicvdash c : \Code{Ty}(c)
}
\end{prooftree}
\quad
\begin{prooftree}
\hypo{ \Gamma(x) = s }
\infer1[\textsc{BtVar}]{
\Gamma \basicvdash x : s
}
\end{prooftree}
\quad
\begin{prooftree}
\hypo{ \Gamma, x{:}b \basicvdash e : s }
\infer1[\textsc{BtFun}]{
\Gamma \basicvdash \zlam{x}{b}{e} : b\sarr s
}
\end{prooftree}
\quad
\begin{prooftree}
  \hypo{ \Gamma, f{:}b\sarr s, x{:}b \basicvdash e : s }
  \infer1[\textsc{BtFix}]{
  \Gamma \basicvdash \zfix{f}{b\sarr s}{x}{b}{e} : b\sarr s
  }
\end{prooftree}
\\[0.8em]
\begin{prooftree}
  \hypo{\Code{Ty}(\eff{op}) = s}
  \infer1[\textsc{BtEfOp}]{
  \Gamma \basicvdash \eff{op} : s
  }
  \end{prooftree}
\quad
\begin{prooftree}
  \hypo{\Code{Ty}(\primop) = s}
  \infer1[\textsc{BtPrimOp}]{
  \Gamma \basicvdash \primop : s
  }
\end{prooftree}
\quad
\begin{prooftree}
  \hypo{\Gamma \basicvdash e_1 : s
  \quad \Gamma \basicvdash e_2 : s}
  \infer1[\textsc{BtChoice}]{
  \Gamma \basicvdash e_1 \oplus e_2 : s
  }
\end{prooftree}
\\[0.8em]
\begin{prooftree}
  \hypo{\Gamma \basicvdash \eff{op}: \overline{b_i}\sarr b \quad \forall i.\Gamma \basicvdash c_i : b_i}
  \infer1[\textsc{BtEfOpApp}]{
  \Gamma \basicvdash \eff{op}\;\overline{c_i} : b
  }
  \end{prooftree}
\quad
\begin{prooftree}
\hypo{\Gamma \basicvdash \primop: \overline{s_i}\sarr s \quad \forall i.\Gamma \basicvdash v_i : s_i}
\infer1[\textsc{BtPrimOpApp}]{
\Gamma \basicvdash \primop\,\overline{v_i} : s
}
\end{prooftree}
\\[0.8em]
\begin{prooftree}
\hypo{\Gamma \basicvdash \phi : \Code{bool}}
\infer1[\textsc{BtAssume}]{
\Gamma \basicvdash \sassume{\phi} : \Unit
}
\end{prooftree}
\quad
\begin{prooftree}
\hypo{\Gamma \basicvdash \phi : \Code{bool}}
\infer1[\textsc{BtAssert}]{
\Gamma \basicvdash \sassert{\phi} : \Unit
}
\end{prooftree}
\\[0.8em]
\begin{prooftree}
  \hypo{\Gamma\basicvdash e : s \quad \Gamma,y{:}s \basicvdash e' : s}
  \infer1[\textsc{BtLetE}]{
  \Gamma \basicvdash \zlet{y}{e}{e'} : s
  }
\end{prooftree}
  \quad
\begin{prooftree}
  \hypo{\Gamma \basicvdash v_1 : s\sarr s' \quad \Gamma \basicvdash v_2 : s}
  \infer1[\textsc{BtApp}]{
  \Gamma \basicvdash v_1\;v_2 : \Unit
  }
\end{prooftree}
}
{\footnotesize
 {\small
\begin{flalign*}
 &\text{\textbf{Basic Qualifier Typing }} & \fbox{$\Gamma \basicvdash \phi: s$}
\end{flalign*}}
\\ \
\begin{prooftree}
\hypo{\Code{Ty}(c) = s }
\infer1[\textsc{BtLitConst}]{
\Gamma \basicvdash c : s
}
\end{prooftree}
\quad
\begin{prooftree}
\hypo{ \Gamma(x) = s }
\infer1[\textsc{BtLitVar}]{
\Gamma \basicvdash x : s
}
\end{prooftree}
\quad
\begin{prooftree}
\hypo{ }
\infer1[\textsc{BtTop}]{
\Gamma \basicvdash \top : \Code{bool}
}
\end{prooftree}
\quad
\begin{prooftree}
\hypo{ }
\infer1[\textsc{BtBot}]{
\Gamma \basicvdash \bot : \Code{bool}
}
\end{prooftree}
\\ \ \\ \ \\
\begin{prooftree}
\hypo{\Code{Ty}(\primop) = \overline{s_i}{\sarr}s \quad \forall i. \Gamma \basicvdash l_i : s_i }
\infer1[\textsc{BtLitOp}]{
\Gamma \basicvdash \primop\,\overline{l_i} : s
}
\end{prooftree}
\quad
\begin{prooftree}
\hypo{
\Gamma \basicvdash \phi : \Code{bool} }
\infer1[\textsc{BtNeg}]{
\Gamma \basicvdash \neg \phi : \Code{bool}
}
\end{prooftree}
\\ \ \\ \ \\
\begin{prooftree}
\hypo{
\Gamma \basicvdash \phi_1 : \Code{bool} \quad
\Gamma \basicvdash \phi_2 : \Code{bool} }
\infer1[\textsc{BtAnd}]{
\Gamma \basicvdash \phi_1 \land \phi_2 : \Code{bool}
}
\end{prooftree}
\quad
\begin{prooftree}
\hypo{
\Gamma \basicvdash \phi_1 : \Code{bool} \quad
\Gamma \basicvdash \phi_2 : \Code{bool} }
\infer1[\textsc{BtOr}]{
\Gamma \basicvdash \phi_1 \lor \phi_2 : \Code{bool}
}
\end{prooftree}
\quad
\begin{prooftree}
\hypo{
\Gamma, x{:}b \basicvdash \phi : \Code{bool} }
\infer1[\textsc{BtForall}]{
\Gamma \basicvdash \forall x{:}b. \phi : \Code{bool}
}
\end{prooftree}
}
\caption{Basic Typing Rules and Qualifier Typing Rules}
    \label{fig:basic-type-rules}
\end{figure}

\newpage
\section{Declarative Typing Rules}\label{sec:tech:typing}

\begin{figure}[!t]
{\footnotesize
{\small\begin{flalign*}
 &\text{\textbf{Type Erasure}} & \fbox{$\eraserf{t} \quad \eraserf{\tau} \quad \eraserf{\Gamma}$}
\end{flalign*}}
\vspace{-1em}
\begin{alignat*}{4}
    \eraserf{\urt{b}{\phi}} &\doteq b \quad&
    \eraserf{x{:}t \sarr \tau} &\doteq \eraserf{t}\sarr\eraserf{\tau} \quad&
    \eraserf{x{:}b \garr t} &\doteq \eraserf{t} \quad&
    \eraserf{\uhat{H}{x{:}t}{F}} &\doteq \eraserf{t}
    \\\eraserf{ \tau_1 \interty \tau_2 } &\doteq \eraserf{\tau_1} \quad&
    \eraserf{\emptyset} &\doteq \emptyset \quad&
    \eraserf{x{:}t, \Gamma} &\doteq x{:}\eraserf{t}, \eraserf{\Gamma} \quad& \quad&
\end{alignat*}
{\small\begin{flalign*}
  &\text{\textbf{Ghost Event Erasure}} & \fbox{$\eraseGhost{\alpha}$}
 \end{flalign*}}
 \vspace{-1em}
 \begin{alignat*}{4}
     \eraseGhost{[\;]} &\doteq [\;] \quad&
     \eraseGhost{\zevent{op}{\overline{c}}\cons \alpha} &\doteq \zevent{op}{\overline{c}}\cons \eraseGhost{\alpha} \quad&
     \eraseGhost{\zeventGray{op}{\overline{c}}\cons \alpha} &\doteq \eraseGhost{\alpha} \quad&
     \eraseGhost{\alpha_1 \listconcat \alpha_2} &\doteq \eraseGhost{\alpha_1} \listconcat \eraseGhost{\alpha_2}
 \end{alignat*}
{\small
\begin{flalign*}
 &\text{\textbf{Well-formedness}} & \fbox{$\Gamma \wellfoundedvdash A \quad \Gamma \wellfoundedvdash \tau \quad \Gamma \wellfoundedvdash t$}
\end{flalign*}
}
\\ \
\begin{prooftree}
\hypo{
\S{Ty}(\eff{op}) = \overline{x_i{:}b_i}\sarr b \quad  \eraserf{\Gamma}, \overline{x_i{:}b_i}, \vnu{:}b \basicvdash \phi : \Code{bool}
}
\infer1[\textsc{WfEvent}]{
\Gamma \wellfoundedvdash \msgB{op}{\overline{x_i\;\vnu}}{\phi}
}
\end{prooftree}
\quad
\begin{prooftree}
\hypo{
\Gamma \wellfoundedvdash A_1 \quad \Gamma \wellfoundedvdash A_2
}
\infer1[\textsc{WfOr}]{
\Gamma \wellfoundedvdash A_1 \lor A_2
}
\end{prooftree}
\\[.8em]
 \begin{prooftree}
\hypo{
\Gamma \wellfoundedvdash A_1 \quad \Gamma \wellfoundedvdash A_2
}
\infer1[\textsc{WfConcat}]{
\Gamma \wellfoundedvdash A_1\seqA A_2
}
\end{prooftree}
\quad
\begin{prooftree}
\hypo{}
\infer1[\textsc{WfEmpty}]{
\Gamma \wellfoundedvdash \emptyset
}
\end{prooftree}
\quad
\begin{prooftree}
\hypo{}
\infer1[\textsc{WfEpsilon}]{
\Gamma \wellfoundedvdash \epsilon
}
\end{prooftree}
\quad
\begin{prooftree}
\hypo{\Gamma \wellfoundedvdash A}
\infer1[\textsc{WfStar}]{
\Gamma \wellfoundedvdash A^*
}
\end{prooftree}
\\[.8em]
\begin{prooftree}
\hypo{
\eraserf{\Gamma}, \vnu{:}b \basicvdash \phi : \Code{bool}
}
\infer1[\textsc{WfPBase}]{
\Gamma \wellfoundedvdash \urt{b}{\phi}
}
\end{prooftree}
\quad
\begin{prooftree}
\hypo{
\Gamma \wellfoundedvdash t_x
}
\hypo{
\Gamma, x{:}\eraserf{t_x} \wellfoundedvdash t
}
\infer2[\textsc{WfPArr}]{
\Gamma \wellfoundedvdash x{:}t_x \sarr t
}
\end{prooftree}
\\[.8em]
\begin{prooftree}
\hypo{
\parbox{50mm}{\center
$\Gamma,x{:}t \wellfoundedvdash H$ \quad
$\Gamma \wellfoundedvdash t$ \quad
$\Gamma,x{:}t \wellfoundedvdash F$
}
}
\infer1[\textsc{WfHF}]{
\Gamma \wellfoundedvdash \uhat{H}{x{:}t}{F}
}
\end{prooftree}
\quad
\begin{prooftree}
  \hypo{
    \Gamma, \effGray{op}{:}s \wellfoundedvdash t
  }
  \infer1[\textsc{WfGEvent}]{
  \Gamma \wellfoundedvdash \effGray{op}:s\garr t
  }
  \end{prooftree}
\\[.8em]
\begin{prooftree}
\hypo{
\Gamma \wellfoundedvdash \tau
}
\hypo{
\Gamma, x{:}\eraserf{t_x} \wellfoundedvdash \tau
}
\infer2[\textsc{WfEfArr}]{
\Gamma \wellfoundedvdash x{:}t_x \sarr \tau
}
\end{prooftree}
\quad
\begin{prooftree}
\hypo{
\Gamma \wellfoundedvdash \tau
}
\hypo{
\Gamma, x{:}b \wellfoundedvdash t
}
\infer2[\textsc{WfGArr}]{
\Gamma \wellfoundedvdash x{:}b \garr t
}
\end{prooftree}
\quad
\begin{prooftree}
\hypo{
\parbox{30mm}{\center
  $\Gamma \wellfoundedvdash \tau_1$\quad
  $\Gamma \wellfoundedvdash \tau_2$ \quad
  $\eraserf{\tau_1} = \eraserf{\tau_2}$
}
}
\infer1[\textsc{WFInter}]{
\Gamma \wellfoundedvdash \tau_1 \interty \tau_2
}
\end{prooftree}
}
\caption{Full set of well-formedness typing rules.}
\label{fig:full-wf-rules}
\end{figure}

The full set of rules for our auxiliary typing relations are shown in
\autoref{fig:full-wf-rules} and \autoref{fig:full-sub-rules}. The full
set of declarative typing rules are shown in
\autoref{fig:full-typing-rules}. We elide the basic typing relation
({\small$\emptyset \basicvdash e : s$}) in the premises of the rules
in \autoref{fig:full-typing-rules}; all of these rules assume any
terms they reference have a basic type.

\begin{figure}[!t]
{\footnotesize
{\small
\begin{flalign*}
&\text{\textbf{Automata Inclusion }} & \fbox{$\Gamma \vdash A \subseteq A$}\quad
&\text{\textbf{Subtyping }} & \fbox{$\Gamma \vdash t <: t \quad \Gamma \vdash \tau <: \tau$}
\end{flalign*}
}
\\ \
\begin{prooftree}
\hypo{
\forall \sigma_2 \in  \denotation{\Gamma}.\exists \sigma_1 \in \denotation{\Gamma}.\denot{\sigma_2(A_2)} \subseteq \denot{\sigma_1(A_1)}
}
\infer1[\textsc{SubAutomata}]{
\Gamma \vdash A_1 \subseteq A_2
}
\end{prooftree}
\quad
\begin{prooftree}
\hypo{
\parbox{50mm}{\center
  $\Gamma \vdash t_1 <: t_2$ \quad
$\Gamma,x{:}t_1 \vdash F_2 \subseteq F_1$ \quad
$\Gamma,x{:}t_1 \vdash H_1 \subseteq H_2$
}
}
\infer1[\textsc{SubHF}]{
\Gamma \vdash \uhat{H_1}{x{:}t_1}{F_1} <: \uhat{H_2}{x{:}t_2}{F_2}
}
\end{prooftree}
\\[.8em]
\begin{prooftree}
\hypo{}
\infer1[\textsc{SubIntLL}]{
\Gamma \vdash \tau_1 \interty \tau_2 <: \tau_1
}
\end{prooftree}
\quad
\begin{prooftree}
\hypo{}
\infer1[\textsc{SubIntLR}]{
\Gamma \vdash \tau_1 \interty \tau_2 <: \tau_2
}
\end{prooftree}
\quad
\begin{prooftree}
\hypo{
\parbox{33mm}{\center
  $\Gamma \vdash \tau <: \tau_1 $ \quad
  $\Gamma \vdash \tau <: \tau_2 $
}
}
\infer1[\textsc{SubIntR}]{
\Gamma \vdash \tau <: \tau_1 \interty \tau_2
}
\end{prooftree}
\\[.8em]
\begin{prooftree}
\hypo{
\parbox{30mm}{\center
$\Gamma \vdash t_2 <: t_1$ \quad
$\Gamma, x{:}t_2 \vdash \tau_1 <: \tau_2$
}
}
\infer1[\textsc{SubArr}]{
\Gamma \vdash x{:}t_1\sarr \tau_1 <: x{:}t_2\sarr \tau_2
}
\end{prooftree}
\quad
\begin{prooftree}
\hypo{
\Gamma, x{:}\urt{b}{\top} \vdash t_1 <: t_2
}
\infer1[\textsc{SubGRight}]{
\Gamma \vdash t_1 <: x{:}b\garr t_2
}
\end{prooftree}
\quad
\begin{prooftree}
\hypo{
\parbox{30mm}{\center
  $\Gamma \wellfoundedvdash t[x\mapsto v]$
}
}
\infer1[\textsc{SubG}]{
\Gamma \vdash x{:}b\garr t <: t[x\mapsto v]
}
\end{prooftree}
\\[.8em]
\begin{prooftree}
  \hypo{
    \Gamma \vdash t_1 <: t_2}
  \infer1[\textsc{SubGEvent}]{
  \Gamma \vdash \effGray{op}:s\garr t_1 <: \effGray{op}:s\garr t_2
  }
  \end{prooftree}
\quad
\begin{prooftree}
\hypo{
\forall \sigma_2 \in\denot{\Gamma}. \exists \sigma_1 \in\denot{\Gamma}. \sigma_1(\phi_2) \implies \sigma_2(\phi_2)
}
\infer1[\textsc{SubPBase}]{
\Gamma \vdash \urt{b}{\phi_1} <: \urt{b}{\phi_2}
}
\end{prooftree}
\\[.8em]
\begin{prooftree}
\hypo{
\parbox{45mm}{\center
$\Gamma \vdash t_2 <: t_1$ \quad
$\Gamma, x{:}t_2 \vdash \tau_1 <: \tau_2$
}
}
\infer1[\textsc{SubPArr}]{
\Gamma \vdash x{:}t_1\sarr \tau_1 <: x{:}t_2\sarr \tau_2
}
\end{prooftree}
}
\caption{Full set of subtyping rules.}
\label{fig:full-sub-rules}
\end{figure}

\begin{figure}[!t]
  {\footnotesize
  {\footnotesize
  \begin{flalign*}
   &\text{\textbf{Auxiliary Typing}} & & \fbox{$\Gamma \wfvdash \tau \quad \Gamma \vdash A \subseteq A \quad \Gamma \vdash \tau <: \tau$} &\text{\textbf{Typing}} &  & \fbox{$\Gamma\vdash \eff{op}: t \quad \Gamma\vdash v : t \quad \Gamma \vdash e: \tau$}
  \end{flalign*}
  }
  \\ \
  \mprooftr{20mm}
  {$\Gamma\vdash v : t$}
  {TRet}
  {$\Gamma\vdash v : \uhat{H}{t}{\epsilon}$}
  \quad
  \mprooftr{15mm}
  {$\Delta(\eff{op}) = t$}
  {TOpCtx}
  {$\Gamma\vdash \eff{op} : t$}
  \quad
  \mprooftr{58mm}
  {
  $\Gamma \vdash \eff{op} : \overline{x_i{:}t_i}\sarr \uhat{H}{t}{\msgA{op}{\overline{x_i}} \seqA F}$ \quad
  $\forall i. \;\Gamma \vdash v_i : t_i$}
  {TEffOp}
  {$\Gamma\vdash \eff{op}{\;\overline{v_i}} : \uhat{H}{t}{\msgA{op}{\overline{v_i}} \seqA F}$}
  \\[.8em]
  \mprooftr{35mm}
{
$\Gamma\vdash e_1 : \tau_1$ \quad $\Gamma\vdash e_2 : \tau_2$}
{TChoice}
{$\Gamma\vdash e_1 \oplus e_2 : \tau_1 \interty \tau_2$}
\quad
\mprooftr{66mm}{
  $\Gamma\vdash e_1 : \uhat{H}{x{:}t_x}{F_1}$ \quad
  $\Gamma, x{:}t_x\vdash e_2 : \uhat{H \seqA F_1}{t}{F_2}$}
{TLet}
{
  $\Gamma\vdash \zlet{x}{e_1}{e_2} : \uhat{H}{t}{F_1 \seqA F_2}$
}
\\[.8em]
\mprooftr{22mm}{
 $\Gamma, x{:}t_x\vdash e : \tau$
}{TFun}{
 $\Gamma\vdash \zzlam{x}{e} : x{:}t_x\sarr \tau$}
\
\mprooftr{47mm}{
}{TFixBase}{
 $\Gamma\vdash \zzfix{f}{x}{e} : \overline{x{:}s}\sarr\uhat{H}{\urt{b}{\bot}}{\emptyset}$}
\
\mprooftr{27mm}{
 $\Gamma\vdash \zzfix{f}{x}{e} : t$ \quad
 $\Gamma, f{:}t\vdash \zzlam{x}{e} : t'$
}{TFixInd}{
 $\Gamma\vdash \zzfix{f}{x}{e} : t'$}
 \\[.8em]
\mprooftr{25mm}
{$\eraserf{\Gamma} \basicvdash c : b$}
{TConst}
{$\Gamma \vdash c : \urt{b}{\vnu = c}$}
\quad
\mprooftr{25mm}
{$\eraserf{\Gamma} \basicvdash x : b$}
{TVar}
{$\Gamma \vdash x : \urt{b}{\vnu = x}$}
\quad
   \mprooftr{30mm}{
   $\Gamma v_1 : x{:}t_2 \sarr \tau$ \quad
   $\Gamma \vdash v_2 : t_2$
  }{TAppEff}{
   $\Gamma\vdash v_1 v_2 : \tau[x\mapsto v_2]$}
\\[.8em]
\mprooftr{52mm}
{
$\Gamma \vdash \primop : \overline{y{:}t_i}\sarr t$ \quad
$\forall i. \Gamma \vdash v_i : t_i$
}
{TOpApp}
{$\Gamma\vdash \primop\ \overline{v} : t$}
\quad
\mprooftr{17mm}
{$\Delta(\primop) = t$}
{TPureOp}
{$\Gamma \vdash \primop : t$}
\\[.8em]
\mprooftr{18mm}
{$\Gamma \vdash v : t$ \quad
$\Gamma \vdash t <: t'$}
{TPureSub}
{$\Gamma \vdash v : t'$}
\quad
\mprooftr{18mm}
{$\Gamma \vdash e : \tau$ \quad
$\Gamma \vdash \tau <: \tau'$}
{TSub}
{$\Gamma \vdash e : \tau'$}
\quad
\mprooftr{38mm}
{$\Gamma \vdash \eff{op} : \overline{x{:}t_x}\sarr\uhat{H}{t}{F}$ \quad
$\Gamma,\overline{x{:}t_x} \vdash H' \subseteq H$
\quad
$\Gamma,\overline{x{:}t_x} \not\vdash H' \subseteq \emptyset$}
{TOpHis}
{$\Gamma\vdash \eff{op} : \overline{x{:}t_x}\sarr\uhat{H'}{t}{F}$}
}
\caption{Full set of typing rules.}
\label{fig:full-typing-rules}
\end{figure}
\newpage
\section{Type Denotation}\label{sec:tech:denot}

\begin{figure}[t!]
\footnotesize{
{\small
\begin{flalign*}
 &\text{\textbf{Well-Formed Event and Trace}} & \fbox{$\wfvdash m \quad \wfvdash \alpha$}
\end{flalign*}
\begin{prooftree}
\hypo{
\emptyset \basicvdash \eff{op} : \overline{b_i}\sarr b \quad
\emptyset \basicvdash c' : b \quad
\forall i. \emptyset \basicvdash c_i : b_i
}
\infer1[\textsc{WfEvent}]{
    \wfvdash \zevent{op}{\overline{c_i}\;c'}
}
\end{prooftree}
\quad
\begin{prooftree}
\hypo{}
\infer1[\textsc{WfNil}]{
    \wfvdash []
}
\end{prooftree}
\quad
\begin{prooftree}
\hypo{
\wfvdash m \quad
\wfvdash \alpha
}
\infer1[\textsc{WfCons}]{
    \wfvdash m \cons \alpha
}
\end{prooftree}
\begin{flalign*}
 &\text{\textbf{Trace Language}} & \fbox{$\denot{A} \in \mathcal{P}(\alpha)$}
\end{flalign*}}
\vspace{-1em}
\begin{flalign*}
\denot{\emptyset} &\doteq \emptyset
    \quad&
\denot{\epsilon} &\doteq \{\; [] \;\}
\\
\denot{\msgB{op}{\overline{x}}{\phi}} &\doteq \{ [\zevent{op}{\overline{c}}] ~|~ \wfvdash [\zevent{op}{\overline{c}}] \land \phi\msubst{x}{c} \}
    \quad&
\denot{A \lorA A'} &\doteq \denot{A} \lor \denot{A'}
\\
\denot{A \seqA A'} &\doteq \{\alpha_1 \listconcat \alpha_2 ~|~ \alpha_1 \in \denot{A} \land \alpha_2 \in \denot{A'} \}
\quad&
\denot{A^*} &\doteq \{ \alpha_1 \listconcat \alpha_2 \listconcat \cdots \listconcat \alpha_n ~|~ \alpha_i \in \denot{A} \land n \geq 0 \}
\end{flalign*}
\vspace{-1.5em}
{\small
\begin{flalign*}
  &\text{\textbf{Type Denotation}}
  & \fbox{$\denot{t} \in \mathcal{P}(c)
    \quad \denot{\tau} \in \mathcal{P}(e)$}
\end{flalign*}}
\vspace{-1.5em}
\begin{align*}
  &\denot{ \urt{b}{\phi} } &&\doteq \{ e ~|~ \emptyset
                                       \basicvdash e : b \land
                                       \forall v. \phi[\nu\mapsto v] \impl \msteptr{[\;]}{e}{[\;]}{v}  \}
  \\ &\denot{ x{:}\urt{b}{\phi} \sarr t } &&\doteq \{ e ~|~ \emptyset \basicvdash e : b\sarr\eraserf{t} \land
                                           \forall v.\, \phi[\nu\mapsto v] \implies e\ v \in  \denot{ t[x \mapsto v ] } \}
  \\ &\denot{ x{:}\urt{b}{\phi}\sarr\tau } &&\doteq \{ e ~|~ \emptyset \basicvdash e : b\sarr\eraserf{\tau} \land
                                           \forall v.\, \phi[\nu\mapsto v] \implies e\ v \in  \denot{ \tau[x\mapsto v ] } \}
  \\ &\denot{ x{:}t_x \sarr t } &&\doteq \{ e ~|~ \emptyset \basicvdash e : \eraserf{t_x\sarr t} \land
                                           \forall v_x \in \denot{t_x}.\, e\ v_x \in  \denot{ t[x \mapsto v_x ] } \} \qquad \text{(if $t_x$ is a function type)}
  \\ &\denot{ x{:}t_x\sarr\tau } &&\doteq \{ e ~|~ \emptyset \basicvdash e : \eraserf{t_x\sarr\tau} \land
                                           \forall v_x \in \denot{t_x}.\, e\ v_x \in  \denot{ \tau[x \mapsto v_x ] } \} \qquad \text{(if $t_x$ is a function type)}
  \\ &\denot{ x{:}b\garr t } &&\doteq \{ e ~|~ 
    \emptyset \basicvdash e : \eraserf{t} \land 
    \forall c{:}b. e \in  \denot{ t[x\mapsto c ] } \}
  \\ &\denot{ \effGray{op}{:}s\garr t } &&\doteq \{ e ~|~ 
    \emptyset \basicvdash e : \eraserf{t} \land 
    e\in\denot{ t } \}
  \\[-.8em] &\denot{\uhat{H}{x{:}\urt{b}{\phi}}{F}} &&\doteq \{e ~|~ \emptyset \basicvdash e : b \land \forall v{:}b. \phi[x\mapsto v] \impl \forall \alpha_f \in \denot{F[x\mapsto v]}. \exists \alpha_h \in\denot{H[x\mapsto v]}. \msteptr{\eraseGhost{\alpha_h}}{e}{\eraseGhost{\alpha_f}}{v} \}
  \\[-.4em] &\denot{\uhat{H}{x{:}t}{F}} &&\doteq \{e ~|~ \emptyset \basicvdash e : \eraserf{t} \land \forall v\in\denot{t}. \forall \alpha_f \in \denot{F}. \exists \alpha_h \in\denot{H}. \msteptr{\eraseGhost{\alpha_h}}{e}{\eraseGhost{\alpha_f}}{v} \}
  \\ &\denot{ \tau_1 \interty \tau_2 } &&\doteq \denot{\tau_1} \cap \denot{\tau_2}
\end{align*}
\vspace{-1.5em}
{\small
\begin{flalign*}
  &\text{\textbf{Type Context Denotation}}
  & \fbox{$\denot{\Gamma} \in \mathcal{P}(\sigma)$}
\end{flalign*}}
\vspace{-1.2em}
\begin{flalign*}
    &\denot{ \emptyset } \doteq \{ \emptyset \} \qquad \denot{ x{:}\urt{b}{\phi} } \doteq  \{ \sigma[x\mapsto v] ~|~ \phi[\vnu \mapsto v], \sigma \in \denot{\Gamma[x\mapsto v]} \}
    &
    &\denot{ x{:}t, \Gamma } \doteq \denot{ \Gamma } \quad \text{(if $t$ is not a base type)}
\end{flalign*}
}
\vspace{-0.4cm}
\caption{Type denotations in $\DSL{}$}
\label{fig:type-denotation}
\vspace{-.4cm}
\end{figure}
\newpage
\section{Auxiliary Functions for Synthesis}\label{sec:tech:algo}

This section describes three auxiliary functions used for generator
synthesis. The first of these, $\normPlan$, converts a SRE formula
into a set of abstract traces. The second, $\termDerive$, generates test generators from a set of abstract traces. The third,
$\S{termTrace}$, generates a straightline program from a single refined abstract trace.

\begin{algorithm}[t!]
  \Procedure{$\normPlan(A) := $}{
  $A \leftarrow \Code{Minimize}(A)$\;
  \Match{$A$}{
    \lCase{$\emptyset$}{ \Return{$\emptyset$} }
    \lCase{$\epsilon$}{ \Return{$\{\epsilon\}$} }
    \lCase{$\msg{op}{\phi}$}{ \Return{$\{\msg{op}{\phi}^\bot\}$} }
    \lCase{$A_1 \seqA A_2$}{ \Return{$\{ \pi_1 \seqA \pi_2 ~|~ \pi_1 \in \normPlan(A_1) \land \pi_2 \in \normPlan(A_2) \}$}}
    \lCase{$A_1 \lorA A_2$}{ \Return{$\normPlan(A_1) \cup \normPlan(A_2)$}}
    \Case{$A^*$}{
      \Match{$\normPlan(A)$}{
        \lCase{$\emptyset$}{ \Return{$\emptyset$} }
        \lCase{$\bigcup \overline{\msg{op}{\phi}}$}{ \Return{$\{[\overline{\msg{op}{\phi}^\bot}]^*\}$} }
        \lCase{$\Pi$}{ \Return{$\{ \pi_1 \listconcat \pi_2 \listconcat \cdots \listconcat \pi_n ~|~ \pi_i \in \Pi \}$} }
      }
    }
  }
  }
  \caption{Abstract Trace Normalization}
  \label{algo:norm}
\end{algorithm}

\paragraph{Normalization}
The SRE normalization function $\normPlan{}$, shown in
\autoref{algo:norm}, first perform a standard symbolic automaton minimization step to remove intersections and complements, and then recursively translates an input SRE into a set of
abstract traces. It divides $A \lorA A'$ into multiple abstract
traces, and leaves the Kleene star over multiple singleton events
unchanged.

\begin{lemma}\label{lemma:norm-sound}[Normalization is sound]
  The normalized result has the same denotation as the input SRE, that
  is, for all SRE $A$ and set of traces $\{\overline{\pi}\}$, {\small
    $\denot{A} = \bigcup_{i} \denot{\pi_i}$}
\end{lemma}

\begin{algorithm}[t!]
    \Procedure{$\S{TermDerive}(C) := $}{
      $E \leftarrow \emptyset$\;
      \ForEach{$(\Gamma,\pi) \in C$} {
        $e_1 \leftarrow \S{DeriveTrace}(\Gamma, \pi)$\;
        $e_2 \leftarrow \S{SynRecursion}(\Gamma, \pi, \Code{UnrollingBound})$\;
        $E \leftarrow E \cup \{e_1, e_2\}$\;
      }
      \Return{$\bigoplus_{e \in E} e$}\;
    }
    \Procedure{$\S{DeriveTrace}(\Gamma, \pi) := $}{
      \Match{$\Gamma$}{
      \lCase{$[]$}{ \Return{$\S{DeriveTraceAux}(\pi)$} }
      \Case{$x{:}\urt{b}{\phi} \cons \Gamma'$}{
      \Return{$\zassume{\overline{x}}{\phi[\vnu \mapsto x]}{\S{DeriveTrace}(\Gamma', \pi)}$}\;
      }
    }
    }
    \Procedure{$\S{DeriveTraceAux}(\Pi) := $}{
      \Match{$\Pi$}{
        \lCase{$\epsilon$}{ \Return{$()$} }
        \lCase{$\Pi^* \seqA \pi'$}{ \Return{$\S{DeriveTraceAux}(\pi')$} }
        \Case{$\msgBGray{op}{\overline{x}}{\phi} \seqA \pi'$}{
          $\overline{x'} \leftarrow \Code{GetFreshNames}(\overline{x})$\;
          $\phi_1 \leftarrow \phi\msubst{x}{x'}$\;
          $e \leftarrow \S{DeriveTraceAux}(\pi')$\;
          \Return{$\zassume{\overline{x'}}{\phi_1}{e}$}
        }
        \Case{$\msgB{op}{\overline{x}\;y}{\phi} \seqA \pi'$}{
          $\overline{x'} \leftarrow \Code{GetFreshNames}(\overline{x})$\;
          $y' \leftarrow \Code{GetFreshName}(y)$\;
          $\phi_1 \leftarrow \exists y.\phi\msubst{x}{x'}$\;
          $\phi_2 \leftarrow \phi\msubst{x}{x'}[y\mapsto y']$\;
          $e \leftarrow \S{DeriveTraceAux}(\pi')$\;
          \Return{$\zassume{\overline{x'}}{\phi_1}{\zlet{y'}{\eff{op}\;\overline{x'}}{\sassert{\phi_2}; e}}$}
        }
      }
      }
    \caption{Term Derivation}
    \label{algo:term-derive}
  \end{algorithm}

  \paragraph{Term Derivation} The term derivation function
  \(\S{TermDerive}\) is shown in \autoref{algo:term-derive}. It first
  converts the input type context into \(\mykw{assume}\) statements
  that use the type qualifiers attached to each variable (line 5), and
  then derives the abstract trace with the help of the
  \(\S{DeriveTrace}\), which recursively transforms the abstract trace
  into a straightline generator program.  Any Kleene stars are dropped
  (line 9), ghost events are converted into \(\mykw{assume}\)
  expressions (line 14), and concrete events are converted into
  effectful operations (line 16-21).  Note that the arguments of the
  effectful operation are provided by $\mykw{assume}$ and the values
  they return are always checked by $\mykw{assert}$ (line 21).

\begin{algorithm}[t!]
  \Procedure{$\S{SynRecursion}(\Gamma, \Pi, \Code{UnrollingBound}) := $}{
     $(T, \{\pi_i\}) \leftarrow \S{MatchTemplate}(\Pi)$\;
     \If{$\forall i. 0 \leq i < \Code{UnrollingBound}. \exists \tau. \Gamma \vdash \Code{Unroll}(T, \{\pi_i\}, i) : \tau$}{
        $e_i \leftarrow \S{DeriveTrace}(\pi_i)$\;
        \Return{$T(\overline{e_i})$}
     }
  }
  \caption{Synthesis of Recursive Programs}
  \label{algo:syn-recursion}
\end{algorithm}

\paragraph{Recursion Synthesis} Kleene stars ($\Pi^*$) in abstract
traces are interpreted as recursion points. The synthesis algorithm
can synthesize a recursive program by expanding the Kleene star using
recursion. As shown in \autoref{algo:syn-recursion}, we employ a
template-based strategy, in which the algorithm matches the given
abstract trace with a template $T$, together with a set of abstract
traces $\{\pi_i\}$ that are matched from the orginal trace. As
dictated by \textsc{TFixInd}, the algorithm unrolls the recursive body
up to a bounded depth and invokes refinement loop (lines 3 in
\autoref{algo:syn-recursion}) to ensure that the unrolled abstract
traces are well-typed.  Finally, the synthesized program is
instantiated by the template with derived programs from these abstract
traces.

\begin{example}[Recursion Template]
  A representative template is a non-deterministic loop, whose form is
  $T \doteq \zlet{f}{\mykw{fix}\ f.\lambda ().\ () \oplus \boxed{e_1};
    f\;(); \boxed{e_2}}{\boxed{e_3};f\;();\boxed{e_4}}$, allowing
  $\Pi^*$ in the abstract trace to be matched. For instance, consider
  the abstract trace $\pi_\Code{stack}$ shown in
  Example~\ref{ex:refinement-process}: %
  \vspace{-.4cm} \par\nobreak {\small %
    \begin{align*}
      &
        \underbrace{\Pi_\Code{noI}^*\msgAGray{pushI}{m\;y}\msgAGray{popI}{m\;y}\Pi_\Code{noI}^*}_{e_3}\;\;
        \underbrace{\msgA{push}{x}\msgAGray{pushI}{n\;x}}_{e_1}\;\;
        \underbrace{\Pi_\Code{noI}^*}_\text{recursion}\;\;
        \underbrace{\msgA{pop}{x}\msgAGray{popI}{n\;x}}_{e_2}\;\;
        \underbrace{\Pi_\Code{noI}^*}_{e_4}
    \end{align*}}\noindent
  Here, $\pi_\Code{stack}$ is divided into five parts, corresponding
  to the four placeholders and the self-recursive call (i.e., $f\;()$)
  in the template. The synthesizer then re-refines the abstract trace for $i$-fold
  recursion unrolling after erasing all qualifiers from the abstract
  trace: %
  \vspace{-.4cm} \par\nobreak {\small %
    \begin{align*}
      \Pi_\Code{anyA}^*\msgCGray{pushI}\msgCGray{popI}&\;\underline{\epsilon}\;\Pi_\Code{anyA}^* \tag{0-fold unrolling}
      \\\Pi_\Code{anyA}^*\msgCGray{pushI}\msgCGray{popI}\Pi_\Code{anyA}^*\;\underline{\msgCGray{pushI}\msgC{push}}&\underline{\msgCGray{popI}\msgC{pop}}\;\Pi_\Code{anyA}^* \tag{1-fold unrolling}
      \\\Pi_\Code{anyA}^*\msgCGray{pushI}\msgCGray{popI}\Pi_\Code{anyA}^*\;\underline{\msgCGray{pushI}\msgC{push}\msgCGray{pushI}\msgC{push}}&\underline{\msgCGray{popI}\msgC{pop}\msgCGray{popI}\msgC{pop}}\;\Pi_\Code{anyA}^* \tag{2-fold unrolling}
    \end{align*}}\noindent
  This refinement ensures that multiple unfoldings remain well-typed,
  which disallows an incorrect recursion template instantiation, as
  shown below: %
  \vspace{-.4cm} \par\nobreak {\small %
    \begin{align*}
      & \underbrace{\Pi_\Code{noI}^*\msgAGray{pushI}{m\;y}\msgAGray{popI}{m\;y}\Pi_\Code{noI}^*\msgA{push}{x}\msgAGray{pushI}{n\;x}}_{e_3}\;\;
        \underbrace{\Pi_\Code{noI}^*}_{e_1}\;\;
        \underbrace{\Pi_\Code{noI}^*}_\text{recursion}\;\;
        \underbrace{\msgA{pop}{x}\msgAGray{popI}{n\;x}}_{e_2}\;\;
        \underbrace{\Pi_\Code{noI}^*}_{e_4}
    \end{align*}
  } %
  \noindent This instantiation will repeatedly perform $\eff{pop}$
  events without any $\eff{push}$ events; its two-fold recursion
  unrolling will therefore be rejected.  The final recursive program
  is assembled by composing synthesized programs from these traces to
  populate the template placeholders, thereby yielding the following
  program:

  \begin{minted}[fontsize = \footnotesize, escapeinside=??, xleftmargin=10pt, linenos]{ocaml}
let rec f () =
  () ?$\oplus$?
    let (x: int) = int_gen () in push x; (* e1 *)
    f ();
    let (z: int) = pop () in assert (x == z) (*e2 *)
in f ()
  \end{minted}
\end{example}

\newpage
\section{Proofs}\label{sec:tech:proof}

We omit the completely standard proof that basic typing
\(\basicvdash e : s\) is sound, assuming that all terms and qualifiers
in our typing rules and theorems are type-safe. Before presenting the
proof of the fundamental theorem and type soundness, we introduce
several useful lemmas.

\subsection{Lemmas}

\subsubsection{Denotations}

\begin{lemma}\label{lemma:lastA-denot}[Denotation of singleton modality] For all symbolic event $\msgB{op}{\overline{x_i}}{\phi}$ and values $\overline{v_i}$,
{\small\begin{align*}
    \phi\msubst{x_i}{v_i} \implies
    [\zevent{op}{\overline{v_i}}] \in \denot{\msgB{op}{\overline{x_i}}{\phi}}
\end{align*}}
\end{lemma}

\begin{lemma}\label{lemma:seqA-denot}[Denotation of concatenation] For all automata $A_1$ and $A_2$ and trace $\alpha$,
{\small\begin{align*}
    \alpha \in \denot{A_1 \seqA A_2} \iff (\exists \alpha_1\;\alpha_2. \alpha = \alpha_1 \listconcat \alpha_2 \land \alpha_1 \in \denot{A_1} \land \alpha_2 \in \denot{A_2})
\end{align*}}
\end{lemma}

% \begin{lemma}\label{lemma:choice-denot}[Denotation of choice] For all term $e_1$ and $e_2$ and \PAT $\tau$,
% {\small\begin{align*}
%     e_1 \in \denot{\tau} \land e_2 \in \denot{\tau} \impl e_1 \oplus e_2 \in \denot{\tau}
% \end{align*}}
% \end{lemma}

\begin{lemma}\label{lemma:pure-denot}[Denotation of pure computation] For all term $e_1$ and $e_2$ and \PAT $\tau$,
{\small\begin{align*}
    (\forall \alpha. \msteptr{\alpha}{e}{[]}{e'}) \impl
   e'\in \denot{\tau} \impl e \in \denot{\tau}
\end{align*}}
\end{lemma}

\subsubsection{Subtyping}

\begin{lemma}\label{lemma:pure-subtyping}[Pure Subtyping Soundness] For Given type context $\Gamma$ and well-formed pure refinement type $t_1$ and $t_2$:
  $\Gamma \vdash t_1 <: t_2 \implies \forall \sigma_2 \in \denot{\Gamma}.\exists \sigma_1 \in \denot{\Gamma}. \denot{\sigma_2(t_2)} \subseteq \denot{\sigma_1(t_1)}$
\end{lemma}

\begin{lemma}\label{lemma:inclusion}[Inclusion] For Given type context $\Gamma$ and well-formed regex $H_1$ and $H_2$:
  $\Gamma \vdash H_1 \subseteq H_2 \implies \forall \sigma_2 \in \denot{\Gamma}.\exists \sigma_1 \in \denot{\Gamma}. \denot{\sigma_2(H_2)} \subseteq \denot{\sigma_1(H_1)}$
\end{lemma}

\begin{lemma}\label{lemma:subtyping}[Subtyping Soundness] For Given type context $\Gamma$ and well-formed \PAT $\tau_1$ and $\tau_2$:
  $\Gamma \vdash \tau_1 <: \tau_2 \implies \forall \sigma_2 \in \denot{\Gamma}.\exists \sigma_1 \in \denot{\Gamma}. \denot{\sigma_2(\tau_2)} \subseteq \denot{\sigma_1(\tau_1)}$
\end{lemma}

\subsubsection{Substitution}

\begin{lemma}[Substitution Lemma]\label{lemma:subt} For Given type context $\Gamma$, variable $x$, well-formed pure refinement type $t$,\PAT $\tau$ and term $e$:
  $\Gamma, x{:}t \vdash e : \tau \impl \forall v. \Gamma \vdash v : t \impl \Gamma, x{:}t \vdash e[x\mapsto v] : \tau[x\mapsto v]$
\end{lemma}

\subsubsection{Type Contexts}

In the \autoref{sec:formal}, we defined the well-formedness of effectful operator typing context, we also define the well-formedness of pure operator typing context.

\begin{definition}[Well-formed pure operator typing context]
  \label{lemma:built-in-pure-typing}
  The operator typing context $\Delta$ is well-formed
  iff the semantics of every pure operator $\eff{op}$ is
  consistent with the type $\overline{x{:}b_{\I{x}}}\garr
  \overline{y{:}t_{\I{y}}}\sarr t$ provided by $\Delta$.
  \vspace{-.4cm} \par\nobreak {\small\begin{align*}
\overline{\forall
      x{:}b_{\I{x}}}.\; \overline{\forall y \in \denotation{t_{\I{y}}}}.\; \forall c \in \denotation{t}. c' \in \denotation{t_{\I{x}}\msubst{y}{c}} \impl \primop (c) \Downarrow c'
  \end{align*}} \vspace{-.4cm} \par\nobreak\noindent
\end{definition}

The typing relation is defined mutually recursively between values, terms and operators, that is:

\begin{theorem}\label{theorem:fundamental-full}[Pure Fundamental Theorem] For Given type context $\Gamma$, term $e$ and value $v$:
  $\Gamma \vdash v: t \implies \forall \sigma \in \denot{\Gamma}. \sigma(v) \in \denot{\sigma(t)}$ and
  $\Gamma \vdash e : \tau \implies \forall \sigma \in \denot{\Gamma}. \sigma(e) \in \denot{\sigma(\tau)}$ and 
  $\Gamma \vdash \eff{op} : t \implies \forall \sigma \in \denot{\Gamma}. \sigma(\lambda \overline{x}.\eff{op}\; \overline{x}) \in \denot{\sigma(t)}$ and 
  $\Gamma \vdash \primop : \overline{y{:}t_i}\sarr t \implies \forall \sigma \in \denot{\Gamma}. \sigma(\lambda \overline{x}.\primop\; \overline{x}) \in \denot{\sigma(t)}$.
\end{theorem}

\begin{proof} We proceed by induction over our type judgment $\Gamma \vdash e : \tau$ and $\Gamma \vdash v : t$, which has cases proved as following:
\begin{enumerate}[label=Case]
\item:
{\footnotesize
\mprooftr{25mm}
{$\eraserf{\Gamma} \basicvdash c : b$}
{TConst}
{$\Gamma \vdash c : \urt{b}{\vnu = c}$}
\ \\ \ \\}
Since the constant $c$ is closed, we have $\sigma(c) = c$ for any substitution $\sigma$. By the semantics of evaluation, $\msteptr{\alpha}{c}{[]}{c}$ holds trivially. Therefore, $c \in \denot{\urt{b}{\vnu = c}}$ directly by the definition of type denotation. This concludes the case.
\item:
{\footnotesize
\mprooftr{25mm}
{$\eraserf{\Gamma} \basicvdash x : b$}
{TVar}
{$\Gamma \vdash x : \urt{b}{\vnu = x}$}
\ \\ \ \\}
Since the variable $x$ is will typed under the context $\Gamma$ in the basic type system, thus there exists a pure refinement type $t$ such that $(x, t) \in \Gamma$. The denotation of $\Gamma$ will always contains a substitution $\sigma$ such that $\sigma = \sigma'[x\mapsto v]$ where $v \in \denot{\sigma(t)}$. Therefore, $\sigma(x) = v$, and $\sigma(x) \in \denot{\urt{b}{\vnu = v}}$. This concludes the case.
\item:
{\footnotesize
\mprooftr{22mm}{
 $\Gamma, x{:}t_x\vdash e : \tau$
}{TFun}{
 $\Gamma\vdash \zzlam{x}{e} : x{:}t_x\sarr \tau$}
\ \\ \ \\}
According to the assumption of this case and the hypothesis, we know:
{\footnotesize\begin{alignat}{2}
  \setcounter{equation}{0}
  &\Gamma, x{:}t_x\vdash e : \tau \ptag{assumption} \label{found:1-1} \\
  &\forall \sigma \in \denot{\Gamma, x{:}t_x}. \sigma(e) \in \denot{\sigma(\tau)} \ptag{induction hypothesis} \label{found:1-2} \\
  &\sigma =\sigma'[x\mapsto v] \quad \text{where} \quad \sigma' \denot{\Gamma} \land v \in \denot{\sigma'(t_x)} \ptag{denotation of type context and \ref{found:1-2}} \label{found:1-3} \\
  &\sigma'(e[x\mapsto v]) \in \denot{\sigma(\tau)} \ptag{by \ref{found:1-2} and \ref{found:1-3} and Lemma~\ref{lemma:subt}} \label{found:1-4} \\
  &\forall \alpha. \msteptr{\alpha}{\zzlam{x}{e}\;v}{[]}{e[x\mapsto v]} \ptag{by the semantics of evaluation} \label{found:1-5} \\
  &\sigma'(\zzlam{x}{e}\;v) \in \denot{\sigma(\tau)} \ptag{by \ref{found:1-4} and \ref{found:1-5} and Lemma~\ref{lemma:pure-denot}} \label{found:1-6}
  \end{alignat}}\noindent
Then, the conclusion is consistent with the denotation definition, thus the case is proved.
\item:
{\footnotesize
\mprooftr{47mm}{
}{TFixBase}{
 $\Gamma\vdash \zzfix{f}{x}{e} : \overline{x{:}s}\sarr\uhat{H}{\urt{b}{\bot}}{\emptyset}$}
\ \\ \ \\}
Any function can have a empty return type for free, thus the case is proved.
\item:
{\footnotesize
\mprooftr{27mm}{
 $\Gamma\vdash \zzfix{f}{x}{e} : t$ \quad
 $\Gamma, f{:}t\vdash \zzlam{x}{e} : t'$
}{TFixInd}{
 $\Gamma\vdash \zzfix{f}{x}{e} : t'$}
\ \\ \ \\}
We inline the recursion function $f$ within its body, then this case is reduced to the case of function typing.
\item:
{\footnotesize
\mprooftr{20mm}
{$\Gamma\vdash v : t$}
{TRet}
{$\Gamma\vdash v : \uhat{H}{t}{\epsilon}$}
\ \\ \ \\}
According to the assumption of this case and the hypothesis, we know:
{\footnotesize\begin{alignat}{2}
  \setcounter{equation}{0}
  &\Gamma\vdash v : t \ptag{assumption} \label{found:2-1} \\
  &\forall \sigma \in \denot{\Gamma}. \sigma(v) \in \denot{\sigma(t)} \ptag{induction hypothesis} \label{found:2-2} \\
  &\forall \alpha.\msteptr{\alpha}{v}{[]}{v} \ptag{by the semantics of evaluation} \label{found:2-3}
  \end{alignat}}\noindent
Then, the conclusion is consistent with the denotation definition, thus the case is proved.
\item:
{\footnotesize
\mprooftr{35mm}
{
$\Gamma\vdash e_1 : \tau_1$ \quad $\Gamma\vdash e_2 : \tau_2$}
{TChoice}
{$\Gamma\vdash e_1 \oplus e_2 : \tau_1 \interty \tau_2$}
\ \\ \ \\}
According to the assumption of this case and the hypothesis, we know:
{\footnotesize\begin{alignat}{2}
  \setcounter{equation}{0}
  &\Gamma\vdash e_1 : \tau_1 \ptag{assumption} \label{found:3-1} \\
  &\forall \sigma \in \denot{\Gamma}. \sigma(e_1) \in \denot{\sigma(\tau_1)} \ptag{induction hypothesis} \label{found:3-2} \\
  &\Gamma\vdash e_2 : \tau_2 \ptag{assumption} \label{found:3-3} \\
  &\forall \sigma \in \denot{\Gamma}. \sigma(e_2) \in \denot{\sigma(\tau_2)} \ptag{induction hypothesis} \label{found:3-4} \\
  &\forall \alpha.\msteptr{\alpha}{e_1 \oplus e_2}{[]}{e_1} \ptag{by the semantics of evaluation} \label{found:3-5} \\
  &\forall \alpha.\msteptr{\alpha}{e_1 \oplus e_2}{[]}{e_2} \ptag{by the semantics of evaluation} \label{found:3-6} \\
  &\sigma(e_1 \oplus e_2) \in \denot{\sigma(\tau_1)} \ptag{by \ref{found:3-2} and \ref{found:3-5} and Lemma~\ref{lemma:pure-denot}} \label{found:3-7} \\
  &\sigma(e_1 \oplus e_2) \in \denot{\sigma(\tau_2)} \ptag{by \ref{found:3-4} and \ref{found:3-6} and Lemma~\ref{lemma:pure-denot}} \label{found:3-8} \\
  &\sigma(e_1 \oplus e_2) \in \denot{\sigma(\tau_1 \interty \tau_2)} \ptag{by definition of denotation} \label{found:3-9}
  \end{alignat}}\noindent
Then, the conclusion is consistent with the denotation definition, thus the case is proved.
\item:
{\footnotesize
\mprooftr{18mm}
{$\Gamma \vdash v : t$ \quad
$\Gamma \vdash t <: t'$}
{TPureSub}
{$\Gamma \vdash v : t'$}
\ \\ \ \\}
According to the assumption of this case and the hypothesis, we know:
{\footnotesize\begin{alignat}{2}
  \setcounter{equation}{0}
  &\Gamma\vdash v : t \ptag{assumption} \label{found:4-1} \\
  &\forall \sigma \in \denot{\Gamma}. \sigma(v) \in \denot{\sigma(t)} \ptag{induction hypothesis} \label{found:4-2} \\
  &\Gamma\vdash t <: t' \ptag{assumption} \label{found:4-3} \\
  &\sigma(v) \in \denot{\sigma(t')} \ptag{by \ref{found:4-2} and \ref{found:4-3} and Lemma~\ref{lemma:pure-subtyping}} \label{found:4-4}
  \end{alignat}}\noindent
Then, the conclusion is consistent with the denotation definition, thus the case is proved.
\item:
{\footnotesize
\mprooftr{18mm}
{$\Gamma \vdash e : \tau$ \quad
$\Gamma \vdash \tau <: \tau'$}
{TSub}
{$\Gamma \vdash e : \tau'$}
\ \\ \ \\}
According to the assumption of this case and the hypothesis, we know:
{\footnotesize\begin{alignat}{2}
  \setcounter{equation}{0}
  &\Gamma\vdash e : \tau \ptag{assumption} \label{found:5-1} \\
  &\forall \sigma \in \denot{\Gamma}. \sigma(e) \in \denot{\sigma(\tau)} \ptag{induction hypothesis} \label{found:5-2} \\
  &\Gamma\vdash \tau <: \tau' \ptag{assumption} \label{found:5-3} \\
  &\sigma(e) \in \denot{\sigma(\tau')} \ptag{by \ref{found:5-2} and \ref{found:5-3} and Lemma~\ref{lemma:subtyping}} \label{found:5-4}
  \end{alignat}}\noindent
Then, the conclusion is consistent with the denotation definition, thus the case is proved.
\item:
{\footnotesize
\mprooftr{38mm}
{$\Gamma \vdash \eff{op} : \overline{x{:}t_x}\sarr\uhat{H}{t}{F}$ \quad
$\Gamma,\overline{x{:}t_x} \vdash H' \subseteq H$
\quad
$\Gamma,\overline{x{:}t_x} \not\vdash H' \subseteq \emptyset$}
{TOpHis}
{$\Gamma\vdash \eff{op} : \overline{x{:}t_x}\sarr\uhat{H'}{t}{F}$}
\ \\ \ \\}
According to the assumption of this case and the hypothesis, we know:
{\footnotesize\begin{alignat}{2}
  \setcounter{equation}{0}
  &\Gamma \vdash \eff{op} : \overline{x{:}t_x}\sarr\uhat{H}{t}{F} \ptag{assumption} \label{found:6-1} \\
  &\forall \sigma \in \denot{\Gamma}. \sigma(\lambda \overline{x}.\eff{op}\;\overline{x}) \in \denot{\sigma(\overline{x{:}t_x}\sarr\uhat{H}{t}{F})} \ptag{induction hypothesis} \label{found:6-2} \\
  &\Gamma,\overline{x{:}t_x} \vdash H' \subseteq H \ptag{assumption} \label{found:6-3} \\
  &\forall \sigma \in \denot{\Gamma}. \sigma(H') \subseteq \sigma(H) \ptag{by \ref{found:6-3} and Lemma~\ref{lemma:inclusion}} \label{found:6-4} \\
  &\sigma(\lambda \overline{x}.\eff{op}\;\overline{x}) \in \denot{\sigma(\overline{x{:}t_x}\sarr\uhat{H'}{t}{F})} \ptag{by \ref{found:6-2} and \ref{found:6-4} and definition~\ref{lemma:built-in-typing} } \label{found:6-4}
  \end{alignat}}\noindent
Then, the conclusion is consistent with the denotation definition, thus the case is proved.
\item:
{\footnotesize
\mprooftr{17mm}
{$\Delta(\primop) = t$}
{TPureOp}
{$\Gamma \vdash \primop : t$}
\ \\ \ \\}
Can be derived directly from the definition of well-formedness of pure operator typing context (Definition~\ref{lemma:built-in-pure-typing}).
\item:
{\footnotesize
\mprooftr{15mm}
{$\Delta(\eff{op}) = t$}
{TOpCtx}
{$\Gamma\vdash \eff{op} : t$}
\ \\ \ \\}
Can be derived directly from the definition of well-formedness of operator typing context (Definition~\ref{lemma:built-in-typing}).
\item:
{\footnotesize
\mprooftr{58mm}
  {
  $\Gamma \vdash \eff{op} : \overline{x_i{:}t_i}\sarr \uhat{H}{t}{\msgA{op}{\overline{x_i}} \seqA F}$ \quad
  $\forall i. \;\Gamma \vdash v_i : t_i$}
  {TEffOp}
  {$\Gamma\vdash \eff{op}{\;\overline{v_i}} : \uhat{H}{t}{\msgA{op}{\overline{v_i}} \seqA F}$}
\ \\ \ \\}
According to the assumption of this case and the hypothesis, we know:
{\footnotesize\begin{alignat}{2}
  \setcounter{equation}{0}
  &\Gamma \vdash \eff{op} : \overline{x_i{:}t_i}\sarr \uhat{H}{t}{\msgA{op}{\overline{x_i}} \seqA F} \ptag{assumption} \label{found:7-1} \\
  &\forall \sigma \in \denot{\Gamma}. \sigma(\lambda \overline{x}.\eff{op}\;\overline{x}) \in \denot{\sigma(\overline{x_i{:}t_i}\sarr \uhat{H}{t}{\msgA{op}{\overline{x_i}} \seqA F})} \ptag{induction hypothesis} \label{found:7-2} \\
  &\forall i. \;\Gamma \vdash v_i : t_i \ptag{assumption} \label{found:7-3} \\
  &\forall i.\forall \sigma \in \denot{\Gamma}. \sigma(v_i) \in \denot{\sigma(t_i)} \ptag{induction hypothesis} \label{found:7-4} \\
  &\forall \alpha. \msteptr{\alpha}{(\lambda \overline{x}.\eff{op}\;\overline{x}) \; \overline{v_i}}{[]}{\eff{op}\; \overline{v_i}}\ptag{by the semantics of evaluation} \label{found:7-5} \\
  &\sigma((\lambda \overline{x}.\eff{op}\;\overline{x}) \; \overline{v_i}) \in \denot{\sigma(\uhat{H}{t}{\msgA{op}{\overline{v_i}} \seqA F})} \ptag{by  \ref{found:7-5} and denotation definition} \label{found:7-6} \\
  &\sigma(\eff{op}\; \overline{v_i}) \in \denot{\sigma(\uhat{H}{t}{\msgA{op}{\overline{v_i}} \seqA F})} \ptag{by \ref{found:7-5} and \ref{found:7-6} and Lemma~\ref{lemma:pure-denot}} \label{found:7-7}
  \end{alignat}}\noindent
Then, the conclusion is consistent with the denotation definition, thus the case is proved.
\item:
{\footnotesize
\mprooftr{52mm}
{
$\Gamma \vdash \primop : \overline{y{:}t_i}\sarr t$ \quad
$\forall i. \Gamma \vdash v_i : t_i$
}
{TOpApp}
{$\Gamma\vdash \primop\ \overline{v} : t$}
\ \\ \ \\}
According to the assumption of this case and the hypothesis, we know:
{\footnotesize\begin{alignat}{2}
  \setcounter{equation}{0}
  &\Gamma \vdash \primop : \overline{y{:}t_i}\sarr t \ptag{assumption} \label{found:8-1} \\
  &\forall \sigma \in \denot{\Gamma}. \sigma(\lambda \overline{x}.\primop\;\overline{x}) \in \denot{\sigma(\overline{y{:}t_i}\sarr t)} \ptag{induction hypothesis} \label{found:8-2} \\
  &\forall i. \;\Gamma \vdash v_i : t_i \ptag{assumption} \label{found:8-3} \\
  &\forall i.\forall \sigma \in \denot{\Gamma}. \sigma(v_i) \in \denot{\sigma(t_i)} \ptag{induction hypothesis} \label{found:8-4} \\
  &\forall \alpha. \msteptr{\alpha}{(\lambda \overline{y}.\primop\;\overline{y}) \; \overline{v_i}}{[]}{\primop\; \overline{v_i}}\ptag{by the semantics of evaluation} \label{found:8-5} \\
  &\sigma((\lambda \overline{y}.\primop\;\overline{y}) \; \overline{v_i}) \in \denot{\sigma(t)} \ptag{by  \ref{found:8-5} and denotation definition} \label{found:8-6} \\
  &\sigma(\primop\; \overline{v_i}) \in \denot{\sigma(t)} \ptag{by \ref{found:8-5} and \ref{found:8-6} and Lemma~\ref{lemma:pure-denot}} \label{found:8-7}
  \end{alignat}}\noindent
Then, the conclusion is consistent with the denotation definition, thus the case is proved.
\item:
{\footnotesize
\mprooftr{30mm}{
  $\Gamma v_1 : x{:}t_2 \sarr \tau$ \quad
  $\Gamma \vdash v_2 : t_2$
 }{TAppEff}{
  $\Gamma\vdash v_1 v_2 : \tau[x\mapsto v_2]$}
\ \\ \ \\}
According to the assumption of this case and the hypothesis, we know:
{\footnotesize\begin{alignat}{2}
  \setcounter{equation}{0}
  &\Gamma \vdash v_1 : x{:}t_2 \sarr \tau \ptag{assumption} \label{found:9-1} \\
  &\forall \sigma \in \denot{\Gamma}. \sigma(v_1) \in \denot{\sigma(x{:}t_2 \sarr \tau)} \ptag{induction hypothesis} \label{found:9-2} \\
  &\Gamma \vdash v_2 : t_2 \ptag{assumption} \label{found:9-3} \\
  &\forall \sigma \in \denot{\Gamma}. \sigma(v_2) \in \denot{\sigma(t_2)} \ptag{induction hypothesis} \label{found:9-4} \\
  &\sigma(v_1 \; v_2) \in \denot{\sigma(\tau[x\mapsto v_2])} \ptag{by definition of denotation land Lemma~\ref{lemma:subt}} \label{found:9-6}
  \end{alignat}}\noindent
Then, the conclusion is consistent with the denotation definition, thus the case is proved.
\end{enumerate}
\end{proof}

\subsection{Type Soundness}

The type soundness can be proved by fundamental theorem directly.

\begin{corollary}[Type Soundness] A generator \(e\) that satisfies
  \(\emptyset \vdash e : \uhat{\epsilon}{\Unit}{F}\) must produce all
  traces accepted by $F$, i.e.,
  $\forall \alpha \in
  \denot{F}. \msteptr{[]}{e}{\eraseGhost{\alpha}}{()}$.
\end{corollary}
\begin{proof} According to the fundamental theorem, we have
{\footnotesize\begin{alignat}{2}
  \setcounter{equation}{0}
  &\emptyset \vdash e : \uhat{\epsilon}{\Unit}{F} \ptag{assumption} \label{sound:1} \\
  &\forall \sigma \in \denot{\emptyset}. \sigma(e) \in \denot{\sigma(\uhat{\epsilon}{\Unit}{F})} \ptag{by Theorem~\ref{theorem:fundamental-full} and \ref{sound:1}} \label{sound:2} \\
  &\forall \alpha_f \in \denot{F}. \exists \alpha_h \in \denot{\epsilon}. \msteptr{\alpha_h}{e}{\eraseGhost{\alpha_f}}{()} \ptag{by \ref{sound:2} and denotation definition} \label{sound:3} \\
  &\alpha_f \in \denot{\epsilon} \iff \alpha_f = [\;] \ptag{by  denotation definition} \label{sound:4} \\
  &\forall \alpha_f \in \denot{F}. \msteptr{\epsilon}{e}{\eraseGhost{\alpha_f}}{()} \ptag{by \ref{sound:3} and \ref{sound:4} and denotation definition} \label{sound:5}
  \end{alignat}}\noindent
The the soundness is proved.
\end{proof}

\subsection{Synthesis is Sound}

As discussed in \autoref{sec:synthesis}, our synthesis algorithm first
refines the input property to a set of realizable abstract
traces and then uses the \(\S{TermDerive}\) function to translate these
traces into a controller program. We introduce a concept of realizability of abstract traces which stand as a bridge between Algorithm~\ref{algo:refine} and Algorithm~\ref{algo:term-derive}.

\begin{definition}[Symbolic event consistent with trace]
  \label{def:real-event-def}
  A symbolic event $\msg{op}{\phi}$ in abstract trace $\pi$ (i.e.,
  $\pi = \pi_h \seqA \msg{op}{\phi} \seqA \pi \seqA \seqA \pi_f$, where $\pi$ is a sequence of ghost events) is
  \emph{consistent} with the trace,
  denoted as $\Gamma \vdash_{R} \msg{op}{\phi} \in \pi$, iff {\small\begin{align*}
    &\Gamma \vdash \eff{op} : \overline{x{:}t_x}\sarr \uhat{\pi_h}{x:t}{\msg{op}{\phi}\seqA\pi}
\end{align*}}\noindent
\end{definition}

\begin{definition}[Realizability]
  \label{def:real-trace-def}
  An abstract trace $\pi$ is realizable under type context $\Gamma$ iff all concrete events in $\pi$ are consistent with $\pi$ under context $\Gamma$, denoted as $\Gamma \vdash_{R} \pi$.
\end{definition}

We now prove that the refined event in Algorithm~\ref{algo:refine} is consistent with the trace.

\begin{lemma}[Refined event is consistent with trace]
  \label{lemma:real-event}
  The refined event in Algorithm~\ref{algo:refine} is consistent with the trace.
\end{lemma}
\begin{proof}
  Assume Algorithm~\ref{algo:refine} refines the event $\msg{op}{\phi}$ in the trace $\pi_h \seqA \msg{op}{\phi} \seqA \pi \seqA \seqA \pi_f$. We then have $\pi_h \subseteq H$ and $\msg{op}{\phi} \seqA \pi \subseteq F$ (line 4 and line 5 in Algorithm~\ref{algo:refine}), where $H$ and $F$ are the history and future regex of the \tyName{} of $\eff{op}$. According to the rule \textsc{TOpHis}, which allows for a smaller history regex, and the rule \textsc{TSub}, which allows for a smaller future regex, we have that $\msg{op}{\phi}$ is consistent with the trace.
\end{proof}

Next, we show that all results of Algorithm~\ref{algo:refine} are realizable traces.

\begin{lemma}[Refined traces are realizable]
  \label{lemma:real-trace}
Assume $(\Gamma, \pi)$ is a result of Algorithm~\ref{algo:refine}. Then $\pi$ is realizable under $\Gamma$, denoted as $\Gamma \vdash_{R} \pi$.
\end{lemma}
\begin{proof}
By contradiction. If $\pi$ is not realizable under $\Gamma$, then there exists a concrete event in $\pi$ that is not consistent with $\Gamma$. However, according to Lemma~\ref{lemma:real-event}, the refined event is consistent with the trace, which leads to a contradiction.
\end{proof}

Now we prove the derived program is type-safe if the refined trace is realizable.

\begin{lemma}[Result of $\S{TraceDerive}$ is type-safe]
  \label{lemma:trace-derive-sound}
  The program derived by $\S{TraceDerive}$ is type-safe if the input trace $\pi$ is realizable under the input type context $\Gamma$.
\end{lemma}
\begin{proof}
Assume the synthesized program is $e$.
Since $\S{TraceDerive}$ drops ghost events, we have $\forall \alpha.\, \msteptr{\alpha}{e}{[]}{()}$. Therefore, there exists a trace $\alpha'$ that fills the ghost events back into $\alpha$ and is also consistent with the abstract trace $\pi$:
\begin{align*}
  \exists \alpha'.\, \eraseGhost{\alpha'} = \alpha \land \exists \sigma \in \denot{\Gamma}.\, \alpha' \in \sigma(\pi)
\end{align*}
Then, we know $e \in \denot{\sigma{\uhat{\epsilon}{\Unit}{\pi}}}$. Thus, according to the fundamental theorem, the derived program is type-safe.
\end{proof}

\begin{theorem}[Derived program of $\S{TermDerive}$ is type-safe]
  \label{theorem:term-derive-sound}
  If the recursion template is type-safe, then the program derived by ($\S{TermDerive}$) is type-safe if all input traces are realizable under their corresponding input type contexts.
\end{theorem}
\begin{proof}
  Since every program derived by $\S{TraceDerive}$ is type-safe, and the recursion template guarantees that the template after filling the holes is also type-safe, each program derived by $\S{TermDerive}$ is type-safe.
  According to the rule \textsc{TChoice}, the union ($\oplus$) of all derived programs is also type-safe.
\end{proof}

\begin{theorem}[Synthesis is Sound]
  \label{theorem:algo-sound-ap}
  Every generator synthesized by \autoref{algo:top-syn} is type-safe
  with respect to our declarative typing system.
\end{theorem}
\begin{proof}
  According to Lemma~\ref{lemma:real-trace} and Lemma~\ref{lemma:trace-derive-sound}, all refined traces are realizable and all derived programs are type-safe. Then, by Theorem~\ref{theorem:term-derive-sound}, the synthesized program is type-safe.
\end{proof}

% \subfile{tech/4-proof-2}
% \subfile{tech/4-proof-3}
\newpage
\section{Benchmark Specifications Explanation}\label{sec:tech:bench-explain}

% Match paper SREGen / starA(anyA - ...) idiom: Kleene iteration over complement of a pattern set.
\newcommand{\SNegStar}[1]{{\bigl(\compA{#1}\bigr)^{\ast}}}
\newcommand{\SNegStarU}[2]{{\bigl(\anyA \setminus ({#1} \cup {#2})\bigr)^{\ast}}}
\newcommand{\SNegStarUU}[3]{{\bigl(\anyA \setminus ({#1} \cup {#2} \cup {#3})\bigr)^{\ast}}}
\newcommand{\SNegStarUUU}[4]{{\bigl(\anyA \setminus ({#1} \cup {#2} \cup {#3} \cup {#4})\bigr)^{\ast}}}
\newcommand{\SNegStarTri}[3]{{\bigl(\anyA \setminus ({#1} \cup {#2} \cup {#3})\bigr)^{\ast}}}
\newcommand{\SNegStarExcl}[1]{{\bigl(\anyA \setminus ({#1})\bigr)^{\ast}}}

This appendix lists the global safety property and the \tyName{}s for every benchmark used in our evaluation. Each of them consistent of a short description, the global property, \tyName{}, and the weakened \tyName{}.

\paragraph{Weakening strategy.}
We tie this appendix to \autoref{sec:overview} and to the weakening study in our evaluation: we fix the SUT and the global property, and relax \tyName{}s to measure how ``lower-quality'' specs affect \name{}. We use two weakening strategies:
\begin{itemize}
  \item weakening qualifiers directly: For a refinement type $\urt{b}{\phi}$ or a symbolic event $\msg{ev}{\phi}$, we replace the qualifier $\phi$ with a weaker $\phi'$ such that $\phi' \implies \phi$. Concretely, this can drop conjuncts or weaken predicates to $\top$.
  \item merge intersected types: In \tyName{}, we use intersection types to represent different situations an operation may face. For example, an initialized system (i.e., $H_1 = \allA \seqA \msgA{init} \seqA \allA$) and an uninitialized system (i.e., $H_2 = (\anyA \setminus \msgA{init})^*$) yield two intersected \tyName{}s with different history automata:
  \begin{align*}
    \uhat{H_1}{\urt{b}{\phi_1}}{F_1} \interty \uhat{H_2}{\urt{b}{\phi_2}}{F_2}
  \end{align*}\noindent
  One weakening step merges them into a single case:
  \begin{align*}
    \uhat{H_1 \lorA H_2}{\urt{b}{\phi_1 \lor \phi_2}}{F_1 \lorA F_2}
  \end{align*}\noindent
  This approach is sound because it additionally allows ``when $H_1$ holds, $F_2$ happens'' and ``when $H_2$ holds, $F_1$ happens'', compared to the original intersected \tyName{}s.
\end{itemize}
The weakened \tyName{}s are presented in the following sections according to these two patterns; thus, we omit most of the details of the weakening process.

\paragraph{Representation conventions.}
For synchronous benchmarks we often split one logical operation $\eff{ev}(\I{args},\I{return})$ into $\eff{evReq}(\I{args})$ and $\eff{evRsp}(\I{return})$, matching request/response APIs. Global properties are given in the form that the test generator is instructed to falsify (negated safety). Futures in each clause relate events across the trace and carry search bias; see \textsc{Stack} for a concrete instance.

\newpage
\subsubsection*{Stack}
Consider an effectful stack implemented with a fixed-size array. An incorrect implementation that fails to grow the array when the stack is full may drop $\eff{push}$ operations. To expose this bug, we ask whether traces of the following form raise an exception:
\begin{align*}
   \zevent{push}{0};\zevent{push}{4};\zevent{push}{7};\zevent
   {pop}{7};\zevent{pop}{4};\textcolor{red}{\zevent{pop}{\_} 
   \quad \text{Exception!}}
\end{align*}
\noindent
The implementation uses an underlying array of length $2$ and mistakenly raises an ``stack is empty'' exception even though element $0$ should still be on the stack. In this benchmark we split $\eff{pop}$ into $\eff{popReq}$ and $\eff{popRsp}$; $\eff{push}$ has no return value, so we only model $\eff{pushReq}$ and omit $\eff{pushRsp}$.
The global property characterizes popping ``away'' a pushed value \(y\) without observing \(\msg{popRsp}{\I{elem} = y}\), and then observing an empty stack---modeling loss of the last pushed element.
\vspace{-.35cm}\par\nobreak{\small
\begin{align*}
  & \allA \seqA \msg{pushReq}{\I{elem} = y} \seqA
    \SNegStar{\msg{popRsp}{\I{elem} = y}} \seqA
    \msg{isEmpRsp}{\mathit{isEmpty} }
\end{align*}}
\vspace{-.4cm}\par\nobreak\noindent
The \tyName{} comprises the following clauses:
\vspace{-.35cm}\par\nobreak{\small
\begin{align*}
  \eff{pushReq} ~:~{}
  &  y{:}\Int \sarr \effLR{\allA}\,\Code{unit}\,\effLR{\msg{pushReq}{\I{elem} = y} \seqA \allA \seqA \msg{popReq}{\true} \seqA \allA}
  \\
  \eff{initStackReq} ~:~{}
  & \effLR{\allA}\,\Code{unit}\,\effLR{\msg{initStackReq}{\true} \seqA \allA}
  \\
  \eff{popReq} ~:~{}
  & \effLR{\allA}\,\Code{unit}\,\effLR{\msg{popReq}{\true} \seqA \msg{popRsp}{\true} \seqA \allA}
  \\
  \eff{popRsp} ~:~{}
  & x{:}\Int \sarr \effLR{\allA}\,\urt{\Code{int}}{\I{elem} = x}\,\effLR{\SNegStar{\msg{pushReq}{\true}}}
  \\
  \eff{isEmpReq} ~:~{}
  & \effLR{\allA}\,\Code{unit}\,\effLR{\msg{isEmpReq}{\true} \seqA \msg{isEmpRsp}{\true} \seqA \allA}
  \\
  \eff{isEmpRsp} ~:~{}
  & z{:}\Bool \sarr \effLR{\allA}\,\urt{\Code{bool}}{\mathit{isEmpty} = z}\,\effLR{\msg{isEmpRsp}{\mathit{isEmpty} = z} \seqA \allA}
\end{align*}}
\vspace{-.4cm}\par\nobreak\noindent
The future automaton in the $\eff{pushReq}$ \tyName{} describes not only $\eff{pushReq}$ itself, but also obligates a later $\eff{popReq}$. That encodes a pairing constraint between $\eff{push}$ and $\eff{pop}$, and acts as a search bias for this benchmark. Similarly, the future of $\eff{popRsp}$ forbids any later $\eff{pushReq}$, so along such traces the stack depth changes monotonically. The weakened \tyName{} drops these biases, yielding specifications that give the synthesizer less guidance.

\newpage
\subsubsection*{Set} This benchmark shares a similar setting with \textsf{Stack}, where we have an effectful set implemented with a fixed-size array. An incorrect implementation that fails to grow the array when the set is full may drop $\eff{insert}$ operations. In this benchmark, we use a ghost event to explicitly track the size of the set. 
Then, the error is that, after initialization and insertion of \(y\), a \emph{mismatched} delete request (\(\neg(\I{elem}=y)\)) should not appear in a realistic violation scenario tied to size discipline; the global property targets a bad delete-generation path right after insert. Precisely, the global property is given by the following:
\vspace{-.35cm}\par\nobreak{\small
\begin{align*}
  & \msg{initSet}{\true} \seqA \allA \seqA \msg{insert}{\I{elem} = y} \seqA
    \allA \seqA \msg{delReq}{\neg(\I{elem} = y)} \seqA \allA
\end{align*}}
\vspace{-.4cm}\par\nobreak\noindent
The \tyName{} are shown below:
\vspace{-.35cm}\par\nobreak{\small
\begin{align*}
  \eff{initSet} ~:~{}
  & \effLR{\allA}\,\Code{unit}\,\effLR{\msg{initSet}{\true} \seqA \msgGray{size}{\mathit{sz} = 0} \seqA \allA}
  \\
  \eff{insert} ~:~{}
  & x{:}\Int \sarr \effLR{\SNegStar{\msg{insert}{\I{elem} \ge x}} \seqA \msgGray{size}{\mathit{sz} = i}}\,\Code{unit}\,
    \\&  \effLR{\msg{insert}{\I{elem} = x} \seqA \allA \seqA \msg{delReq}{\I{elem} = x} \seqA \allA}
  \\
  \eff{delReq} ~:~{}
  & x{:}\Int \sarr \effL{}\SNegStar{\msg{delReq}{\I{elem} = x}} \seqA \msg{insert}{\I{elem} = x} \seqA \SNegStar{\msg{delReq}{\I{elem} = x}} \seqA 
  \\& \quad \msgGray{size}{\mathit{sz} = i+1}\effR{}\,\Code{unit}\,
    \effLR{\msg{delReq}{\I{elem} = x} \seqA \msg{delRsp}{\true} \seqA \msgGray{size}{\mathit{sz} = i} \seqA \allA}
  \\
  \eff{delRsp} ~:~{}
  & r{:}\Bool \sarr \effLR{\allA}\,\urt{\Code{bool}}{\mathit{res} = r}\,\effLR{\allA}
  \\
  \eff{isEmpReq} ~:~{}
  & \effLR{\allA}\,\Code{unit}\,\effLR{\msg{isEmpReq}{\true} \seqA \msg{isEmpRsp}{\true} \seqA \allA}
  \\
  \effGray{size} ~:~{}
  & i{:}\Int \sarr \effLR{\allA}\,\urt{\Code{int}}{\mathit{sz} = i}\,\effLR{\allA}
  \\
  \eff{isEmpRsp} ~:~{}
  & b{:}\Bool \sarr \effLR{\allA}\,\urt{\Code{bool}}{\mathit{isEmpty} = b}\,\effLR{\msg{isEmpRsp}{\mathit{isEmpty} = b} \seqA \allA}
\end{align*}}
\vspace{-.4cm}\par\nobreak\noindent
The weakened \tyName{} drops much of the \(\eff{delReq}\) precondition chain: histories no longer enforce the fine-grained ``no repeated delete / size bookkeeping'' pattern, and a typical relaxation widens the precondition to \(\allA \seqA \msg{insert}{\I{elem} = x} \seqA \allA\).

\newpage
\subsubsection*{Filesystem}
Consider a path-tree API with create/delete/exists operations. Here deleting a non-empty directory reportedly succeeds, yet a child path still appears to exist---a path-metadata inconsistency our property surfaces. Thus, the global property is given by the following:
\vspace{-.35cm}\par\nobreak{\small
\begin{align*}
  & \allA \seqA \msg{delReq}{\I{path} = \mathit{parent}(\I{chp})} \seqA
    \msg{delRsp}{\mathit{succ} } \seqA \allA \seqA
  \\&
    \msg{extReq}{\I{path} = \I{chp}} \seqA
    \msg{extRsp}{\mathit{exists} } \seqA \allA
\end{align*}}
\vspace{-.4cm}\par\nobreak\noindent
The \tyName{} is defined as follows:
\vspace{-.35cm}\par\nobreak{\small
\begin{align*}
  \eff{initReq} ~:~{}
  & \effLR{\allA}\,\Code{unit}\,\effLR{\msgC{initReq}\seqA \SNegStar{\msg{initReq}{\true}}}
  \\
  \eff{crtReq} ~:~{}
  & p{:}\mathsf{Path} \garr c{:}\urt{\mathsf{Byte}}{\I{isRoot}(\mathit{parent}(p))} \sarr
  \\ & \quad
    \effLR{\SNegStar{\msg{crtReq}{\I{path}=\mathit{parent}(p) \land \I{isDir}(\I{content})}}}\,\Code{unit}\,
    \\& \effLR{\msgA{crtReq}{p\;c} \seqA \msg{crtRsp}{\neg\mathit{succ}} \seqA \allA} \interty
  \\
  & p{:}\mathsf{Path} \garr c{:}\urt{\mathsf{Byte}}{\neg\,\I{isRoot}(p)} \sarr
  \\ & \quad
    \effLR{\msgC{initReq}\seqA \SNegStar{\msgA{delReq}{p}} \seqA \msg{crtReq}{\I{path}=p} \seqA \msg{crtRsp}{\mathit{succ}} \seqA \SNegStar{\msgA{delReq}{p}}}\,\Code{unit}\,
  \\ & \quad
    \effLR{\msgA{crtReq}{p\;c} \seqA \msg{crtRsp}{\neg\mathit{succ}} \seqA \allA} \interty
  \\
  & p{:}\mathsf{Path} \garr c{:}\urt{\mathsf{Byte}}{\neg\,\I{isRoot}(\mathit{parent}(p))} \sarr
  \\ & \quad
    \effL{}\allA \seqA \msg{crtReq}{\I{path}=\mathit{parent}(p) \land \I{isDir}(\I{content})} \seqA \msg{crtRsp}{\mathit{succ}} \seqA 
  \\ & \quad \SNegStar{\msg{delReq}{\I{path}=\mathit{parent}(p)}}\effR{}\,\Code{unit}\,
    \effLR{\msgA{crtReq}{p\;c} \seqA \msg{crtRsp}{\mathit{succ}} \seqA \allA \seqA \msgA{delReq}{p} \seqA \allA} \interty
  \\
  & p{:}\mathsf{Path} \garr c{:}\urt{\mathsf{Byte}}{\I{isRoot}(\mathit{parent}(p))} \sarr
  \\ & \quad
    \effLR{\allA}\,\Code{unit}\,
    \effLR{\msgA{crtReq}{p\;c} \seqA \msg{crtRsp}{\mathit{succ}} \seqA \allA \seqA \msg{crtReq}{\I{path}=\mathit{parent}(p)} \seqA \allA \seqA \msgA{delReq}{p} \seqA \allA}
  \\
  \eff{crtRsp} ~:~{}
  & s{:}\Bool \sarr \effLR{\allA}\,\urt{\Code{bool}}{\mathit{succ}=s}\,\effLR{\allA}
  \\
  \eff{delReq} ~:~{}
  & p{:}\urt{\mathsf{Path}}{\I{isRoot}(p)} \sarr \effLR{\allA}\,\Code{unit}\,
    \effLR{\msgA{delReq}{p} \seqA \msg{delRsp}{\neg\mathit{succ}} \seqA \allA} \interty
  \\
  & p{:}\urt{\mathsf{Path}}{\neg\,\I{isRoot}(p)} \sarr
  \\ & \quad
    \effLR{\allA \seqA \msgA{delReq}{p} \seqA \SNegStar{\msg{crtReq}{\I{path}=p}}}\,\Code{unit}\,
    \effLR{\msgA{delReq}{p} \seqA \msg{delRsp}{\neg\mathit{succ}} \seqA \allA} \interty
  \\
  & p{:}\urt{\mathsf{Path}}{\neg\,\I{isRoot}(p)} \sarr
  \\ & \quad
    \effLR{\msgC{initReq}\seqA \SNegStar{\msgA{delReq}{p}} \seqA \msg{crtReq}{\I{path}=p} \seqA \msg{crtRsp}{\mathit{succ}} \seqA \SNegStar{\msgA{delReq}{p}}}\,\Code{unit}\,
  \\ & \quad
    \effLR{\msgA{delReq}{p} \seqA \msg{delRsp}{\mathit{succ}} \seqA \SNegStar{\msg{crtReq}{\true}}}
  \\
  \eff{delRsp} ~:~{}
  & s{:}\Bool \sarr \effLR{\allA}\,\urt{\Code{bool}}{\mathit{succ}=s}\,\effLR{\allA}
  \\
  \eff{extReq} ~:~{}
  & p{:}\urt{\mathsf{Path}}{\I{isRoot}(p)} \sarr \effLR{\allA}\,\Code{unit}\,
    \effLR{\msgA{extReq}{p} \seqA \msg{extRsp}{\mathit{exists}} \seqA \allA} \interty
  \\
  & p{:}\urt{\mathsf{Path}}{\neg\,\I{isRoot}(p)} \sarr
  \\ & \quad
    \effLR{\allA \seqA \msgA{delReq}{p} \seqA \SNegStar{\msg{crtReq}{\I{path}=p}}}\,\Code{unit}\,
    \effLR{\msgA{extReq}{p} \seqA \msg{extRsp}{\neg\mathit{exists}} \seqA \allA} \interty
  \\
  & p{:}\urt{\mathsf{Path}}{\neg\,\I{isRoot}(p)} \sarr
  \\ & \quad
    \effLR{\allA \seqA \msg{crtReq}{\I{path}=p} \seqA \msg{crtRsp}{\mathit{succ}} \seqA \SNegStar{\msgA{delReq}{p}}}\,\Code{unit}\,
    \effLR{\msgA{extReq}{p} \seqA \msg{extRsp}{\mathit{exists}} \seqA \allA}
  \\
  \eff{extRsp} ~:~{}
  & e{:}\Bool \sarr \effLR{\allA}\,\urt{\Code{bool}}{\mathit{exists}=e}\,\effLR{\allA}
\end{align*}}
\vspace{-.4cm}\par\nobreak\noindent
The weakened \tyName{} merges intersected types across the create/delete/exists clauses and drops or weakens predicates such as \(\I{isDir}(\I{content})\) on paths, paralleling the weakened study in the evaluation.

\newpage
\subsubsection*{Graph}
Consider a dynamic graph with fresh nodes and directed edges. Two edges \(x\!\to\!y\) and \(y\!\to\!z\) are added (with side conditions on \(\I{isStart}\)). A bug in the underlying program can only be detected on a connected graph, which is captured by the following global property:
\vspace{-.35cm}\par\nobreak{\small
\begin{align*}
  & \allA \seqA \msg{addE}{\I{st}=x \land \I{ed}=y \land \neg\,\I{isStart}(\I{st})} \seqA \allA \seqA
  \\&\quad \msg{addE}{\I{st}=y \land \I{ed}=z} \seqA \allA \seqA
    \msg{isConnReq}{\true} \seqA \msg{isConnRsp}{\mathit{isConn} }
\end{align*}}
\vspace{-.4cm}\par\nobreak\noindent
The \tyName{} comprises the following clauses:
\vspace{-.35cm}\par\nobreak{\small
\begin{align*}
  \eff{init} ~:~{}
  & \effLR{\allA}\,\Code{unit}\,\effLR{\SNegStar{\msg{init}{\true}}}
  \\
  \eff{newNReq} ~:~{}
  & \effLR{\msg{init}{\true}\seqA\allA}\,\Code{unit}\,
    \effLR{\msg{newNReq}{\true}\seqA\msg{newNRsp}{\true}\seqA\allA}
  \\
  \eff{newNRsp} ~:~{}
  & x{:}\Int \sarr \effLR{\SNegStar{\msg{newNRsp}{\true}}}\,\urt{\Code{int}}{\I{nid}=x \land \I{isStart}(x)}\,\effLR{\allA} \interty
  \\
  & x{:}\Int \sarr \effLR{\SNegStar{\msg{newNRsp}{\true}}\seqA\msg{newNRsp}{\neg(\I{nid}=x)}\seqA\SNegStar{\msg{newNRsp}{\I{nid}=x}}}\,\\ &\urt{\Code{int}}{\I{nid}=x \land \neg\,\I{isStart}(x)}\,\effLR{\allA\seqA\msg{addE}{\I{ed}=x}\seqA\allA}
  \\
  \eff{addE} ~:~{}
  & x{:}\Int \garr y{:}\urt{\Int}{x\neq y} \sarr \effLR{\allA\seqA\msg{newNRsp}{\I{nid}=x}\seqA\allA\seqA\msg{newNRsp}{\I{nid}=y}\seqA\allA}\,\Code{unit}\,
    \\ &\effLR{\msg{addE}{\I{st}=x\land\I{ed}=y\land\neg\,\I{isStart}(y)} \seqA \SNegStar{\msg{addE}{\I{st}=x\land\I{ed}=y}}}
  \\
  \eff{isConnReq} ~:~{}
  & \effLR{\allA}\,\Code{unit}\,\effLR{\msg{isConnReq}{\true}\seqA\msg{isConnRsp}{\true}\seqA\allA}
  \\
  \eff{isConnRsp} ~:~{}
  & v{:}\Bool \sarr \effLR{\allA}\,\urt{\Code{bool}}{\mathit{isConn}=v}\,\effLR{\msg{isConnRsp}{\mathit{isConn}=v}\seqA\allA}
\end{align*}}
\vspace{-.4cm}\par\nobreak\noindent
The \tyName{} requires a node is either the initial node, or connected from another node. The weakened \tyName{} merges two intersected types of \(\eff{newNRsp}\) and loosens \(\I{nid}/\I{isStart}\) constraints.

\newpage
\subsubsection*{NFA} This benchmark has a similar setting to \textsf{Graph} but requires the graph to be a well-formed NFA with at least three edges. Then, the global property is:
\vspace{-.35cm}\par\nobreak{\small
\begin{align*}
  & \allA \seqA \msg{addE}{\true} \seqA \allA \seqA \msg{addE}{\true} \seqA \allA \seqA
    \msg{addE}{\true} \seqA \allA \seqA
    \msg{isNFARsp}{\mathit{isNFA} }
\end{align*}}
\vspace{-.4cm}\par\nobreak\noindent
The \tyName{} is:
\vspace{-.35cm}\par\nobreak{\small
\begin{align*}
  \eff{init} ~:~{}
  & \effLR{\allA}\,\Code{unit}\,\effLR{\SNegStar{\msg{init}{\true}}}
  \\
  \eff{newNReq} ~:~{}
  & \effLR{\msg{init}{\true}\seqA\allA}\,\Code{unit}\,
    \effLR{\msg{newNReq}{\true}\seqA\msg{newNRsp}{\true}\seqA\allA} 
  \\
  \eff{newNRsp} ~:~{}
  & x{:}\Int \sarr \effLR{\SNegStar{\msg{newNRsp}{\true}}}\,\urt{\Code{int}}{\I{nid}=x}\,\effLR{\allA} \interty
  \\
  & x{:}\Int \sarr \effLR{\SNegStar{\msg{newNRsp}{\true}}\seqA\msg{newNRsp}{\neg(\I{nid}=x)}\seqA\SNegStar{\msg{newNRsp}{\I{nid}=x}}}\,  \\ & \urt{\Code{int}}{\I{nid}=x}\,\effLR{\allA\seqA\msg{addE}{\I{ed}=x}\seqA\allA}
  \\
  \eff{setInitN} ~:~{}
  & x{:}\Int \sarr \effLR{\allA\seqA\msg{newNRsp}{\I{nid}=x}\seqA\allA}\,\Code{unit}\,\effLR{\msg{setInitN}{\I{nid}=x}\seqA\SNegStar{\msg{setInitN}{\true}}}
  \\
  \eff{setFinalN} ~:~{}
  & x{:}\Int \sarr \effLR{\allA\seqA\msg{newNRsp}{\I{nid}=x}\seqA\allA}\,\Code{unit}\,\effLR{\msg{setFinalN}{\I{nid}=x}\seqA\SNegStar{\msg{setFinalN}{\true}}}
  \\
  \eff{addE} ~:~{}
  & x{:}\Int \garr c{:}\urt{\mathsf{char}}{\true} \garr y{:}\urt{\Int}{x\neq y} \sarr
  \\ & \quad
    \effLR{\allA\seqA\msg{newNRsp}{\I{nid}=x}\seqA\allA\seqA\msg{newNRsp}{\I{nid}=y}\seqA\allA\seqA\msg{addE}{\I{ed}=x}}\,\Code{unit}\,
  \\ & \quad
    \effLR{\msg{addE}{\I{st}=x\land\I{ch}=c\land\I{ed}=y}\seqA\SNegStar{\msg{addE}{\I{st}=x\land\I{ch}=c}}} \interty
  \\
  & x{:}\Int \garr c{:}\urt{\mathsf{char}}{\true} \garr y{:}\urt{\Int}{x\neq y} \sarr
  \\ & \quad
    \effLR{\allA\seqA\msg{setInitN}{\I{nid}=x}\seqA\allA\seqA\msg{newNRsp}{\I{nid}=y}\seqA\allA}\,\Code{unit}\,
  \\ & \quad
    \effLR{\msg{addE}{\I{st}=x\land\I{ch}=c\land\I{ed}=y}\seqA\SNegStar{\msg{addE}{\I{st}=x\land\I{ch}=c}}} \interty
  \\
  & x{:}\Int \garr c{:}\urt{\mathsf{char}}{\true} \garr y{:}\urt{\Int}{x\neq y} \sarr
  \\ & \quad
    \effLR{\allA\seqA\msg{newNRsp}{\I{nid}=y}\seqA\allA\seqA\msg{newNRsp}{\I{nid}=x}\seqA\allA\seqA\msg{addE}{\I{ed}=x}}\,\Code{unit}\,
  \\ & \quad
    \effLR{\msg{addE}{\I{st}=x\land\I{ch}=c\land\I{ed}=y}\seqA\SNegStar{\msg{addE}{\I{st}=x\land\I{ch}=c}}}
  \\
  \eff{isNFAReq} ~:~{}
  & \effLR{\allA}\,\Code{unit}\,\effLR{\msg{isNFAReq}{\true}\seqA\msg{isNFARsp}{\true}\seqA\allA}
  \\
  \eff{isNFARsp} ~:~{}
  & v{:}\Bool \sarr \effLR{\allA\seqA\msg{setFinalN}{\true}\seqA\allA}\,\urt{\Code{bool}}{\mathit{isNFA}=v}\,\effLR{\allA}
\end{align*}}
\vspace{-.4cm}\par\nobreak\noindent
The weakened \tyName{} merges intersected types of \(\eff{newNRsp}\) and \(\eff{addE}\), and relaxes ordering constraints on \(\msg{addE}{\cdot}\).

\newpage
\subsubsection*{IFC} As introduced in \autoref{sec:overview}, we would like to test the ENNI property, where we are interested in the implementation of the IFC machine over different operators; thus, the global properties vary the offending primitive (\(\eff{store}/\eff{add}/\eff{load}\)) across three goals.
Then, the global properties are
\vspace{-.35cm}\par\nobreak{\small
\begin{align*}
  & \allA \seqA \msg{store}{\true} \seqA \allA \seqA \msg{enniRsp}{\mathit{enni}}
  \\
  & \allA \seqA \msg{add}{\true} \seqA \allA \seqA \msg{enniRsp}{\mathit{enni}}
  \\
  & \allA \seqA \msg{load}{\true} \seqA \allA \seqA \msg{enniRsp}{\mathit{enni}}
\end{align*}}
\vspace{-.4cm}\par\nobreak\noindent
The \tyName{} is defined as (here \(\mathit{addr}\) abbreviates \(\I{isAddr}\)):
\vspace{-.35cm}\par\nobreak{\small
\begin{align*}
  \eff{push} ~:~{}
  & l{:}\Bool \garr x{:}\urt{\Int}{\true} \garr y{:}\urt{\Int}{\true} \sarr \effLR{\SNegStar{\msgGray{depth}{\true}}}\,\Code{unit}\,
  \\ & \quad
    \effLR{\msgB{push}{\I{low}\;\I{e_l}\;\I{e_r}}{\I{low}=l\land\I{e_l}=x\land\I{e_r}=y\land\mathit{addr}(\I{e_l})\land\mathit{addr}(\I{e_r})}\seqA\msgGray{depth}{\mathit{depth}=1}\seqA\allA} \interty
  \\
  & d{:}\Int \garr l{:}\Bool \garr x{:}\urt{\Int}{\true} \garr y{:}\urt{\Int}{\true} \sarr \effLR{\allA\seqA\msgGray{depth}{\mathit{depth}=d}}\,\Code{unit}\,
  \\ & \quad
    \effLR{\msgB{push}{\I{low}\;\I{e_l}\;\I{e_r}}{\I{low}=l\land\I{e_l}=x\land\I{e_r}=y\land\mathit{addr}(\I{e_l})\land\mathit{addr}(\I{e_r})}\seqA\msgGray{depth}{\mathit{depth}=d+1}\seqA\allA}
  \\
  \eff{pop} ~:~{}
  & d{:}\urt{\Int}{d>0} \sarr \effLR{\allA\seqA\msgGray{depth}{\mathit{depth}=d \land d>0}}\,\Code{unit}\,
    \effLR{\msg{pop}{\true}\seqA\msgGray{depth}{\mathit{depth}=d-1}\seqA\allA}
  \\
  \eff{load} ~:~{}
  & d{:}\urt{\Int}{d>0} \sarr \effLR{\allA\seqA\msg{store}{\true}\seqA\msg{push}{\true}\seqA\msgGray{depth}{\mathit{depth}=d \land d>0}}\,\Code{unit}\,
  \\ & \quad
    \effLR{\msg{load}{\true}\seqA\msgGray{depth}{\mathit{depth}=d}\seqA\allA\seqA\msg{store}{\true}\seqA\allA}
  \\
  \eff{store} ~:~{}
  & d{:}\urt{\Int}{d>1} \sarr \effLR{\allA\seqA\msg{push}{\true}\seqA\msgGray{depth}{\mathit{depth}=d \land d>1}}\,\Code{unit}\,
    \\& \effLR{\msg{store}{\true}\seqA\msgGray{depth}{\mathit{depth}=d-2}\seqA\allA} \interty
  \\
  & d{:}\urt{\Int}{d>1} \sarr \effLR{\allA\seqA\msgGray{depth}{\mathit{depth}=d \land d>1}}\,\Code{unit}\,
    \\ & \effLR{\msg{store}{\true}\seqA\msgGray{depth}{\mathit{depth}=d-2}\seqA\allA}
  \\
  \eff{add} ~:~{}
  & d{:}\urt{\Int}{d>1} \sarr \effLR{\allA\seqA\msgGray{depth}{\mathit{depth}=d \land d>1}}\,\Code{unit}\,
    \\& \effLR{\msg{add}{\true}\seqA\msgGray{depth}{\mathit{depth}=d-1}\seqA\allA}
  \\
  \eff{enniReq} ~:~{}
  & \effLR{\allA}\,\Code{unit}\,\effLR{\msg{enniReq}{\true}\seqA\msg{enniRsp}{\mathit{enni}}\seqA\allA}
  \\
  \eff{enniRsp} ~:~{}
  & v{:}\Bool \sarr \effLR{\allA}\,\urt{\Code{bool}}{\mathit{enni}=v}\,\effLR{\allA}
  \\
  \effGray{depth} ~:~{}
  & d{:}\Int \sarr \effLR{\allA}\,\urt{\Code{int}}{\mathit{depth}=d}\,\effLR{\allA}
\end{align*}}
\vspace{-.4cm}\par\nobreak\noindent
Here we use a ghost event $\effGray{depth}$ to record the current depth of the IFC stack.
The weakened \tyName{} relaxes the \(\I{low}/\I{e_l}/\I{e_r}\) coupling on \(\msgB{push}{\cdots}{\cdots}\) toward weaker implications and merges the intersected types of the two \(\eff{store}\) clauses under looser stack-depth premises.

\subsubsection*{DeBruijn STLC} The setting of this benchmark is already discussed in the beginning of this appendix (cf. \autoref{sec:tech:details-overview}). The weakened \tyName{} removes all type-related constraints to enable the generator to produce all untyped STLC programs.

\newpage

\subsubsection*{Hashtable}
Consider a string-keyed hash table with find/replace that is accessible to multiple threads; we want to detect the concurrency error. Concretely, there is a situation that, after \(\msg{repReq}{\I{key}=k,\I{vl}=v_2}\) with no further \(\eff{addReq}/\eff{repReq}\) on \(k\), a \(\eff{findRsp}\) may still disagree with \(v_2\)---the property captures that mismatch. Then, the global property is defined as:
\vspace{-.35cm}\par\nobreak{\small
\begin{align*}
  & \allA \seqA \msg{repReq}{\I{key}=k \land \I{vl}=v_2} \seqA
  \\&\quad \SNegStarU{\msg{addReq}{\I{key}=k}}{\msg{repReq}{\I{key}=k}} \seqA
    \msg{findRsp}{\I{key}=k \land \neg(\I{vl}=v_2)}
\end{align*}}
\vspace{-.4cm}\par\nobreak\noindent
The \tyName{} comprises the following clauses:
\vspace{-.35cm}\par\nobreak{\small
\begin{align*}
  \eff{addReq} ~:~{}
  & k{:}\Int,\; v{:}\Int \sarr \effLR{\msg{initReq}{\true}\seqA\allA}\,\Code{unit}\,\effLR{\msg{addReq}{\I{key}=k\land\I{vl}=v}\seqA\allA}
  \\
  \eff{findReq} ~:~{}
  & k{:}\Int \sarr 
  \effL{} \allA\seqA\msg{repReq}{\I{key}=k\land\I{vl}=v_1}\seqA 
    \\& \SNegStarUUU{\msg{rmReq}{\I{key}=k}}{\msg{clReq}{\true}}{\msg{repReq}{\I{key}=k}}{\msg{addReq}{\I{key}=k}} \effR{} 
    \\&\,\Code{unit}\, 
    \effLR{\msg{findReq}{\I{key}=k}\seqA\msg{findRsp}{\I{key}=k\land\I{vl}=v_1}\seqA\allA}
  \\
  \eff{repReq} ~:~{}
  & k{:}\Int,\; v{:}\Int \sarr \effLR{\allA\seqA\msg{addReq}{\I{key}=k}\seqA\SNegStar{\msg{rmReq}{\I{key}=k}}}\,\Code{unit}\,
  \\&\quad \effLR{\msg{repReq}{\I{key}=k\land\I{vl}=v}\seqA\allA}
\end{align*}}
\vspace{-.4cm}\par\nobreak\noindent
The weakened \tyName{} shortens \(\eff{findReq}\) / \(\eff{repReq}\) precondition chains substantially.

\subsubsection*{Transaction} As explained in the read atomicity example in \autoref{sec:overview}.

\newpage

\subsubsection*{DataBase Setting}
The four benchmarks below share the same assumption about the underlying database. We adopt the weak isolation levels of MonkeyDB~\cite{MonkeyDB}, which exposes four APIs:
\begin{align*}
  &\eff{begin} : \Unit \sarr \Code{tid}
  \\&\eff{put}: \Code{tid} \sarr \Code{key} \sarr \Code{value} \sarr \Unit
  \\&\eff{commit}: \Code{tid} \sarr \Code{cid}
  \\&\eff{get}: \Code{tid} \sarr \Code{key} \sarr \Code{tid} \times \Code{cid} \times \Code{value} 
\end{align*}
\noindent
Calling $\eff{begin}$ starts a new transaction, and the engine returns a fresh transaction identifier. Each of $\eff{put}$, $\eff{get}$, and $\eff{commit}$ takes that $\Code{tid}$ so the store can tell which transaction is running. The $\eff{put}$ operation updates the value stored at a key in the usual way. The $\eff{commit}$ operation closes the transaction named by the supplied id; the engine assigns a monotonic \emph{commit id} $\Code{cid}$ that records commit order. For $\eff{get}$, the result bundles the read value with the ($\Code{tid}$,$\Code{cid}$) pair of the writer whose commit produced it; tracking $\Code{cid}$ is what lets us state weak isolation. Under causal consistency (CC), for example, reads must be justified by a transaction that commits before the reader commits. Concretely, CC is captured by:
\vspace{-.35cm}\par\nobreak{\small
\begin{align*}
  \eff{begin} ~:~{}
  & i{:}\Int \sarr \effLR{\SNegStar{\msg{begin}{\I{tid}\geq i}}}\,\Code{unit}\,
    \effLR{\msg{begin}{\I{tid}=i}\seqA\allA}
  \\
  \eff{put} ~:~{}
  & i{:}\Int \sarr k{:}\Code{Key} \sarr z{:}\Code{Value} \sarr \effLR{\allA\seqA\msg{begin}{\I{tid}=i}\seqA\SNegStar{\msg{commit}{\I{tid}=i}}}
  \\&\quad \Code{unit}\,
    \effLR{\msg{put}{\I{tid}=i\land\I{key}=k\land\I{vl}=z}\seqA\allA}
  \\
  \eff{commit} ~:~{}
  & j{:}\Int \garr i{:}\Int \sarr \effLR{\SNegStar{\msg{commit}{\I{tid}=i \lor \I{cid}\geq j}}\seqA\msg{begin}{\I{tid}=i}\seqA\SNegStar{\msg{commit}{\I{tid}=i \lor \I{cid}\geq j}}}
  \\&\quad \Code{unit}\,
    \effLR{\msg{commit}{\I{tid}=i\land\I{cid}=j}\seqA\SNegStarU{\msg{put}{\I{tid}=i}}{\msg{get}{\I{tid}=i}}}
  \\&(\texttt{Causal Consistency:}) \quad 
  \\\eff{get} ~:~{}
  & p{:}\Int\garr q{:}\Int \garr z{:}\Code{Value} \garr i{:}\Int \sarr k{:}\Code{Key} \sarr {}
  \\ & \quad
    \effLR{\allA\seqA\msgA{put}{p\;k\;z}\seqA\SNegStar{\msg{put}{\I{tid}=i\land\I{key}=k}}\seqA\msgA{commit}{p\;q}\seqA\SNegStar{\msg{put}{\I{tid}=i\land\I{key}=k}}}
  \\ & \quad
  \Code{unit}\,\effLR{\msg{get}{\I{tid}=i\land\I{key}=k\land\I{prevTid}=p\land\I{prevCid}=q\land\I{vl}=z}\seqA\SNegStar{\msg{commit}{\I{tid}=i\land\I{cid} < p}}}
\end{align*}}
\vspace{-.4cm}\par\nobreak\noindent
These clauses force $\Code{tid}$ and $\Code{cid}$ to increase monotonically, allow $\eff{commit}$ only after the matching $\eff{begin}$, and pair commits with transactions. Weak isolation shows up chiefly in $\eff{get}$: its history automaton may justify a read from any earlier committed writer, not necessarily the latest. If the current transaction $T_i$ observes writer $T_p$ at commit id $q$, the future automaton blocks committing $T_i$ with any $\Code{cid}<q$, so $T_i$ cannot commit ahead of $T_p$. Read committed (RC) keeps the same surface API but weakens the \tyName{} for $\eff{get}$:
\vspace{-.35cm}\par\nobreak{\small
\begin{align*}
  &(\texttt{Read Committed:}) \quad 
  \\\eff{get} ~:~{}
  & p{:}\Int\garr q{:}\Int \garr z{:}\Code{Value} \garr i{:}\Int \sarr k{:}\Code{Key} \sarr {} 
  \\ & \quad
    \effLR{\allA\seqA\msgA{put}{p\;k\;z}\seqA\SNegStar{\msg{put}{\I{tid}=i\land\I{key}=k}}\seqA\msgA{commit}{p\;q}\seqA\SNegStar{\msg{commit}{\true}}}
  \\ & \quad
  \Code{unit}\,\effLR{\msg{get}{\I{tid}=i\land\I{key}=k\land\I{prevTid}=p\land\I{prevCid}=q\land\I{vl}=z}\seqA\allA}
\end{align*}}
\vspace{-.4cm}\par\nobreak\noindent
Here the history forces $T_p$ to be the latest committed writer on the key (nothing commits after $\msgA{commit}{p\;q}$). Unlike CC, the future adds no further restriction. In the four benchmarks below we take CC as the default \tyName{}; weakening yields RC.

\newpage
\subsubsection*{Shopping cart} This is a benchmark from MonkeyDB~\cite{MonkeyDB}, which is a simple shopping cart system. The system allows users to add and delete items to their cart, and the cart is stored in the database. The global property is committing \(\msg{put}{\I{key}=x,\I{vl}=y}\) should rule out a later \(\msg{get}{\cdot}\) that reads a stale list (\(\neg(\I{vl}=y)\)), that is
\vspace{-.35cm}\par\nobreak{\small
\begin{align*}
  & \allA \seqA \msg{put}{\I{key}=x \land \I{vl}=y} \seqA
    \SNegStar{\msg{put}{\I{key}=x}} \seqA
    \msg{get}{\I{key}=x \land \neg(\I{vl}=y)} \seqA \allA
\end{align*}}
\vspace{-.4cm}\par\nobreak\noindent
The \tyName{} comprises the following clauses (\(\mathit{emp}\), \(\mathit{cons}\), and \(\mathit{remove}\) are the axiomatized list operations):
\vspace{-.35cm}\par\nobreak{\small
\begin{align*}
  &(\texttt{Causal Consistency Specification is omitted}) \quad 
  \\
  \eff{addItemReq} ~:~{}
  & i{:}\Int\garr l{:}\List[\Int] \garr x{:}\Int \sarr y{:}\Int \sarr \effLR{\allA}\,\Code{unit}\,
  \\ & \quad
    \effL{}\msg{addItemReq}{\I{user}=x\land\I{item}=y}\seqA\msg{begin}{\I{tid}=i}\seqA \allA \seqA
  \\ & \quad \msg{get}{\I{tid}=i\land\I{key}=x\land\I{vl}=l}\seqA\allA\seqA\msg{put}{\I{tid}=i\land\I{key}=x\land\I{vl}=\mathit{cons}(y,l)}\seqA
  \\ & \quad \SNegStar{\msg{put}{\I{key}=x}}\seqA\msg{commit}{\I{tid}=i}\seqA\msg{addItemRsp}{\true}\seqA\allA\effLR{}
  \\
  \eff{addItemRsp} ~:~{}
  & \effLR{\allA}\,\Code{unit}\,\effLR{\allA}
  \\
  \eff{deleteItemReq} ~:~{}
  & {:}\Int\garr l{:}\List[\Int] \garr x{:}\Int \sarr y{:}\Int \sarr \effLR{\allA}\,\Code{unit}\,
  \\ & \quad
    \effL{}\msg{deleteItemReq}{\I{user}=x\land\I{item}=y}\seqA\msg{begin}{\I{tid}=i}\seqA \allA \seqA
  \\ & \quad \msg{get}{\I{tid}=i\land\I{key}=x\land\I{vl}=l}\seqA\allA\seqA\msg{put}{\I{tid}=i\land\I{key}=x\land\I{vl}=\mathit{remove}(y,l)}\seqA
  \\ & \quad  \SNegStar{\msg{put}{\I{key}=x}}\seqA\msg{commit}{\I{tid}=i}\seqA\msg{deleteItemRsp}{\true}\seqA\allA\effR{}
  \\
  \eff{deleteItemRsp} ~:~{}
  & \effLR{\allA}\,\Code{unit}\,\effLR{\allA}
\end{align*}}
\vspace{-.4cm}\par\nobreak\noindent
These clauses encode the underlying database steps used by $\eff{addItem}$ and $\eff{deleteItem}$: read the prior cart (a list of item ids) and write back an update. The $\allA$ gaps between those operations admit other trace events (e.g.\ $\eff{put}$ from concurrent transactions). They also forbid \emph{two uncommitted transactions from updating the same key}---the same ``write-write conflict'' that InnoDB-style engines reject. The constraint appears in the future as the pattern ``$\msg{put}{\I{key}=x \land \ldots}\SNegStar{\msg{put}{\I{key}=x}}\msg{commit}{\ldots}$'': after one transaction touches key $x$, no second transaction may write $x$ until the first commits. For the weakened \tyName{}, we switch the MonkeyDB layer to RC as above; the cart operations themselves are unchanged.

\newpage
\subsubsection*{Courseware} This is a benchmark from MonkeyDB~\cite{MonkeyDB}, where courseware is a system allow student to add and drop courses. The global property is the violation of strong consistency:
\vspace{-.35cm}\par\nobreak{\small
\begin{align*}
  & \allA \seqA \msg{put}{\I{key}=x \land \I{vl}=y} \seqA
    \SNegStar{\msg{put}{\I{key}=x}} \seqA
   \msg{get}{\I{key}=x \land \neg(\I{vl}=y)} \seqA \allA
\end{align*}}
\vspace{-.4cm}\par\nobreak\noindent
The \tyName{} is defined as follows (\(\mathit{emp}\), \(\mathit{cons}\), and \(\mathit{remove}\) are as in the Shopping cart benchmark):
\vspace{-.35cm}\par\nobreak{\small
\begin{align*}
  &(\texttt{Causal Consistency Specification is omitted}) \quad 
  \\
  \eff{regStuReq} ~:~{}
  & i{:}\Int \garr x{:}\Int \sarr \effLR{\SNegStarExcl{\msg{regStuReq}{\true} \cup \msg{get}{\I{key}=x} \cup \msg{put}{\I{key}=x}}}\,\Code{unit}\,
  \\ & \quad
    \effL{}\msg{regStuReq}{\I{user}=x}\seqA\msg{begin}{\I{tid}=i}\seqA \allA \seqA \msg{put}{\I{tid}=i\land\I{key}=x\land\mathit{emp}(\I{vl})} \seqA \allA \seqA
    \\& \quad \msg{put}{\I{tid}=i\land\I{key}=x\land\I{vl}}\seqA\msg{commit}{\I{tid}=i}\seqA \allA\seqA\msg{regStuRsp}{\true}\seqA
    \\& \quad \SNegStar{\msg{regStuReq}{\true}}\effR{}
  \\
  \eff{regStuRsp} ~:~{}
  & \effLR{\allA}\,\Code{unit}\,\effLR{\allA}
  \\
  \eff{crtCourseReq} ~:~{}
  & i{:}\Int,\; x{:}\Int \sarr \effLR{\allA}\,\Code{unit}\,
    \effL{}\msg{crtCourseReq}{\I{course\_id}=x}\seqA\msg{begin}{\I{tid}=i}\seqA\allA\seqA 
    \\& \quad \msg{putCourses}{\I{tid}=i\land\I{key}=x\land\I{vl}}\seqA\allA\seqA\msg{commit}{\I{tid}=i}\seqA\msg{crtCourseRsp}{\true}\seqA\allA\effR{}
  \\
  \eff{crtCourseRsp} ~:~{}
  & \effLR{\allA}\,\Code{unit}\,\effLR{\allA}
  \\
  \eff{enrStuReq} ~:~{}
  & i{:}\Int \garr l{:}\List[\Int] \garr x{:}\Int \sarr y{:}\Int  \sarr \effLR{\allA}\,\Code{unit}\,
  \\ & \quad
    \effL{}\msg{enrStuReq}{\I{user}=x\land\I{item}=y}\seqA\msg{begin}{\I{tid}=i}\seqA\allA\seqA
  \\ & \quad \msg{get}{\I{tid}=i\land\I{key}=x\land\I{vl}=l}\seqA\allA\seqA\msg{put}{\I{tid}=i\land\I{key}=x\land\I{vl}=\mathit{cons}(y,l)}\seqA
  \\ & \quad \SNegStar{\msg{put}{\I{key}=x}}\seqA\msg{commit}{\I{tid}=i}\seqA\msg{enrStuRsp}{\true}\seqA\allA\effR{}
  \\
  \eff{enrStuRsp} ~:~{}
  & \effLR{\allA}\,\Code{unit}\,\effLR{\allA}
  \\
  \eff{unenrStuReq} ~:~{}
  & i{:}\Int \garr l{:}\List[\Int] \garr x{:}\Int \sarr y{:}\Int  \sarr \effLR{\allA}\,\Code{unit}\,
  \\ & \quad
    \effL{}\msg{unenrStuReq}{\I{user}=x\land\I{item}=y}\seqA\msg{begin}{\I{tid}=i}\seqA\allA\seqA 
    \\ & \quad \msg{get}{\I{tid}=i\land\I{key}=x\land\I{vl}=l}\seqA\allA\seqA\msg{put}{\I{tid}=i\land\I{key}=x\land\I{vl}=\mathit{remove}(y,l)}\seqA
  \\ & \quad \SNegStar{\msg{put}{\I{key}=x}}\seqA\msg{commit}{\I{tid}=i}\seqA\msg{unenrStuRsp}{\true}\seqA\allA\effR{}
  \\
  \eff{unenrStuRsp} ~:~{}
  & \effLR{\allA}\,\Code{unit}\,\effLR{\allA}
\end{align*}}
\vspace{-.4cm}\par\nobreak\noindent
These clauses encode the underlying database steps used by courseware's API, which is similar to the \textsf{Shopping} benchmark. They also forbid \emph{two uncommitted transactions from updating the same key}---the same ``write-write conflict'' that InnoDB-style engines reject. For the weakened \tyName{}, we switch the MonkeyDB layer to RC as above; the cart operations themselves are unchanged.

\newpage
\subsubsection*{Twitter} This is a benchmark from MonkeyDB~\cite{MonkeyDB}, modeling social follow sets and timelines on top of the transactional store. The global property is a violation of strong consistency: after observing follows list $\ell$ for user \(u\), with no further \(\eff{put}\) between, a later \(\eff{get}\) may read a different set.
\vspace{-.35cm}\par\nobreak{\small
\begin{align*}
  & \allA \seqA \msg{get}{\I{user}=u \land \I{fls}=\ell} \seqA
    \SNegStar{\msg{put}{\I{user}=u}} \seqA
    \msg{get}{\I{user}=u \land \neg(\I{fls}=\ell)} \seqA \allA
\end{align*}}
\vspace{-.4cm}\par\nobreak\noindent
The \tyName{} is defined as follows (\(\mathit{emp}\), \(\mathit{cons}\), and \(\mathit{remove}\) are as in the Shopping cart benchmark):
\vspace{-.35cm}\par\nobreak{\small
\begin{align*}
  &(\texttt{Causal Consistency Specification is omitted}) \quad
  \\
  \eff{followReq} ~:~{}
  & i{:}\Int,\garr\ell{:}\List[\Int]\garr u{:}\Int \sarr f{:}\Int \sarr \effLR{\allA}\,\Code{unit}\,
  \\ & \quad
    \effL{}\msg{followReq}{\I{user}=u\land\I{follow\_o}=f}\seqA\msg{begin}{\I{tid}=i}\seqA\msg{get}{\I{tid}=i\land\I{user}=u\land\I{fls}=\ell}\seqA \allA \seqA
  \\ & \quad \msg{put}{\I{tid}=i\land\I{user}=u\land\I{fls}=\mathit{cons}(f,\ell)}\seqA\SNegStar{\msg{put}{\I{user}=u}}\seqA
  \\ & \quad \msg{commit}{\I{tid}=i}\seqA\msg{followRsp}{\true}\seqA\allA\effR{}
  \\
  \eff{followRsp} ~:~{}
  & \effLR{\allA}\,\Code{unit}\,\effLR{\allA}
  \\
  \eff{unfollowReq} ~:~{}
  & i{:}\Int,\garr\ell{:}\List[\Int]\garr u{:}\Int \sarr f{:}\Int \sarr \effLR{\allA}\,\Code{unit}\,
  \\ & \quad
    \effL{}\msg{unfollowReq}{\I{user}=u\land\I{unfollow\_o}=f}\seqA\msg{begin}{\I{tid}=i}\seqA
    \\ & \quad \msg{get}{\I{tid}=i\land\I{user}=u\land\I{fls}=\ell}\seqA \allA \seqA
  \\ & \quad \msg{put}{\I{tid}=i\land\I{user}=u\land\I{fls}=\mathit{remove}(f,\ell)}\seqA\SNegStar{\msg{put}{\I{user}=u}}\seqA
  \\ & \quad \msg{commit}{\I{tid}=i}\seqA\msg{unfollowRsp}{\true}\seqA\allA\effR{}
  \\
  \eff{unfollowRsp} ~:~{}
  & \effLR{\allA}\,\Code{unit}\,\effLR{\allA}
  \\
  \eff{postTweetReq} ~:~{}
  & i{:}\Int,\garr\ell{:}\List[\Int]\garr u{:}\Int \sarr f{:}\Int \sarr \effLR{\allA}\,\Code{unit}\,
  \\ & \quad
    \effL{}\msg{postTweetReq}{\I{user}=u\land\I{tweet}=\tau}\seqA\msg{begin}{\I{tid}=i}\seqA
    \\ & \quad\msg{getT}{\I{tid}=i\land\I{user}=u\land\I{tweets}=\ell}\seqA \allA \seqA
  \\ & \quad \msg{putT}{\I{tid}=i\land\I{user}=u\land\I{tweets}=\mathit{cons}(\tau,\ell)}\seqA\SNegStar{\msg{putT}{\I{user}=u}}\seqA
  \\ & \quad\msg{commit}{\I{tid}=i}\seqA\msg{postTweetRsp}{\true}\seqA\allA\effR{}
  \\
  \eff{postTweetRsp} ~:~{}
  & \effLR{\allA}\,\Code{unit}\,\effLR{\allA}
  \\
  \eff{timelineReq} ~:~{}
  & i{:}\Int,\garr\ell{:}\List[\Int]\garr u{:}\Int \sarr \effLR{\allA}\,\Code{unit}\,
  \\ & \quad
    \effL{}\msg{timelineReq}{\I{user}=u}\seqA\msg{begin}{\I{tid}=i}\seqA\msg{getT}{\I{tid}=i\land\I{user}=u\land\I{tweets}=\ell}\seqA
    \\ & \quad\msg{timelineRsp}{\I{tweets}=\ell}\seqA\allA\effR{}
  \\
  \eff{timelineRsp} ~:~{}
  & \effLR{\allA}\,\Code{unit}\,\effLR{\allA}
\end{align*}}
\vspace{-.4cm}\par\nobreak\noindent
These specifications mirror the Shopping-style pattern: each API issues \(\eff{begin}/\eff{commit}\) around reads and updates of the underlying rows, with \(\SNegStar{\msg{put}{\I{user}=u}}\) / \(\SNegStar{\msg{putT}{\I{user}=u}}\) forbidding write--write conflicts on those projections. For the weakened \tyName{}, we switch the MonkeyDB layer to RC as above; the Twitter operations themselves are unchanged.

\newpage
\subsubsection*{SmallBank} This is a benchmark from MonkeyDB~\cite{MonkeyDB}, implementing the SmallBank workload (accounts, checking, and savings) on the transactional store. The global property is a violation of strong consistency: two reads of checking balance \(b\) for customer \(c\) with no \(\msg{putC}{\I{cm}=c}\) in between should agree, yet a later read may disagree.
\vspace{-.35cm}\par\nobreak{\small
\begin{align*}
  & \allA \seqA \msg{getC}{\I{cm}=c \land \I{bal}=b} \seqA
    \SNegStar{\msg{putC}{\I{cm}=c}} \seqA
    \msg{getC}{\I{cm}=c \land \neg(\I{bal}=b)} \seqA \allA
\end{align*}}
\vspace{-.4cm}\par\nobreak\noindent
The \tyName{} is defined as follows:
\vspace{-.35cm}\par\nobreak{\small
\begin{align*}
  &(\texttt{Causal Consistency Specification is omitted}) \quad
  \\
  \eff{openAReq} ~:~{}
  & i{:}\Int,\; n{:}\Int,\; c{:}\Int \sarr \effLR{\SNegStarExcl{\msg{openAReq}{\I{name}=n} \cup \msg{openAReq}{\I{cm}=c}}}\,\Code{unit}\,
  \\ & \quad
    \effL{}\msg{openAReq}{\I{name}=n\land\I{cm}=c}\seqA\msg{begin}{\I{tid}=i}\seqA\msg{putA}{\I{tid}=i\land\I{name}=n\land\I{cm}=c}\seqA
  \\ & \quad \msg{putS}{\I{tid}=i\land\I{cm}=c\land\I{bal}=0}\seqA\msg{putC}{\I{tid}=i\land\I{cm}=c\land\I{bal}=0}\seqA
  \\ & \quad \msg{commit}{\I{tid}=i}\seqA\msg{openARsp}{\true}\seqA\allA\effR{}
  \\
  \eff{openARsp} ~:~{}
  & \effLR{\allA}\,\Code{unit}\,\effLR{\allA}
  \\
  \eff{amalReq} ~:~{}
  & i{:}\Int,\; c_0{:}\Int,\; c_1{:}\Int \sarr \effLR{\allA\seqA\msg{openAReq}{\I{cm}=c_0\lor\I{cm}=c_1}\seqA\allA\seqA\msg{openAReq}{\I{cm}=c_0\lor\I{cm}=c_1}\seqA\allA}\,\Code{unit}\,
  \\ & \quad
    \effL{}\msg{amalReq}{\I{cm0}=c_0\land\I{cm1}=c_1}\seqA\msg{begin}{\I{tid}=i}\seqA\msg{getS}{\I{tid}=i\land\I{cm}=c_0}\seqA
  \\ & \quad \msg{getC}{\I{tid}=i\land\I{cm}=c_0}\seqA\msg{getC}{\I{tid}=i\land\I{cm}=c_1}\seqA \allA \seqA
  \\ & \quad \msg{putS}{\I{tid}=i\land\I{cm}=c_0\land\I{bal}=0}\seqA\msg{putC}{\I{tid}=i\land\I{cm}=c_0\land\I{bal}=0}\seqA\msg{putC}{\I{tid}=i\land\I{cm}=c_1}\seqA
  \\ & \quad \SNegStarExcl{\msg{putS}{\I{cm}=c_0} \cup \msg{putC}{\I{cm}=c_0} \cup \msg{putC}{\I{cm}=c_1}}\seqA\msg{commit}{\I{tid}=i}\seqA\msg{amalRsp}{\true}\seqA\allA\effR{}
  \\
  \eff{amalRsp} ~:~{}
  & \effLR{\allA}\,\Code{unit}\,\effLR{\allA}
  \\
  \eff{balanceReq} ~:~{}
  & i{:}\Int,\; n{:}\Int,\; c{:}\Int,\; b_s{:}\Int,\; b_c{:}\Int \sarr \effLR{\allA\seqA\msg{openAReq}{\I{name}=n\land\I{cm}=c}\seqA\allA}\,\Code{unit}\,
  \\ & \quad
    \effL{}\msg{balanceReq}{\I{name}=n}\seqA\msg{begin}{\I{tid}=i}\seqA\msg{getA}{\I{tid}=i\land\I{name}=n\land\I{cm}=c}\seqA
  \\ & \quad \msg{getS}{\I{tid}=i\land\I{cm}=c\land\I{bal}=b_s}\seqA\msg{getC}{\I{tid}=i\land\I{cm}=c\land\I{bal}=b_c}\seqA \allA \seqA
  \\ & \quad \msg{commit}{\I{tid}=i}\seqA\msg{balanceRsp}{\I{bal}=b_s+b_c}\seqA\allA\effR{}
  \\
  \eff{balanceRsp} ~:~{}
  & \effLR{\allA}\,\Code{unit}\,\effLR{\allA}
  \\
  \eff{depositReq} ~:~{}
  & i{:}\Int,\; n{:}\Int,\; c{:}\Int,\; b{:}\Int,\; a{:}\Int \sarr \effLR{\allA\seqA\msg{openAReq}{\I{name}=n\land\I{cm}=c}\seqA\allA}\,\Code{unit}\,
  \\ & \quad
    \effL{}\msg{depositReq}{\I{name}=n\land\I{amt}=a}\seqA\msg{begin}{\I{tid}=i}\seqA\msg{getA}{\I{tid}=i\land\I{name}=n\land\I{cm}=c}\seqA
  \\ & \quad \msg{getC}{\I{tid}=i\land\I{cm}=c\land\I{bal}=b}\seqA \allA \seqA\msg{putC}{\I{tid}=i\land\I{cm}=c}\seqA
  \\ & \quad \SNegStar{\msg{putC}{\I{cm}=c}}\seqA\msg{commit}{\I{tid}=i}\seqA\msg{depositRsp}{\true}\seqA\allA\effR{}
  \\
  \eff{depositRsp} ~:~{}
  & \effLR{\allA}\,\Code{unit}\,\effLR{\allA}
  \\
  \eff{sendReq} ~:~{}
  & i{:}\Int,\; s{:}\Int,\; d{:}\Int,\; b_s{:}\Int,\; b_d{:}\Int,\; a{:}\Int \sarr {}
  \\ & \quad
    \effLR{\allA\seqA\msg{openAReq}{\I{cm}=s\lor\I{cm}=d}\seqA\allA\seqA\msg{openAReq}{\I{cm}=s\lor\I{cm}=d}\seqA\allA}\,\Code{unit}\,
  \\ & \quad
    \effL{}\msg{sendReq}{\I{srcid}=s\land\I{destid}=d\land\I{amt}=a}\seqA\msg{begin}{\I{tid}=i}\seqA
  \\ & \quad \msg{getC}{\I{tid}=i\land\I{cm}=s\land\I{bal}=b_s}\seqA\msg{getC}{\I{tid}=i\land\I{cm}=d\land\I{bal}=b_d}\seqA \allA \seqA
  \\ & \quad \msg{putC}{\I{tid}=i\land\I{cm}=s\land\I{bal}=b_s-a}\seqA\msg{putC}{\I{tid}=i\land\I{cm}=d\land\I{bal}=b_d+a}\seqA
  \\ & \quad \SNegStarUU{\msg{putC}{\I{cm}=s}}{\msg{putC}{\I{cm}=d}}\seqA\msg{commit}{\I{tid}=i}\seqA\msg{sendRsp}{\true}\seqA\allA\effR{}
  \\
  \eff{sendRsp} ~:~{}
  & \effLR{\allA}\,\Code{unit}\,\effLR{\allA}
  \\
\end{align*}}
\vspace{-.4cm}\par\nobreak\noindent

\vspace{-.35cm}\par\nobreak{\small
  \begin{align*}
  \eff{transReq} ~:~{}
  & i{:}\Int,\; n{:}\Int,\; c{:}\Int,\; b{:}\Int,\; a{:}\Int \sarr \effLR{\allA\seqA\msg{openAReq}{\I{name}=n\land\I{cm}=c}\seqA\allA}\,\Code{unit}\,
  \\ & \quad
    \effL{}\msg{transReq}{\I{name}=n\land\I{amt}=a}\seqA\msg{begin}{\I{tid}=i}\seqA\msg{getA}{\I{tid}=i\land\I{name}=n\land\I{cm}=c}\seqA
  \\ & \quad \msg{getS}{\I{tid}=i\land\I{cm}=c\land\I{bal}=b}\seqA \allA \seqA\msg{putS}{\I{tid}=i\land\I{cm}=c\land\I{bal}=b+a}\seqA
  \\ & \quad \SNegStar{\msg{putS}{\I{cm}=c}}\seqA\msg{commit}{\I{tid}=i}\seqA\msg{transRsp}{\true}\seqA\allA\effR{}
  \\
  \eff{transRsp} ~:~{}
  & \effLR{\allA}\,\Code{unit}\,\effLR{\allA}
  \\
  \eff{writeCReq} ~:~{}
  & i{:}\Int,\; n{:}\Int,\; c{:}\Int,\; a{:}\Int \sarr \effLR{\allA\seqA\msg{openAReq}{\I{name}=n\land\I{cm}=c}\seqA\allA}\,\Code{unit}\,
  \\ & \quad
    \effL{}\msg{writeCReq}{\I{name}=n\land\I{amt}=a}\seqA\msg{begin}{\I{tid}=i}\seqA\msg{getA}{\I{tid}=i\land\I{name}=n\land\I{cm}=c}\seqA
  \\ & \quad \msg{getS}{\I{tid}=i\land\I{cm}=c}\seqA\msg{getC}{\I{tid}=i\land\I{cm}=c}\seqA \allA \seqA
  \\ & \quad \msg{putC}{\I{tid}=i\land\I{cm}=c}\seqA\SNegStar{\msg{putC}{\I{cm}=c}}\seqA\msg{commit}{\I{tid}=i}\seqA\msg{writeCRsp}{\true}\seqA\allA\effR{}
  \\
  \eff{writeCRsp} ~:~{}
  & \effLR{\allA}\,\Code{unit}\,\effLR{\allA}
\end{align*}}
\vspace{-.4cm}\par\nobreak\noindent
Each scenario threads \(\eff{begin}/\eff{commit}\) with the same write--write gaps on the relevant account rows (e.g.\ \(\SNegStar{\msg{putC}{\I{cm}=c}}\) before commit) as in \textsf{Shopping}. For the weakened \tyName{}, we switch the MonkeyDB layer to RC as above; the SmallBank scenarios themselves are unchanged.

\subsection*{P Language Benchmarks}
The remaining models are asynchronous message-passing protocols drawn from PBench. As in \textsc{Stack} and the synchronous benchmarks above, we give a short scenario, the global property as an SRE, the operation-wise \tyName{} clauses (where futures often encode protocol order and search bias), and a note on weakening.

\subsubsection*{Database}
This model is a simplified database with APIs: \(\eff{writeReq}/\eff{writeRsp}\) and \(\eff{readReq}/\eff{readRsp}\). It must satisfy the Read-Your-Writes (RYW) policy: reads should respect the last completed write unless another write slips in between; the generator therefore chases a successful read that nevertheless returns \(y\neq x\) right after \(\msg{writeRsp}{\I{va}=x}\) and with no further \(\msg{writeRsp}{\true}\) in the gap. Thus, the global property is:
\vspace{-.35cm}\par\nobreak{\small
\begin{align*}
  & \allA \seqA \msg{writeRsp}{\I{va}=x} \seqA \SNegStar{\msg{writeRsp}{\true}} \seqA
    \msg{readRsp}{\I{va}=y \land x \neq y \land \mathit{st}} \seqA \allA
\end{align*}}
\vspace{-.4cm}\par\nobreak\noindent
The \tyName{} comprises the following clauses:
\vspace{-.35cm}\par\nobreak{\small
\begin{align*}
  \eff{readReq} ~:~{}
  & x{:}\Int \sarr \effLR{\allA\seqA\msg{writeReq}{\I{va}=x}\seqA\SNegStar{\msg{writeReq}{\true}}}\,\Code{unit}\,\effLR{\msg{readRsp}{\I{va}=x\land \mathit{st}}} \interty
  \\
  & \effLR{\SNegStar{\msg{writeReq}{\true}}}\,\Code{unit}\,\effLR{\msg{readRsp}{\neg\mathit{st}}}
  \\
  \eff{writeReq} ~:~{}
  & x{:}\Int \sarr \effLR{\allA}\,\Code{unit}\,\effLR{\msg{writeReq}{\I{va}=x}\seqA\msg{writeRsp}{\I{va}=x}}
  \\
  \eff{writeRsp} ~:~{}
  & x{:}\Int \sarr \effLR{\allA}\,\urt{\Code{int}}{\I{va}=x}\,\effLR{\allA}
  \\
  \eff{readRsp} ~:~{}
  & x{:}\Int \garr  s{:}\Bool \sarr \effLR{\allA}\,\urt{\Code{int}}{\I{va}=x \land \mathit{st}=s}\,\effLR{\allA}
\end{align*}}
\vspace{-.4cm}\par\nobreak\noindent
The history on \(\eff{readReq}\) ties a successful read to a matching prior \(\eff{writeReq}\) (the Kleene gap excludes stray writes), which biases generation toward read-your-writes-shaped traces. The weakened \tyName{} relaxes those payload ties on \(\eff{writeRsp}\) and \(\eff{readRsp}\) by weakening the qualifier directly.

\newpage
\subsubsection*{Firewall} The firewall divides the nodes into ``internal nodes'' and ``external nodes''; external nodes cannot send messages to internal nodes until they are recorded on the ``white list'' of the firewall. To add an external node to the white list, at least one internal node must send it a message. The firewall under test has incorrect white list recording, which leads to an error described by the following global property:
\vspace{-.35cm}\par\nobreak{\small
\begin{align*}
  & \SNegStarU{\msg{eExReq}{\I{node}=n}}{\msg{eFwReq}{\I{node}=n}} \seqA
    \msg{eInReq}{\I{node}=n} \seqA \allA \seqA
    \msg{eExRsp}{\I{node}=n \land \neg\,\mathit{stat}} \seqA \allA
\end{align*}}
\vspace{-.4cm}\par\nobreak\noindent
The \tyName{} comprises the following clauses:
\vspace{-.35cm}\par\nobreak{\small
\begin{align*}
  \eff{eStart} ~:~{}
  & n{:}\mathsf{Node} \sarr \effLR{\allA}\,\Code{unit}\,\effLR{\msg{eStart}{\I{node}=n}\seqA\msg{eInReq}{\I{node}=n}}
  \\
  \eff{eInReq} ~:~{}
  & n{:}\mathsf{Node} \sarr \effLR{\SNegStarUU{\msg{eExReq}{\I{node}=n}}{\msg{eInReq}{\I{node}=n}}{\msg{eFwReq}{\I{node}=n}}}\,\Code{unit}\,
  \\&\quad \effLR{\msg{eFwReq}{\I{node}=n}}
  \\
  \eff{eFwReq} ~:~{}
  & n{:}\mathsf{Node} \sarr \effLR{\SNegStar{\msg{eFwReq}{\I{node}=n}}}\,\Code{unit}\,\effLR{\msg{eExReq}{\I{node}=n}}
  \\
  \eff{eExReq} ~:~{}
  & n{:}\mathsf{Node} \sarr \effLR{\allA\seqA\msg{eInReq}{\I{node}=n}\seqA\SNegStar{\msg{eInReq}{\true}}}\,\Code{unit}\,
  \\ & \quad \effLR{\msg{eExReq}{\I{node}=n}\seqA\msg{eExRsp}{\I{node}=n\land\mathit{stat}}} \interty
  \\
  & n{:}\mathsf{Node} \sarr \effLR{\allA\seqA\msg{eInReq}{\neg(\I{node}=n)}\seqA\SNegStar{\msg{eInReq}{\true}}}\,\Code{unit}\,
  \\ & \quad \effLR{\msg{eExReq}{\I{node}=n}\seqA\msg{eExRsp}{\I{node}=n\land\neg\mathit{stat}}}
  \\
  \eff{eExRsp} ~:~{}
  & n{:}\mathsf{Node}\garr s{:}\Bool \sarr \effLR{\allA}\,\urt{\Code{bool}}{\mathit{stat} = s}\,\effLR{\allA}
\end{align*}}
\vspace{-.4cm}\par\nobreak\noindent
Histories on \(\eff{eInReq}\) and \(\eff{eFwReq}\) rule out re-entrant forwarding noise before the next hop, sharpening the causal shape the tester must satisfy. The weakened \tyName{} merges the intersected \(\eff{eExReq}\) cases and weakens guards on those events.

\subsubsection*{RingLeaderElection}
Processes ring-elect a leader using \(\eff{eWakeup}\) and \(\eff{eNom}\), with outcomes published as \(\msg{eWon}{\mathit{ld}=\cdots}\); at most one leader value should ultimately win. The search targets two \(\eff{eWon}\) announcements that contradict each other on \(\mathit{ld}\). The global property is:
\vspace{-.35cm}\par\nobreak{\small
\begin{align*}
  & \allA \seqA \msg{eWon}{\mathit{ld}=\ell} \seqA \allA \seqA
    \msg{eWon}{\neg(\mathit{ld}=\ell)} \seqA \allA
\end{align*}}
\vspace{-.4cm}\par\nobreak\noindent
The \tyName{} comprises the following clauses (\(\mathit{next}\) is axiomatized on ring nodes):
\vspace{-.35cm}\par\nobreak{\small
\begin{align*}
  \eff{eWakeup} ~:~{}
  & n{:}\mathsf{Node} \sarr \effLR{\SNegStar{\msg{eWon}{\true}}}\,\Code{unit}\,\effLR{\msg{eNom}{\I{node}=n\land\mathit{ld}=n}}
  \\
  \eff{eNom} ~:~{}
  & n{:}\mathsf{Node} \garr \ell{:}\urt{\mathsf{Node}}{\ell\neq n} \sarr \effLR{\allA}\,\Code{unit}\,\effLR{\msg{eNom}{\I{node}=n\land\mathit{ld}=\ell}\seqA\msg{eWon}{\mathit{ld}=\ell}} \interty
  \\
  & n{:}\mathsf{Node} \garr \ell{:}\urt{\mathsf{Node}}{\ell=n} \sarr \effLR{\allA}\,\Code{unit}\,\effL{}\msg{eNom}{\I{node}=n\land\mathit{ld}=\ell}\seqA
  \\ & \quad \msg{eNom}{\I{node}=\mathit{next}(n)\land\mathit{ld}=\ell}\effR{}
  \\
  \eff{eWon} ~:~{}
  & \ell{:}\mathsf{Node} \sarr \effLR{\allA}\,\Code{unit}\,\effLR{\msg{eWon}{\mathit{ld}=\ell}}
\end{align*}}
\vspace{-.4cm}\par\nobreak\noindent
The future on \(\eff{eWakeup}\) forbids any earlier \(\msg{eWon}{\true}\), steering traces toward elections that have not yet closed; that is the main search bias for this benchmark. The weakened \tyName{} drops these guard constraints on \(\eff{eWakeup}\) and \(\eff{eWon}\).

\newpage
\subsubsection*{Bankserver}
The bank server sequences balance reads and updates around \(\eff{wdReq}\) and \(\eff{eWDRsp}\); withdrawals are meant to stay within available funds. The buggy server still reports a failed withdrawal on \(\I{ac}\) (\(\neg\mathit{status}\)) in a situation where the recorded balances should have sufficed. The global property is:
\vspace{-.35cm}\par\nobreak{\small
\begin{align*}
  & \allA \seqA \msg{eWDRsp}{\I{acId}=\I{ac} \land \neg\,\mathit{status}} \seqA \allA
\end{align*}}
\vspace{-.4cm}\par\nobreak\noindent
The \tyName{} comprises the following clauses (\(\mathsf{aid}/\mathsf{rid}\) are phantom account/request ids):
\vspace{-.35cm}\par\nobreak{\small
\begin{align*}
  \eff{initAc} ~:~{}
  & ac{:}\mathsf{aid}\garr ba{:}\urt{\Int}{ba>0} \sarr \effLR{\SNegStar{\msg{initAc}{\I{acId}=ac}}}\,\Code{unit}\,\effLR{\msg{initAc}{\I{acId}=ac\land\I{bal}=ba}}
  \\
  \eff{wdReq} ~:~{}
  & id{:}\mathsf{rid}\garr ac{:}\mathsf{aid}\garr am{:}\urt{\Int}{am>0} \sarr \effLR{\allA}\,\Code{unit}\,
  \\ & \quad \effLR{\msg{wdReq}{\I{rId}=id\land\I{acId}=ac\land\I{amt}=am}\seqA\msg{read}{\I{rId}=id\land\I{acId}=ac\land\I{amt}=am}}
  \\
  \eff{read} ~:~{}
  & ba{:}\Int\garr id{:}\mathsf{rid}\garr am{:}\garr ac{:}\mathsf{aid} \sarr 
  \\ & \quad \effLR{\allA\seqA\msg{initAc}{\I{acId}=ac\land\I{bal}=ba}\seqA\SNegStarU{\msg{read}{\I{acId}=ac}}{\msg{initAc}{\I{acId}=ac}}}\,\Code{unit}\,
  \\ & \quad \effLR{\msg{read}{\I{rId}=id\land\I{amt}=am\land\I{acId}=ac}\seqA\msg{readRsp}{\I{rId}=id\land\I{amt}=am\land\I{acId}=ac\land\I{bal}=ba}}
  \\
  \eff{readRsp} ~:~{}
  & id{:}\mathsf{rid}\garr am{:}\garr ac{:}\mathsf{aid}\garr ba{:}\urt{\Int}{ba>0 \land ba>am} \sarr {}
  \\ & \quad \effL{}\SNegStarU{\msg{readRsp}{\true}}{\msg{read}{\true}}\seqA\msg{wdReq}{\I{rId}=id\land\I{acId}=ac\land\I{amt}=am}\seqA
  \\& \quad \SNegStar{\msg{wdReq}{\true}}\effR{}\,\Code{unit}\,
  \\ & \quad \effL{}\msg{readRsp}{\I{acId}=ac\land\I{bal}=ba}\seqA\msg{update}{\I{acId}=ac\land\I{bal}=ba-am}\seqA
  \\&\quad \msg{eWDRsp}{\I{rId}=id\land\I{acId}=ac\land\I{bal}=ba-am\land\mathit{status}}\effR{}
  \\
  & id{:}\mathsf{rid}\garr am{:}\garr ac{:}\mathsf{aid}\garr ba{:}\urt{\Int}{ba>0 \land ba\le am} \sarr \effLR{\allA}\,\Code{unit}\,
  \\ & \quad \effLR{\msg{readRsp}{\I{rId}=id\land\I{amt}=am\land\I{acId}=ac\land\I{bal}=ba}\seqA\msg{eWDRsp}{\I{rId}=id\land\I{acId}=ac\land\neg\mathit{status}}}
  \\
  \eff{update} ~:~{}
  & ac{:}\mathsf{aid}\garr ba{:}\urt{\Int}{ba>0} \sarr \effLR{\allA}\,\Code{unit}\,\effLR{\msg{update}{\I{acId}=ac\land\I{bal}=ba}}
  \\
  \eff{eWDRsp} ~:~{}
  & id{:}\mathsf{rid}\garr ac{:}\mathsf{aid}\garr ba{:}\Int,\; st{:}\Bool \sarr \effLR{\allA}\,\urt{\Code{unit}}{\I{rId}=id\land\I{acId}=ac\land\I{bal}=ba\land\mathit{status}=st}\,\effLR{\allA}
\end{align*}}
\vspace{-.4cm}\par\nobreak\noindent
The successful branch requires \(\I{bal}>\I{amt}\) at the read response so the backend can emit \(\eff{update}\); the failing branch skips the update. The weakened \tyName{} simplifies response/update qualifiers and widens read-query intersected types.

\newpage
\subsubsection*{Simplified2pc}
The SUT is a simplified two-phase commit client: committed state should read atomically, so after a write response commits \(x\) and no further write completes in the middle, a read response should not return \(y\neq x\). In SRE form this is a committed \(\eff{writeRsp}\), a Kleene stretch with no further successful \(\eff{writeRsp}\), a \(\eff{readRsp}\) that mismatches the written value, and a final stretch with no further \(\eff{readRsp}\) or \(\eff{writeRsp}\).
The global property is:
\vspace{-.35cm}\par\nobreak{\small
\begin{align*}
  & \allA \seqA
    \msg{writeRsp}{\I{va}=x \land \mathit{stat}} \seqA
    \SNegStar{\msg{writeRsp}{\mathit{stat}}} \seqA
    \msg{readRsp}{\I{va}=y \land x \neq y} \seqA
    \\&\quad \SNegStarU{\msg{readRsp}{\true}}{\msg{writeRsp}{\true}}
\end{align*}}
\vspace{-.4cm}\par\nobreak\noindent
The \tyName{} is:
\vspace{-.35cm}\par\nobreak{\small
\begin{align*}
  \eff{readReq} ~:~{}
  & \effLR{\allA}\,\Code{unit}\,\effLR{\msg{readReq}{\true}\seqA\msg{getReq}{\true}}
  \\
  \eff{getReq} ~:~{}
  & \effLR{\SNegStar{\msg{commit}{\true}}}\,\Code{unit}\,\effLR{\msg{readRsp}{\I{va}=-1}}
  \\
  \eff{readRsp} ~:~{}
  & x{:}\Int \sarr \effLR{\allA}\,\urt{\Code{int}}{\I{va}=x}\,\effLR{\allA}
  \\
  \eff{writeReq} ~:~{}
  & x{:}\urt{\Int}{x\ge 0} \sarr \effLR{\allA}\,\Code{unit}\,\effLR{\msg{writeReq}{\I{va}=x}\seqA\msg{putReq}{\I{va}=x}}
  \\
  \eff{putReq} ~:~{}
  & x{:}\Int \sarr \effLR{\allA}\,\Code{unit}\,\effLR{\msg{putReq}{\I{va}=x}\seqA\msg{putRsp}{\I{va}=x}}
  \\
  \eff{putRsp} ~:~{}
  & x{:}\Int \garr  s{:}\Bool \sarr \effLR{\allA}\,\Code{unit}\,\effLR{\msg{putRsp}{\I{va}=x\land\mathit{stat}}\seqA\msg{writeRsp}{\I{va}=x\land\mathit{stat}=s}\seqA\msg{commit}{\true}}
  \\
  & x{:}\Int \garr  s{:}\Bool \sarr \effLR{\allA}\,\Code{unit}\,\effLR{\msg{putRsp}{\I{va}=x\land\mathit{stat}=s}\seqA\msg{writeRsp}{\I{va}=x\land\neg\mathit{stat}}\seqA\msg{abort}{\true}}
  \\
  \eff{writeRsp} ~:~{}
  & x{:}\Int \garr  s{:}\Bool \sarr \effLR{\allA}\,\urt{\Code{int}}{\I{va}=x\land\mathit{stat}=s}\,\effLR{\allA}
  \\
  \eff{commit} ~:~{}
  & \effLR{\allA}\,\Code{unit}\,\effLR{\msg{commit}{\true}}
  \\
  \eff{abort} ~:~{}
  & \effLR{\allA}\,\Code{unit}\,\effLR{\msg{abort}{\true}}
\end{align*}}
\vspace{-.4cm}\par\nobreak\noindent
The weakened \tyName{} drops non-negative value constraints on \(\eff{writeReq}\); many \texttt{putRsp}, \texttt{writeRsp}, and \texttt{readRsp} payload tests become \texttt{true}.

\subsubsection*{Heartbeat}
The heartbeat failure detector should align ``nodes down'' announcements with actual shutdowns. Note that messages can be lost due to network issues. The system has a timeout retry limit as $3$ (i.e., $\I{trial}$). The bug class is a false alarm: \(\eff{eNotifyNodesDown}\) fires while \(\msg{eShutDown}{\true}\) never appears on either side in the exhibited pattern. The global property is:
\vspace{-.35cm}\par\nobreak{\small
\begin{align*}
  & \SNegStar{\msg{eShutDown}{\true}} \seqA \msg{eNotifyNodesDown}{\true} \seqA \SNegStar{\msg{eShutDown}{\true}}
\end{align*}}
\vspace{-.4cm}\par\nobreak\noindent
The \tyName{} comprises the following clauses:
\vspace{-.35cm}\par\nobreak{\small
\begin{align*}
  \eff{eNotifyNodesDown} ~:~{}
  & \effLR{\SNegStar{\msg{eNotifyNodesDown}{\true}}}\,\Code{unit}\,\effLR{\allA}
  \\
  \eff{eNetworkError} ~:~{}
  & tl{:}\Int \sarr \effLR{\allA\seqA\msg{ePong}{\I{trial}=tl}\seqA\allA}\,\Code{unit}\,\effLR{\msg{ePongLost}{\I{trial}=tl}}
  \\
  \eff{ePing} ~:~{}
  & tl{:}\Int \sarr \effLR{\SNegStar{\msg{eShutDown}{\true}}}\,\Code{unit}\,\effLR{\msg{ePong}{\I{trial}=tl}}
  \\
  \eff{eShutDown} ~:~{}
  & \effLR{\SNegStar{\msg{eShutDown}{\true}}}\,\Code{unit}\,\effLR{\msg{eShutDown}{\true}}
  \\
  \eff{eStart} ~:~{}
  & \effLR{\SNegStarTri{\msg{ePing}{\true}}{\msg{ePongLost}{\true}}{\msg{eStart}{\true}}}\,\Code{unit}\,\effLR{\msg{ePing}{\I{trial}=1}}
  \\
  \eff{ePong} ~:~{}
  & tl{:}\Int \sarr \effLR{\SNegStar{\msg{ePongLost}{\I{trial}=tl}}}\,\Code{unit}\,\effLR{\allA}
  \\
  \eff{ePongLost} ~:~{}
  & tl{:}\urt{\Int}{tl=1} \sarr \effLR{\SNegStar{\msg{ePongLost}{\I{trial}=tl}}}\,\Code{unit}\,\effLR{\msg{ePing}{\I{trial}=2}}
  \\
  & tl{:}\urt{\Int}{tl=2} \sarr \effLR{\SNegStar{\msg{ePongLost}{\I{trial}=tl}}}\,\Code{unit}\,\effLR{\msg{ePing}{\I{trial}=3}}
  \\
  & tl{:}\urt{\Int}{tl=3} \sarr \effLR{\allA}\,\Code{unit}\,\effLR{\msg{eNotifyNodesDown}{\true}}
\end{align*}}
\vspace{-.4cm}\par\nobreak\noindent
The weakened \tyName{} keeps the trial-level structure but softens side conditions.

\newpage
\subsubsection*{Chainreplication}
Chain-replicated storage ought to make completed writes visible to later successful reads on the same key. This chain is simplified to only have three nodes: the head node, middle node, and tail node; since the Chain Replication protocol has three different policies for these three kinds of nodes. Here a client still reads a stale value: the write on \(\I{key}=k\) finishes (\(\eff{writeRsp}\)), yet the next successful read (\(\mathit{st}\)) returns \(y\neq x\) although no further \(\msg{writeRsp}{\true}\) occurs between them. The global property is:
\vspace{-.35cm}\par\nobreak{\small
\begin{align*}
  & \allA \seqA \msg{writeReq}{\I{key}=k \land \I{va}=x} \seqA
    \SNegStar{\msg{writeRsp}{\true}} \seqA
    \msg{readRsp}{\I{key}=k \land \I{va}=y \land y \neq x \land \mathit{st}} \seqA
    \\&\quad \SNegStar{\msg{writeRsp}{\true}}
\end{align*}}
\vspace{-.4cm}\par\nobreak\noindent
The \tyName{} comprises the following clauses (\(\mathit{mid1}/\mathit{mid2}\) axiomatize mid vs.\ tail nodes):
\vspace{-.35cm}\par\nobreak{\small
\begin{align*}
  \eff{writeReq} ~:~{}
  & k{:}\mathsf{Key}\garr x{:}\Int \sarr \effL{} \allA \effR{}\,\Code{unit}\,\effL{} \msg{writeReq}{\I{key}=k\land\I{va}=x}\seqA
  \\& \quad \msg{writeToMid}{\I{key}=k\land\I{va}=x\land\mathit{mid1}(\I{node})} \effR{}
  \\
  \eff{writeToMid} ~:~{}
  & k{:}\mathsf{Key}\garr x{:}\Int \sarr 
  \\& \quad  n{:}\urt{\mathsf{Node}}{\mathit{mid2}(n)} \sarr \effL{} \allA \effR{}\,\urt{\Code{unit}}{\I{key}=k\land\I{va}=x\land\I{node}=n}\,
  \\& \quad \effL{} \msg{writeToTail}{\I{key}=k\land\I{va}=x} \effR{} \interty
  \\
  & k{:}\mathsf{Key}\garr x{:}\Int \sarr
  \\& \quad n{:}\urt{\mathsf{Node}}{\mathit{mid1}(n)} \sarr \effL{} \allA \effR{}\,\urt{\Code{unit}}{\I{key}=k\land\I{va}=x\land\I{node}=n}\,\
  \\& \quad \effL{} \msg{writeToMid}{\I{key}=k\land\I{va}=x\land\I{node}=\mathit{next}(n)} \effR{}
  \\
  \eff{writeToTail} ~:~{}
  & k{:}\mathsf{Key}\garr x{:}\Int \sarr \effL{} \SNegStar{\msg{crashTail}{\true}} \effR{}\,\urt{\Code{unit}}{\I{key}=k\land\I{va}=x}\,\effL{} \allA \effR{} \interty
  \\
  & k{:}\mathsf{Key}\garr x{:}\Int \sarr \effL{} \SNegStar{\msg{crashTail}{\true}} \effR{}\,\urt{\Code{unit}}{\I{key}=k\land\I{va}=x}\,
  \\& \quad \effL{} \msg{writeRsp}{\I{key}=k\land\I{va}=x} \effR{}
  \\
  \eff{writeRsp} ~:~{}
  & k{:}\mathsf{Key}\garr x{:}\Int \sarr \effL{} \allA \effR{}\,\urt{\Code{unit}}{\I{key}=k\land\I{va}=x}\,\effL{} \allA \effR{}
  \\
  \eff{readReq} ~:~{}
  & k{:}\mathsf{Key} \sarr \effL{} \SNegStarU{\msg{crashTail}{\true}}{\msg{writeToTail}{\I{key}=k}} \effR{}\,\Code{unit}\,\effL{} \msg{readRsp}{\I{key}=k\land\neg\mathit{st}} \effR{} \interty
  \\
  & x{:}\Int \garr  k{:}\mathsf{Key} \sarr \effL{} \SNegStar{\msg{crashTail}{\true}}\seqA\msg{crashTail}{\true}\seqA
  \\ & \quad \SNegStar{\msg{crashTail}{\true}}\seqA\msg{writeToMid}{\I{key}=k\land\I{va}=x\land\mathit{mid2}(\I{node})}\seqA
  \\ & \quad \SNegStarU{\msg{crashTail}{\true}}{\msg{writeToMid}{\I{key}=k\land\I{va}=x\land\mathit{mid2}(\I{node})}} \effR{}\,\Code{unit}\,{}
  \\ & \quad \effL{} \msg{readRsp}{\I{key}=k\land\I{va}=x\land\mathit{st}} \effR{} \interty
  \\
  & x{:}\Int \garr  k{:}\mathsf{Key} \sarr \effL{} \SNegStar{\msg{crashTail}{\true}}\seqA\msg{writeToMid}{\I{key}=k\land\I{va}=x\land\mathit{mid2}(\I{node})}\seqA
  \\ & \quad \SNegStarU{\msg{crashTail}{\true}}{\msg{writeToMid}{\I{key}=k\land\mathit{mid2}(\I{node})}}\seqA\msg{crashTail}{\true}\seqA
  \\ & \quad \SNegStarU{\msg{crashTail}{\true}}{\msg{writeToMid}{\I{key}=k\land\mathit{mid2}(\I{node})}} \effR{}\,\Code{unit}\,{}
  \\ & \quad \effL{} \msg{readRsp}{\I{key}=k\land\I{va}=x\land\mathit{st}} \effR{} \interty
  \\
  & x{:}\Int \garr  k{:}\mathsf{Key} \sarr \effL{} \SNegStarU{\msg{crashTail}{\true}}{\msg{writeToMid}{\I{key}=k\land\mathit{mid2}(\I{node})}}\seqA
  \\ & \quad \msg{crashTail}{\true}\seqA\SNegStarU{\msg{crashTail}{\true}}{\msg{writeToMid}{\I{key}=k\land\mathit{mid2}(\I{node})}} \effR{}\,\Code{unit}\,
  \\ & \quad \effL{} \msg{readRsp}{\I{key}=k\land\neg\mathit{st}} \effR{} \interty
  \\
  & x{:}\Int \garr  k{:}\mathsf{Key} \sarr \effL{} \SNegStar{\msg{crashTail}{\true}}\seqA\msg{writeToTail}{\I{key}=k\land\I{va}=x}\seqA
  \\ & \quad \SNegStarU{\msg{crashTail}{\true}}{\msg{writeToTail}{\true}} \effR{}\,\Code{unit}\,\effL{} \msg{readRsp}{\I{key}=k\land\I{va}=x\land\mathit{st}} \effR{}
  \\
  \eff{readRsp} ~:~{}
  & k{:}\mathsf{Key}\garr x{:}\Int \garr  s{:}\Bool \sarr \effL{} \allA \effR{}\,\urt{\Code{int}}{\I{key}=k\land\I{va}=x\land\mathit{st}=s}\,\effL{} \allA \effR{}
  \\
  \eff{crashTail} ~:~{}
  & \effL{} \allA \effR{}\,\Code{unit}\,\effL{} \allA \effR{}
\end{align*}}
\vspace{-.4cm}\par\nobreak\noindent

The weakened \tyName{} merges intersected \(\eff{readReq}\) types into fewer success/failure skeletons.

\newpage
\subsubsection*{Paxos}
Paxos exposes an observable learner \(\eff{eLearn}\) that ought to settle on a single replicated value. We use $*\eff{Lost}*$ to simulate the network issue as introduced in the \textsf{HeartBeat} benchmark. The stress trace lets the learner emit two learns with different payloads \(x\) and \(y\) (\(x\neq y\)). The global property is:
\vspace{-.35cm}\par\nobreak{\small
\begin{align*}
  & \allA \seqA \msg{eLearn}{\I{va}=x} \seqA \allA \seqA
    \msg{eLearn}{\I{va}=y \land x \neq y} \seqA \allA
\end{align*}}
\vspace{-.4cm}\par\nobreak\noindent
The \tyName{} comprises the following clauses (\(\mathit{pr1}/\mathit{pr2}\) partition proposers and \(\mathit{acc1}/\mathit{acc2}\) partition acceptors):
\vspace{-.35cm}\par\nobreak{\small
\begin{align*}
  \eff{eStart} ~:~{}
  & p{:}\mathsf{Node}\garr x{:}\mathsf{tVal} \sarr \effL{} \allA\seqA\msg{eLostPrepReq}{\I{proposer}=p\land\mathit{acc2}(\I{acceptor})}\seqA\allA \effR{}\,
  \\ & \quad \Code{unit}\,\effL{} \msg{ePrepReq}{\I{proposer}=p\land\mathit{acc1}(\I{acceptor})\land\I{va}=x} \effR{} \interty
  \\
  & p{:}\mathsf{Node}\garr x{:}\mathsf{tVal} \sarr \effL{} \allA\seqA\msg{eLostPrepReq}{\I{proposer}=p\land\mathit{acc1}(\I{acceptor})}\seqA\allA \effR{}\,
  \\ & \quad \Code{unit}\,\effL{} \msg{ePrepReq}{\I{proposer}=p\land\mathit{acc2}(\I{acceptor})\land\I{va}=x} \effR{}
  \\
  \eff{eLearn} ~:~{}
  & x{:}\mathsf{tVal} \sarr \effL{} \allA \effR{}\,\urt{\Code{unit}}{\I{va}=x}\,\effL{} \allA \effR{}
  \\
  \eff{eLostPrepReq} ~:~{}
  & p{:}\mathsf{Node}\garr ac{:}\mathsf{Node} \sarr \effL{} \allA \effR{}\,\urt{\Code{unit}}{\I{proposer}=p\land\I{acceptor}=ac}\,\effL{} \allA \effR{}
  \\
  \eff{ePrepReq} ~:~{}
  & p{:}\urt{\mathsf{Node}}{\mathit{pr2}(p)}\garr ac{:}\mathsf{Node}\garr x{:}\mathsf{tVal} \sarr {}
  \\ & \quad \effL{} \SNegStarU{\msg{ePrepReq}{\I{acceptor}=ac}}{\msg{eAcReq}{\true}}\seqA
  \\ & \quad \msg{ePrepReq}{\mathit{pr1}(\I{proposer})\land\I{acceptor}=ac}\seqA\SNegStarU{\msg{ePrepReq}{\I{acceptor}=ac}}{\msg{eAcReq}{\true}} \effR{}\,
  \\ & \quad \urt{\Code{unit}}{\I{proposer}=p\land\I{acceptor}=ac\land\I{va}=x}\,
  \\ & \quad \effL{} \msg{ePrepRsp}{\I{acceptor}=ac\land\I{promised}=p\land\I{va}=x} \effR{} \interty
  \\
  & p{:}\mathsf{Node}\garr ac{:}\mathsf{Node}\garr x{:}\mathsf{tVal} \sarr \effL{} \SNegStarU{\msg{ePrepReq}{\I{acceptor}=ac}}{\msg{eAcReq}{\true}} \effR{}\,
  \\ & \quad \urt{\Code{unit}}{\I{proposer}=p\land\I{acceptor}=ac\land\I{va}=x}\,
  \\ & \quad \effL{} \msg{ePrepRsp}{\I{acceptor}=ac\land\I{promised}=p\land\I{va}=x} \effR{} \interty
  \\
  & p{:}\mathsf{Node}\garr ac{:}\mathsf{Node}\garr x{:}\mathsf{tVal} \sarr \effL{} \allA\seqA\msg{eLostPrepRsp}{\I{acceptor}=ac\land\I{promised}=p}\seqA\allA \effR{}\,
  \\ & \quad \urt{\Code{unit}}{\I{proposer}=p\land\I{acceptor}=ac\land\I{va}=x}\,\effL{} \allA \effR{}
  \\
  \eff{eLostPrepRsp} ~:~{}
  & ac{:}\mathsf{Node}\garr p{:}\mathsf{Node} \sarr \effL{} \allA \effR{}\,\urt{\Code{unit}}{\I{acceptor}=ac\land\I{promised}=p}\,\effL{} \allA \effR{}
  \\
  \eff{ePrepRsp} ~:~{}
  & ac{:}\mathsf{Node}\garr p{:}\mathsf{Node}\garr x{:}\mathsf{tVal}\garr ap{:}\mathsf{Node} \sarr 
  \\ & \quad \effL{} \SNegStar{\msg{ePrepRsp}{\I{promised}=p}} \effR{}\,
  \\ & \quad \urt{\Code{unit}}{\I{acceptor}=ac\land\I{promised}=p\land\I{va}=x\land\mathit{nAccepted}=ap}\,
  \\ & \quad \effL{} \msg{eAcReq}{\I{acceptor}=ac\land\I{proposer}=p\land\I{va}=x} \effR{} \interty
  \\
  & ac{:}\mathsf{Node}\garr p{:}\mathsf{Node}\garr x{:}\mathsf{tVal}\garr ap{:}\mathsf{Node} \sarr 
  \\ & \quad \effL{} \SNegStar{\msg{ePrepRsp}{\I{promised}=p}} \effR{}\,
  \\ & \quad \urt{\Code{unit}}{\I{acceptor}=ac\land\I{promised}=p\land\I{va}=x\land\mathit{nAccepted}=ap}\,{}
  \\ & \quad \effL{} \msg{eAcReq}{\mathit{acc1}(\I{acceptor})\land\I{proposer}=p\land\I{va}=x}\seqA
  \\ & \quad \msg{eAcReq}{\mathit{acc2}(\I{acceptor})\land\I{proposer}=p\land\I{va}=x} \effR{} \interty
  \\
  & ac{:}\mathsf{Node}\garr p{:}\mathsf{Node}\garr x{:}\mathsf{tVal}\garr ap{:}\mathsf{Node} \sarr 
  \\ & \quad \effL{} \allA\seqA\msg{eLostAcReq}{\I{acceptor}=ac\land\I{proposer}=p}\seqA\allA \effR{}\,
  \\ & \quad \urt{\Code{unit}}{\I{acceptor}=ac\land\I{promised}=p\land\I{va}=x\land\mathit{nAccepted}=ap}\,\effL{} \allA \effR{} \interty
  \\
  & ac{:}\mathsf{Node}\garr p{:}\mathsf{Node}\garr x{:}\mathsf{tVal}\garr ap{:}\mathsf{Node} \sarr 
  \\ & \quad \effL{} \SNegStar{\msg{ePrepRsp}{\I{promised}=p}} \effR{}\,
  \\ & \quad \urt{\Code{unit}}{\I{acceptor}=ac\land\I{promised}=p\land\I{va}=x\land\mathit{nAccepted}=ap}\,
  \\ & \quad \effL{} \msg{eAcReq}{\I{acceptor}=ac\land\I{proposer}=p\land\I{va}=x} \effR{}
\end{align*}}
\vspace{-.4cm}\par\nobreak\noindent

\vspace{-.35cm}\par\nobreak{\small
\begin{align*}
\eff{eLostAcReq} ~:~{}
  & p{:}\mathsf{Node}\garr ac{:}\mathsf{Node} \sarr \effL{} \allA \effR{}\,\urt{\Code{unit}}{\I{proposer}=p\land\I{acceptor}=ac}\,\effL{} \allA \effR{}
  \\
  \eff{eAcReq} ~:~{}
  & p{:}\urt{\mathsf{Node}}{\mathit{pr2}(p)}\garr ac{:}\mathsf{Node}\garr x{:}\mathsf{tVal} \sarr {}
  \\ & \quad \effL{} \SNegStar{\msg{eAcReq}{\I{acceptor}=ac}}\seqA\msg{eAcReq}{\mathit{pr1}(\I{proposer})\land\I{acceptor}=ac}\seqA
  \\ & \quad \SNegStar{\msg{eAcReq}{\I{acceptor}=ac}} \effR{}\,
  \\ & \quad \urt{\Code{unit}}{\I{proposer}=p\land\I{acceptor}=ac\land\I{va}=x}\,
  \\ & \quad \effL{} \msg{eAcRsp}{\I{proposer}=p\land\I{acceptor}=ac\land\I{accepted}=p\land\I{va}=x} \effR{} \interty
  \\
  & p{:}\mathsf{Node}\garr ac{:}\mathsf{Node}\garr x{:}\mathsf{tVal} \sarr 
  \\ & \quad \effL{} \SNegStar{\msg{eAcReq}{\I{acceptor}=ac}} \effR{}\,\urt{\Code{unit}}{\I{proposer}=p\land\I{acceptor}=ac\land\I{va}=x}\,
  \\ & \quad \effL{} \msg{eAcRsp}{\I{proposer}=p\land\I{acceptor}=ac\land\I{accepted}=p\land\I{va}=x} \effR{} \interty
  \\
  & p{:}\mathsf{Node}\garr ac{:}\mathsf{Node}\garr x{:}\mathsf{tVal} \sarr 
  \\ & \quad \effL{} \allA\seqA\msg{eLostAcRsp}{\I{proposer}=p\land\I{acceptor}=ac}\seqA\allA \effR{}\,
  \\ & \quad \urt{\Code{unit}}{\I{proposer}=p\land\I{acceptor}=ac\land\I{va}=x}\,\effL{} \allA \effR{}
  \\
  \eff{eLostAcRsp} ~:~{}
  & p{:}\mathsf{Node}\garr ac{:}\mathsf{Node} \sarr \effL{} \allA \effR{}\,\urt{\Code{unit}}{\I{proposer}=p\land\I{acceptor}=ac}\,\effL{} \allA \effR{}
  \\
  \eff{eAcRsp} ~:~{}
  & p{:}\mathsf{Node}\garr ac{:}\mathsf{Node}\garr ap{:}\urt{\mathsf{Node}}{ap=p}\garr x{:}\mathsf{tVal} \sarr {}
  \\ & \quad \effL{} \SNegStar{\msg{ePrepRsp}{\true}}\seqA\msg{eAcRsp}{\I{proposer}=p\land\I{accepted}=ap}\seqA\allA \effR{}\,
  \\ & \quad \urt{\Code{unit}}{\I{proposer}=p\land\I{acceptor}=ac\land\I{accepted}=ap\land\I{va}=x}\,\effL{} \msg{eLearn}{\I{va}=x} \effR{} \interty
  \\
  & p{:}\mathsf{Node}\garr ac{:}\mathsf{Node}\garr ap{:}\urt{\mathsf{Node}}{ap=p}\garr x{:}\mathsf{tVal} \sarr 
  \\ & \quad \effL{} \SNegStar{\msg{eAcRsp}{\I{proposer}=p\land\I{accepted}=ap}} \effR{}\,
  \\ & \quad \urt{\Code{unit}}{\I{proposer}=p\land\I{acceptor}=ac\land\I{accepted}=ap\land\I{va}=x}\,\effL{} \msg{eLearn}{\I{va}=x} \effR{}
\end{align*}}
\vspace{-.4cm}\par\nobreak\noindent
The weakened \tyName{} merges intersected prepare/accept types and turns many payload tests into \texttt{true}.

\newpage
\subsubsection*{Raft}
This two-node Raft sketch pressures leader election under faults. Node \(n_1\) appends \(x\) and becomes leader, then \(n_2\) becomes leader while the suffix skips the usual vote traffic (\(\eff{eVoteReq}\), \(\eff{eVoteRsp}\), \(\eff{eBeLeader}\)), exposing brittle election dynamics. The global property is:
\vspace{-.35cm}\par\nobreak{\small
\begin{align*}
  & \allA \seqA \msg{eAppEntry}{\I{node}=n_1 \land \I{va}=x} \seqA \allA \seqA
    \msg{eBeLeader}{\mathit{ld}=n_1 \land n_1 \neq n_2} \seqA \allA \seqA
  \\& \quad \msg{eBeLeader}{\mathit{ld}=n_2} \seqA
    \SNegStarUU{\msg{eBeLeader}{\true}}{\msg{eVoteReq}{\true}}{\msg{eVoteRsp}{\true}}
\end{align*}}
\vspace{-.4cm}\par\nobreak\noindent
The \tyName{} comprises the following clauses:
\vspace{-.35cm}\par\nobreak{\small
\begin{align*}
  \eff{eStart} ~:~{}
  & \effL{} \SNegStar{\msg{eShutDown}{\true}} \effR{}\,\Code{unit}\,\effL{} \allA \effR{}
  \\
  \eff{eClientPutRsp} ~:~{}
  & x{:}\mathsf{tVal}\garr st{:}\Bool \sarr \effL{} \allA \effR{}\,\urt{\Code{unit}}{\I{va}=x\land\mathit{stat}=st}\,\effL{} \allA \effR{}
  \\
  \eff{eClientPut} ~:~{}
  & x{:}\mathsf{tVal} \sarr \effL{} \allA \effR{}\,\Code{unit}\,\effL{} \msg{eAppEntry}{\mathit{nodeId}(\I{node})=1\land\I{va}=x}\seqA
  \\ & \quad \msg{eAppEntry}{\mathit{nodeId}(\I{node})=2\land\I{va}=x} \effR{}
  \\
  \eff{eAppEntry} ~:~{}
  & n{:}\mathsf{Node}\garr x{:}\mathsf{tVal} \sarr \effL{} \allA \effR{}\,\urt{\Code{unit}}{\I{node}=n\land\I{va}=x}\,\effL{} \allA \effR{}
  \\
  \eff{eShutDown} ~:~{}
  & \effL{} \allA \effR{}\,\urt{\Code{unit}}{\true}\,\effL{} \msg{eTimeout}{\mathit{nodeId}(\I{dest})=1} \effR{}
  \\
  & \effL{} \allA \effR{}\,\urt{\Code{unit}}{\true}\,\effL{} \msg{eTimeout}{\mathit{nodeId}(\I{dest})=2} \effR{}
  \\
  \eff{eTimeout} ~:~{}
  & d{:}\mathsf{Node} \sarr \effL{} \allA \effR{}\,\urt{\Code{unit}}{\I{dest}=d}\,\effL{} \msg{eVoteReq}{\I{src}=d\land\I{leader}=d\land\I{dest}=\mathit{next}(d)} \effR{}
  \\
  \eff{eVoteReq} ~:~{}
  & s{:}\mathsf{Node}\garr d{:}\mathsf{Node}\garr ld{:}\urt{\mathsf{Node}}{ld=d} \sarr 
  \\ & \quad \effL{} \allA \effR{}\,\urt{\Code{unit}}{\I{src}=s\land\I{dest}=d\land\I{leader}=ld}\,\effL{} \msg{eVoteRsp}{\I{src}=d\land\I{dest}=s\land\neg\mathit{stat}} \effR{} \interty
  \\
  & s{:}\mathsf{Node}\garr d{:}\mathsf{Node}\garr ld{:}\urt{\mathsf{Node}}{ld\neq d} \sarr 
  \\ & \quad \effL{} \allA \effR{}\,\urt{\Code{unit}}{\I{src}=s\land\I{dest}=d\land\I{leader}=ld}\,
  \\ & \quad \effL{} \msg{eVoteRsp}{\I{src}=d\land\I{dest}=s\land\mathit{stat}=(\mathit{nodeId}(d)<\mathit{nodeId}(ld))} \effR{}
  \\
  \eff{eVoteRsp} ~:~{}
  & s{:}\mathsf{Node}\garr d{:}\urt{\mathsf{Node}}{d\neq s}\garr st{:}\Bool \sarr 
  \\ & \quad \effL{} \allA \effR{}\,\urt{\Code{unit}}{\I{src}=s\land\I{dest}=d\land\mathit{stat}=st}\,\effL{} \msg{eBeLeader}{\mathit{ld}=d} \effR{}
  \\
  \eff{eBeLeader} ~:~{}
  & ld{:}\mathsf{Node} \sarr \effL{} \allA \effR{}\,\urt{\Code{unit}}{\mathit{ld}=ld}\,\effL{} \allA \effR{}
\end{align*}}
\vspace{-.4cm}\par\nobreak\noindent
The weakened \tyName{} relaxes most payloads to \texttt{true} while keeping vote-control branching.

\newpage
\subsubsection*{Anonreadatomicity} The setting of this benchmark is already explained in \autoref{sec:impl}. The global property is:
\vspace{-.35cm}\par\nobreak{\small
\begin{align*}
  & \allA \seqA
    \msg{eUpdateRsp}{\I{gid}=\I{id} \land \I{key}=k \land \I{vl}=v_1 \land v_1 \neq v_2 \land \textit{ok}(\I{status})} \seqA
  \\&\quad \SNegStar{\msg{eUpdateRsp}{\I{gid}=\I{id} \land \I{key}=k \land \textit{ok}(\I{status})}} \seqA
  \\&\quad \msg{eReadRsp}{\I{gid}=\I{id} \land \I{key}=k \land \I{vl}=v_2 \land v_2 \neq v_1 \land \textit{ok}(\I{status})} \seqA \allA
\end{align*}}
\vspace{-.4cm}\par\nobreak\noindent
The \tyName{} is defined as:
\vspace{-.35cm}\par\nobreak{\small
\begin{align*}
  \eff{eStartTxnReq} ~:~{}
  & id{:}\mathsf{tGid} \sarr \effL{} \allA \effR{}\,\Code{unit}\,\effL{} \msg{eStartTxnRsp}{\I{gid}=id} \effR{}
  \\
  \eff{eStartTxnRsp} ~:~{}
  & id{:}\mathsf{tGid} \sarr \effL{} \SNegStar{\msg{eStartTxnRsp}{\I{gid}=id}} \effR{}\,\urt{\Code{unit}}{\I{gid}=id}\,
  \\ & \quad \effL{} \SNegStar{\msg{eStartTxnRsp}{\I{gid}=id}} \effR{}
  \\
  \eff{eReadReq} ~:~{}
  & id{:}\mathsf{tGid}\garr k{:}\mathsf{Key} \sarr \effL{} \allA\seqA\msg{eStartTxnRsp}{\I{gid}=id}\seqA\allA \effR{}\,\Code{unit}\,
  \\ & \quad \effL{} \msg{eShardReadKeyReq}{\I{gid}=id\land\I{key}=k} \effR{}
  \\
  \eff{eReadRsp} ~:~{}
  & id{:}\mathsf{tGid}\garr k{:}\mathsf{Key}\garr va{:}\mathsf{tVal}\garr st{:}\mathsf{tCmdStatus} \sarr 
  \\ & \quad \effL{} \allA \effR{}\,\urt{\Code{unit}}{\I{gid}=id\land\I{key}=k\land\I{value}=va\land\I{status}=st}\,\effL{} \allA \effR{}
  \\
  \eff{eUpdateReq} ~:~{}
  & id{:}\mathsf{tGid}\garr k{:}\mathsf{Key}\garr va{:}\mathsf{tVal} \sarr 
  \\ & \quad \effL{} \allA\seqA\msg{eStartTxnRsp}{\I{gid}=id}\seqA\allA \effR{}\,\Code{unit}\,
  \\ & \quad \effL{} \msg{eShardUpdateKeyReq}{\I{gid}=id\land\I{key}=k\land\I{value}=va} \effR{}
  \\
  \eff{eUpdateRsp} ~:~{}
  & id{:}\mathsf{tGid}\garr k{:}\mathsf{Key}\garr va{:}\mathsf{tVal}\garr st{:}\mathsf{tCmdStatus} \sarr 
  \\ & \quad \effL{} \allA \effR{}\,\urt{\Code{unit}}{\I{gid}=id\land\I{key}=k\land\I{value}=va\land\I{status}=st}\,\effL{} \allA \effR{}
  \\
  \eff{eCommitTxnReq} ~:~{}
  & id{:}\mathsf{tGid} \sarr \effL{} \allA\seqA\msg{eStartTxnRsp}{\I{gid}=id}\seqA\allA \effR{}\,\Code{unit}\,\effL{} \msg{eShardPrepReq}{\I{gid}=id} \effR{}
  \\
  \eff{eCommitTxnRsp} ~:~{}
  & id{:}\mathsf{tGid}\garr txnst{:}\mathsf{tTxnStatus} \sarr \effL{} \allA \effR{}\,\urt{\Code{unit}}{\I{gid}=id\land\I{txnstatus}=txnst}\,\effL{} \allA \effR{}
  \\
  \eff{eRollbackTxnReq} ~:~{}
  & id{:}\mathsf{tGid} \sarr \effL{} \allA\seqA\msg{eStartTxnRsp}{\I{gid}=id}\seqA\allA \effR{}\,\Code{unit}\
  \\ & \quad \effL{} \msg{eShardAbortTxn}{\I{gid}=id}\seqA\msg{eCommitTxnRsp}{\I{gid}=id\land\mathit{aborted}(\I{txnstatus})} \effR{}
  \\
  \eff{eShardReadKeyReq} ~:~{}
  & va{:}\mathsf{tVal}\garr id{:}\mathsf{tGid}\garr k{:}\mathsf{Key} \sarr {}
  \\ & \quad \effL{} \allA\seqA\msg{eShardUpdateKeyReq}{\neg(\I{gid}=id)\land\I{key}=k\land\I{value}=va}\seqA\allA\seqA
  \\ & \quad \msg{eStartTxnRsp}{\I{gid}=id}\seqA\allA \effR{}\,\urt{\Code{unit}}{\I{gid}=id\land\I{key}=k}\,
  \\ & \quad \effL{} \msg{eShardReadKeyRsp}{\I{gid}=id\land\I{key}=k\land\I{value}=va\land\mathit{ok}(\I{status})} \effR{} \interty
  \\
  & va{:}\mathsf{tVal}\garr id{:}\mathsf{tGid}\garr k{:}\mathsf{Key} \sarr {}
  \\ & \quad \effL{} \allA\seqA\msg{eStartTxnRsp}{\I{gid}=id}\seqA\allA\seqA
  \\ & \quad \msg{eShardUpdateKeyReq}{\I{gid}=id\land\I{key}=k\land\I{value}=va}\seqA\allA \effR{}\,
  \\ & \quad \urt{\Code{unit}}{\I{gid}=id\land\I{key}=k}\,
  \\ & \quad \effL{} \msg{eShardReadKeyRsp}{\I{gid}=id\land\I{key}=k\land\I{value}=va\land\mathit{ok}(\I{status})} \effR{} \interty
  \\
  & id{:}\mathsf{tGid}\garr k{:}\mathsf{Key} \sarr \effL{} \allA\seqA\msg{eShardAbortTxn}{\I{gid}=id}\seqA\allA \effR{}\,
  \\ & \quad \urt{\Code{unit}}{\I{gid}=id\land\I{key}=k}\,
  \\ & \quad \effL{} \msg{eShardReadKeyRsp}{\I{gid}=id\land\I{key}=k\land\mathit{abort}(\I{status})} \effR{} 
  \\
  \eff{eShardReadKeyRsp} ~:~{}
  & id{:}\mathsf{tGid}\garr k{:}\mathsf{Key}\garr va{:}\mathsf{tVal}\garr st{:}\mathsf{tCmdStatus} \sarr 
  \\ & \quad \effL{} \allA \effR{}\,\urt{\Code{unit}}{\I{gid}=id\land\I{key}=k\land\I{value}=va\land\I{status}=st}\,
  \\ & \quad \effL{} \msg{eReadRsp}{\I{gid}=id\land\I{key}=k\land\I{value}=va\land\I{status}=st} \effR{}
\end{align*}}
\vspace{-.4cm}\par\nobreak\noindent

\vspace{-.35cm}\par\nobreak{\small
\begin{align*}
  \eff{eShardUpdateKeyReq} ~:~{}
  & id{:}\mathsf{tGid}\garr k{:}\mathsf{Key}\garr va{:}\mathsf{tVal} \sarr 
  \\ & \quad \effL{} \allA\seqA\msg{eStartTxnRsp}{\I{gid}=id}\seqA\allA \effR{}\,\urt{\Code{unit}}{\I{gid}=id\land\I{key}=k\land\I{value}=va}\,
  \\ & \quad \effL{} \msg{eShardUpdateKeyRsp}{\I{gid}=id\land\I{key}=k\land\I{value}=va\land\mathit{ok}(\I{status})} \effR{}
  \\
  & id{:}\mathsf{tGid}\garr k{:}\mathsf{Key}\garr va{:}\mathsf{tVal} \sarr 
  \\ & \quad \effL{} \allA\seqA\msg{eShardAbortTxn}{\I{gid}=id}\seqA\allA \effR{}\,\urt{\Code{unit}}{\I{gid}=id\land\I{key}=k\land\I{value}=va}\,
  \\ & \quad \effL{} \msg{eShardUpdateKeyRsp}{\I{gid}=id\land\I{key}=k\land\I{value}=va\land\mathit{abort}(\I{status})} \effR{} \interty
  \\
  \eff{eShardUpdateKeyRsp} ~:~{}
  & id{:}\mathsf{tGid}\garr k{:}\mathsf{Key}\garr va{:}\mathsf{tVal}\garr st{:}\mathsf{tCmdStatus} \sarr 
  \\ & \quad \effL{} \allA \effR{}\,\urt{\Code{unit}}{\I{gid}=id\land\I{key}=k\land\I{value}=va\land\I{status}=st}\,
  \\ & \quad \effL{} \msg{eUpdateRsp}{\I{gid}=id\land\I{key}=k\land\I{value}=va\land\I{status}=st} \effR{}
  \\
  \eff{eShardCommitTxn} ~:~{}
  & id{:}\mathsf{tGid} \sarr \effL{} \allA \effR{}\,\urt{\Code{unit}}{\I{gid}=id}\,\effL{} \allA \effR{}
  \\
  \eff{eShardAbortTxn} ~:~{}
  & id{:}\mathsf{tGid} \sarr \effL{} \allA \effR{}\,\urt{\Code{unit}}{\I{gid}=id}\,\effL{} \allA \effR{}
  \\
  \eff{eShardPrepReq} ~:~{}
  & id{:}\mathsf{tGid} \sarr {}
  \\ & \quad \effL{} \SNegStar{\msg{eShardPrepReq}{\I{gid}=id}}\seqA\msg{eShardAbortTxn}{\I{gid}=id}\seqA
  \\ & \quad \SNegStar{\msg{eShardPrepReq}{\I{gid}=id}} \effR{}\,\urt{\Code{unit}}{\text{as }\msg{eShardUpdateKeyReq}{\I{gid}=id}}\, 
  \\ & \quad \effL{} \msg{eShardPrepRsp}{\I{gid}=id\land\neg\mathit{bstatus}} \effR{}
  \\
  & id{:}\mathsf{tGid} \sarr {}
  \\ & \quad \effL{} \SNegStar{\msg{eShardPrepReq}{\I{gid}=id}}\seqA\msg{eStartTxnRsp}{\I{gid}=id}\seqA
  \\ & \quad \SNegStarU{\msg{eShardAbortTxn}{\I{gid}=id}}{\msg{eShardPrepReq}{\I{gid}=id}} \effR{}\,\urt{\Code{unit}}{\I{gid}=id}\, 
  \\ & \quad \effL{} \msg{eShardPrepRsp}{\I{gid}=id\land\mathit{bstatus}} \effR{}
  \\
  \eff{eShardPrepRsp} ~:~{}
  & id{:}\mathsf{tGid} \sarr {}
  \\ & \quad \effL{} \allA \effR{}\,\urt{\Code{unit}}{\I{gid}=id\land\neg\mathit{bstatus}}\, 
  \\ & \quad {\scriptsize \effL{} \msg{eCommitTxnRsp}{\I{gid}=id\allowbreak\land\mathit{aborted}(\I{txnstatus})}\seqA\allowbreak\msg{eShardAbortTxn}{\I{gid}=id} \effR{}}
  \\
  & id{:}\mathsf{tGid} \sarr {}
  \\ & \quad \effL{} \allA \effR{}\,\urt{\Code{unit}}{\I{gid}=id\land\mathit{bstatus}}\, 
  \\ & \quad {\scriptsize \effL{} \msg{eCommitTxnRsp}{\I{gid}=id\allowbreak\land\mathit{committed}(\I{txnstatus})}\seqA\allowbreak\msg{eShardCommitTxn}{\I{gid}=id} \effR{}}
\end{align*}}
\vspace{-.4cm}\par\nobreak\noindent
The weakened \tyName{} turns many response/shard guards into \texttt{true} and merges intersected types of read/update shard requests.

\newpage
\section{Evaluation Details}\label{sec:tech:evaluation}

\autoref{tab:evaluation-detail} and \autoref{tab:evaluation-detail2}
provide additional details about the experimental results in our
evaluation section, including the number of local variables (\#var),
the number of effectful operations (\#$\mykw{eff}$), the number of
assertions (\#$\mykw{assert}$), and the number of executions for both
random generator and manual written test generators. 
We also present the results of all benchmarks in both OCaml and P Language with weakened specifications in \autoref{tab:evaluation-rich} and \autoref{tab:evaluation-rich-2}. Three of these benchmarks (\textsf{IFCLoad}$_{\mathsf{W}}$, \textsf{DeBruijnHO}$_{\mathsf{W}}$, and \textsf{Transaction}$_{\mathsf{W}}$) are already shown in \autoref{sec:impl}.

\begin{table}[t!]
  \renewcommand{\arraystretch}{0.8}
  % \vspace*{-.1in}
  \caption{\small Detailed results of table \ref{tab:evaluation}.}
\vspace*{-.1in}
\footnotesize
\setlength{\tabcolsep}{3.3pt}
\begin{tabular}{c||ccc||ccc||ccc||c|cc}
  \toprule
Benchmark & \#$\eff{op}$ & \multicolumn{2}{c||}{\#qualifier} & \#var &\#$\mykw{eff}$ & \#$\mykw{assert}$ & \multicolumn{3}{c||}{\scriptsize \# Num. Executions} & t$_\text{total}$(s) & \#SMT & \#refine  \\
 &  &uHAT &goal & & & & {\scriptsize Clouseau} & {\scriptsize Random} & {\scriptsize Manual} &  &  \\
 \midrule
 \textsf{Stack}\cite{ZYDJ24} & $6$ & $12$ & $3$ & $12$ & $15$ & $4$ & $1.0$ & $31.9$ & - & $0.60$ & $98$ & $4$\\
 \midrule
 \textsf{Set}\cite{ZYDJ24} & $7$ & $22$ & $3$ & $7$ & $12$ & $4$ & $1.0$ & $22.1$ & - & $1.28$ & $255$ & $12$\\
 \midrule
 \textsf{Filesystem}\cite{ZYDJ24} & $7$ & $67$ & $4$ & $11$ & $13$ & $8$ & $1.3$ & $2812.1$ & - & $6.02$ & $853$ & $14$\\
 \midrule
 \textsf{Graph}\cite{ZYDJ24} & $6$ & $24$ & $4$ & $5$ & $16$ & $5$ & $1.0$ & {\tiny TO} & - & $16.61$ & $2114$ & $16$\\
 \midrule
 \textsf{NFA}\cite{ZYDJ24} & $8$ & $50$ & $4$ & $7$ & $26$ & $6$ & $1.0$ & {\tiny TO} & - & $21.85$ & $5080$ & $6$\\
 \midrule
 \textsf{IFCAdd}\cite{pbt-ifc} & $8$ & $37$ & $2$ & $10$ & $8$ & $6$ & $1.6$ & {\tiny TO} & $3383.4$ & $0.96$ & $154$ & $8$\\
 \midrule
 \textsf{IFCStore}\cite{pbt-ifc} & $8$ & $37$ & $2$ & $10$ & $8$ & $6$ & $2.8$ & {\tiny TO} & $5049.1$ & $1.02$ & $160$ & $8$\\
 \midrule
 \textsf{IFCLoad}\cite{pbt-ifc} & $8$ & $37$ & $2$ & $20$ & $16$ & $12$ & $15.4$ & {\tiny TO} & $11293.7$ & $5.76$ & $930$ & $18$\\
 \midrule
 \textsf{DeBruijnFO}\cite{CoverageType} & $10$ & $117$ & $4$ & $10$ & $18$ & $12$ & $1.5$ & {\tiny TO} & $634.0$  & $40.38$ & $4765$ & $26$\\
 \midrule
 \textsf{DeBruijnHO}\cite{CoverageType} & $10$ & $118$ & $4$ & $10$ & $18$ & $14$ & $1.4$ & {\tiny TO} & {\tiny TO} & $91.73$ & $9034$ & $31$\\
 \midrule
 \textsf{HashTable}\cite{OcamlMulticorePBT} & $13$ & $38$ & $4$ & $6$ & $6$ & $2$ & $1.0$ & {\tiny TO} & $64.8$ & $0.76$ & $57$ & $6$\\
 \midrule
 \textsf{Transaction} & $8$ & $45$ & $3$ & $25$ & $15$ & $12$ & $1.2$ & {\tiny TO} & - & $2.89$ & $423$ & $16$\\
 \midrule
 \textsf{Shopping}\cite{MonkeyDB} & $10$ & $60$ & $3$ & $31$ & $17$ & $11$ & $1.0$ & $20.0$ & {\tiny TO} & $22.10$ & $1269$ & $30$\\
 \midrule
 \textsf{Courseware}\cite{MonkeyDB} & $16$ & $106$ & $3$ & $34$ & $18$ & $12$ & $1.0$ & $57.5$ & {\tiny TO} & $27.71$ & $1479$ & $33$\\
 \midrule
 \textsf{Twitter}\cite{MonkeyDB} & $16$ & $99$ & $3$ & $46$ & $24$ & $16$ & $1.0$ & $6.3$ & {\tiny TO} & $61.96$ & $2339$ & $24$\\
 \midrule
 \textsf{Smallbank}\cite{MonkeyDB} & $22$ & $162$ & $3$ & $69$ & $28$ & $22$ & $1.9$ & {\tiny TO} & - & $163.55$ & $10263$ & $29$\\
 \bottomrule
 \end{tabular}
\label{tab:evaluation-detail}
\end{table}

\begin{table}[t!]
  \renewcommand{\arraystretch}{0.8}
  % \vspace*{-.1in}
  \caption{\small Detail results of table \ref{tab:evaluation2}.}
\vspace*{-.1in}
\footnotesize
\setlength{\tabcolsep}{3.3pt}
\begin{tabular}{c||ccc||ccc||ccc||c|cc}
  \toprule
Benchmark & \#$\eff{op}$ & \multicolumn{2}{c||}{\#qualifier} & \#var &\#$\mykw{eff}$ & \#$\mykw{assert}$ & \multicolumn{3}{c||}{\scriptsize \# Num. Executions} & t$_\text{total}$(s) & \#SMT & \#refine \\
 &  &uHAT &goal & & & & {\scriptsize Clouseau} & {\scriptsize Random} & {\scriptsize Manual} &  &  \\
 \midrule
 {\scriptsize\textsf{Database}} & $4$ & $15$ & $3$ & $6$ & $3$$3$ & $4$ & $1.0$ & $5.6$ & - & $0.32$ & $49$ & $6$\\
 \midrule
{\scriptsize\textsf{Firewall}}\cite{MessageChain} & $5$ & $26$ & $4$ & $12$ & $2$$8$ & $9$ & $1.0$ & $10.0$ & - & $0.59$ & $222$ & $10$\\
\midrule
{\scriptsize\textsf{RingLeaderElection}}\cite{MessageChain} & $3$ & $14$ & $2$ & $12$ & $2$$6$ & $7$ & $1.0$ & $17.7$ & - & $0.83$ & $100$ & $8$\\
\midrule
{\scriptsize\textsf{BankServer}}\cite{DGJ+13} & $6$ & $42$ & $1$ & $15$ & $2$$3$ & $5$ & $1.0$ & $40.5$ & $2.2$ & $0.17$ & $27$ & $5$\\
\midrule
{\scriptsize\textsf{Simplified2PC}}\cite{DGJ+13} & $9$ & $32$ & $5$ & $7$ & $2$$6$ & $5$ & $2.1$ & $125.0$ & $5.8$ & $1.03$ & $88$ & $8$\\
\midrule
{\scriptsize\textsf{HeartBeat}}\cite{DGJ+13} & $7$ & $31$ & $3$ & $9$ & $4$$10$ & $9$ & $1.0$ & $70.9$ & $6.0$ & $1.10$ & $145$ & $14$\\
\midrule
{\scriptsize\textsf{ChainReplication}}\cite{ModP} & $7$ & $72$ & $4$ & $24$ & $4$$8$ & $9$ & $1.0$ & $588.2$ & $500.0$ & $19.24$ & $2054$ & $169$\\
\midrule
{\scriptsize\textsf{Paxos}}\cite{ModP} & $10$ & $110$ & $2$ & $36$ & $4$$10$ & $13$ & $1.0$ & {\tiny TO} & $666.7$ & $23.70$ & $1763$ & $77$\\
\midrule
{\scriptsize\textsf{Raft}} & $9$ & $39$ & $6$ & $21$ & $2$$10$ & $10$ & $1.0$ & {\tiny TO} & - & $26.84$ & $1262$ & $78$\\
\midrule
{\scriptsize\textsf{AnonReadAtomicity}} & $17$ & $113$ & $3$ & $37$ & $5$$11$ & $12$ & $1.0$ & $53.0$ & - & $18.35$ & $1909$ & $16$\\
\bottomrule
\end{tabular}
\label{tab:evaluation-detail2}
\end{table}

\newcolumntype{P}[1]{>{\centering\arraybackslash}p{#1}}

\begin{table}[t!]
  \renewcommand{\arraystretch}{0.8}
  % \vspace*{-.1in}
  \caption{\small The results of the benchmarks in OCaml with weakened specifications. }
  \vspace*{-.1in} \footnotesize
%\multicolumn{2}{c||}{} &
\setlength{\tabcolsep}{1.3pt}
\begin{tabular}{c|P{4.8cm}||ccc||cc||c|ccc}
  \toprule
  \multicolumn{2}{c||}{Benchmark} & \#$\eff{op}$ & \multicolumn{2}{c||}{\#qualifier} & \multicolumn{2}{c||}{\scriptsize \# Num. Executions} & t$_\text{total}$(s) & \#evt & \#refine & \#SMT \\
  \midrule
Name & Property &   &uHAT &goal & {\scriptsize Clouseau} & {\scriptsize Baseline} &  &  \\
\midrule
\textsf{Stack}$_{\mathsf{W}}$\cite{ZYDJ24} & {\scriptsize Pushes and pops are correctly paired.} & $6$ & $11$ & $3$ & $1.2$ & $30.4$ & $0.92$ & $15$ & $4$ & $149$\\
\midrule
\textsf{Set}$_{\mathsf{W}}$\cite{ZYDJ24} & {\scriptsize Membership holds for every element inserted into the set.} & $7$ & $20$ & $3$ & $1.1$ & $30.1$ & $1.41$ & $12$ & $12$ & $265$\\
\midrule
\textsf{Filesystem}$_{\mathsf{W}}$\cite{ZYDJ24} & {\scriptsize A valid file path only contains non-deleted entries.} & $7$ & $47$ & $4$ & $4.3$ & $4103.0$ & $3.84$ & $10$ & $12$ & $469$\\
\midrule
\midrule
\textsf{Graph}$_{\mathsf{W}}$\cite{ZYDJ24} & {\scriptsize A serialized stream of nodes and edges is re-constituted to
  form a fully-connected graph.} & $6$ & $17$ & $4$ & $1.0$ & {\tiny Timeout} & $1.67$ & $11$ & $11$ & $233$\\
\midrule
\textsf{NFA}$_{\mathsf{W}}$\cite{ZYDJ24} & {\scriptsize An NFA reaches a final state for every string in the language it accepts.} & $8$ & $33$ & $4$ & $1.0$ & {\tiny Timeout} & $7.93$ & $24$ & $5$ & $1908$\\
\midrule
\textsf{IFCStore}$_{\mathsf{W}}$\cite{pbt-ifc} & {\scriptsize A well-behaved IFC program containing a $\Code{Store}$ command never leaks a secret.} & $8$ & $33$ & $2$ & $6.0$ & $6082.8^{\dagger}$ & $0.89$ & $8$ & $8$ & $136$\\
\midrule
\textsf{IFCAdd}$_{\mathsf{W}}$\cite{pbt-ifc} & {\scriptsize A well-behaved IFC program containing an $\Code{Add}$ command never leaks a secret.} & $8$ & $33$ & $2$ & $3.4$ & $6621.5^{\dagger}$ & $0.99$ & $8$ & $8$ & $136$\\
\midrule
\textsf{IFCLoad}$_{\mathsf{W}}$\cite{pbt-ifc} & {\scriptsize A well-behaved IFC program containing a $\Code{Load}$ command never leaks a secret.} & $8$ & $33$ & $2$ & $115.4$ & $4278.3^{\dagger}$ & $5.01$ & $16$ & $18$ & $754$\\
\midrule
\textsf{DeBruijnFO}$_{\mathsf{W}}$\cite{CoverageType} & {\scriptsize An STLC interpreter correctly evaluates a
  well-typed first-order STLC program that uses a de Bruijn representation.} & $10$ & $106$ & $4$ & {\tiny Timeout} & $303.6^{\dagger}$ & $3.03$ & $7$ & $22$ & $857$\\
\midrule
\textsf{DeBruijnHO}$_{\mathsf{W}}$\cite{CoverageType} & {\scriptsize An STLC interpreter correctly evaluates a
  a well-typed higher-order STLC program that uses a de Bruijn representation.} & $10$ & $107$ & $4$ & {\tiny Timeout} & {\tiny Timeout}$^{\dagger}$ & $3.32$ & $7$ & $22$ & $957$\\
\midrule
\midrule
\textsf{Shopping}$_{\mathsf{W}}$\cite{MonkeyDB} & {\scriptsize All items added to a cart can be checked-out.} & $10$ & $58$ & $3$ & $1.0$ & $20.0^{\dagger}$ & $18.57$ & $17$ & $30$ & $1174$\\
\midrule
\textsf{HashTable}$_{\mathsf{W}}$\cite{OcamlMulticorePBT} & {\scriptsize No updates to a concurrent hashtable are ever lost.} & $13$ & $33$ & $4$ & $1.0$ & $2.5^{\dagger}$ & $0.84$ & $6$ & $6$ & $56$\\
\midrule
\textsf{Transaction}$_{\mathsf{W}}$ & {\scriptsize Asynchronous read operations are logically atomic.} & $8$ & $40$ & $3$ & {\tiny Timeout} & {\tiny Timeout} & $2.61$ & $15$ & $15$ & $372$\\
\midrule
\textsf{Courseware}$_{\mathsf{W}}$\cite{MonkeyDB} & {\scriptsize Every student enrolled in a course exists in the
  enrollment database for that course.} & $16$ & $100$ & $3$ & $1.0$ & $57.5^{\dagger}$ & $25.54$ & $18$ & $33$ & $1374$\\
\midrule
\textsf{Twitter}$_{\mathsf{W}}$\cite{MonkeyDB} & {\scriptsize Posted tweets are visible to all followers.} & $16$ & $95$ & $3$ & $1.0$ & $6.3^{\dagger}$ & $27.28$ & $18$ & $33$ & $1374$\\
\midrule
\textsf{Smallbank}$_{\mathsf{W}}$\cite{OLTPBench} & {\scriptsize Account updates are strongly consistent.} & $22$ & $154$ & $3$ & {\tiny Timeout} & {\tiny Timeout} & $5.74$ & $14$ & $14$ & $528$\\
\bottomrule
\end{tabular}
\label{tab:evaluation-rich}
\vspace*{-.15in}
\end{table}

\begin{table}[t!]
  \renewcommand{\arraystretch}{0.8}
  % \vspace*{-.1in}
  \caption{\small The results of the benchmarks in P Language with weakened specifications. }
  \vspace*{-.1in} \footnotesize
%\multicolumn{2}{c||}{} &
\setlength{\tabcolsep}{1.3pt}
\begin{tabular}{c|P{4.0cm}||ccc||cc||c|ccc}
  \toprule
  \multicolumn{2}{c||}{Benchmark} & \#$\eff{op}$ & \multicolumn{2}{c||}{\#qualifier} & \multicolumn{2}{c||}{\scriptsize \# Num. Executions} & t$_\text{total}$(s) & \#evt & \#refine & \#SMT \\
  \midrule
Name & Property &   &uHAT &goal & {\scriptsize Clouseau} & {\scriptsize Baseline} &  &  \\
\midrule
{\scriptsize\textsf{Database}$_{\mathsf{W}}$} & {\scriptsize The database maintains a Read-Your-Writes policy.} & $4$ & $13$ & $3$ & $1.0$ & {\tiny TO} & $0.34$ & $6$ & $6$ & $51$\\
\midrule
{\scriptsize\textsf{Firewall}$_{\mathsf{W}}$}\cite{MessageChain} & {\scriptsize Requests generated inside the firewall are eventually propagated to the outside.} & $5$ & $20$ & $4$ & {\tiny TO} & {\tiny TO} & $0.61$ & $6$ & $6$ & $189$\\
\midrule
{\scriptsize\textsf{RingLeaderElection}$_{\mathsf{W}}$}\cite{MessageChain} & {\scriptsize There is always a single unique  leader.} & $3$ & $13$ & $2$ & $2.0$ & {\tiny TO} & $0.82$ & $8$ & $8$ & $102$\\
\midrule
{\scriptsize\textsf{BankServer}$_{\mathsf{W}}$}\cite{DGJ+13} & {\scriptsize Withdrawals in excess of the available balance are never allowed.} & $6$ & $36$ & $1$ & $1.6$ & {\tiny TO}$^{\dagger}$ & $0.17$ & $5$ & $5$ & $27$\\
\midrule
{\scriptsize\textsf{Simplified2PC}$_{\mathsf{W}}$}\cite{DGJ+13} & {\scriptsize Transactions are atomic.} & $9$ & $24$ & $5$ & $2.9$ & {\tiny TO}$^{\dagger}$ & $1.05$ & $8$ & $8$ & $86$\\
\midrule
{\scriptsize\textsf{HeartBeat}$_{\mathsf{W}}$}\cite{DGJ+13} & {\scriptsize All available nodes are identified by a  detector.} & $7$ & $24$ & $3$ & {\tiny TO} & {\tiny TO}$^{\dagger}$ & $0.36$ & $6$ & $6$ & $41$\\
\midrule
{\scriptsize\textsf{ChainReplication}$_{\mathsf{W}}$}\cite{ModP} & {\scriptsize Concurrent updates are never lost.} & $7$ & $47$ & $4$ & $5.2$ & {\tiny TO}$^{\dagger}$ & $4.22$ & $11$ & $11$ & $332$\\
\midrule
{\scriptsize\textsf{Paxos}$_{\mathsf{W}}$}\cite{ModP} & {\scriptsize Logs are correctly replicated.} & $10$ & $85$ & $2$ & {\tiny TO} & {\tiny TO}$^{\dagger}$ & $21.83$ & $12$ & $61$ & $2046$\\
\midrule
{\scriptsize\textsf{Raft}$_{\mathsf{W}}$} & {\scriptsize Leader election is robust to faults.} & $9$ & $34$ & $6$ & $1.0$ & {\tiny TO} & $27.60$ & $12$ & $12$ & $989$\\
\midrule
\textsf{AnonReadAtomicity}$_{\mathsf{W}}$ & {\scriptsize Read Atomicity is preserved.} & $17$ & $98$ & $3$ & $1.0$ & $53.3^{\dagger}$ & $10.68$ & $12$ & $12$ & $991$\\
\bottomrule
\end{tabular}
\label{tab:evaluation-rich-2}
\vspace*{-.15in}
\end{table}

\else
\fi

\end{document}